\documentclass{article} 

\usepackage[a4paper, total={6in, 8in}]{geometry}
\usepackage{multicol}
\usepackage{latexsym}
\usepackage{amssymb,amsbsy,amsmath,amsfonts,amssymb,amscd,mathtools,amsfonts,mathrsfs,bm}
\usepackage[shortlabels]{enumitem}
\usepackage[applemac]{inputenc}
\usepackage{graphicx}
\usepackage{cite}
\usepackage{hyperref}
\usepackage{dsfont}
\usepackage{subcaption}
\usepackage{array}% http://ctan.org/pkg/array
\usepackage{pdfsync}
\usepackage{color} 
\usepackage{bbm} 
\def\rev#1{\textcolor{black}{#1}} % REVISION blue
\definecolor{lightgray}{gray}{0.5}

\usepackage{colortbl}
\definecolor{LightCyan}{rgb}{0.88,1,1}
\newcolumntype{s}{>{\columncolor{LightCyan}}c}
\usepackage[makeroom]{cancel}

\newcommand{\commentout}[1]{}

\newcommand {\Chi} {{\bf \raise 2pt \hbox{$\chi$}} }

\newcommand{\beq}{\begin{equation}}
\newcommand{\eeq}{\end{equation}}
\newcommand{\beqa}{\begin{eqnarray}}
\newcommand{\eeqa}{\end{eqnarray}}
\newcommand{\bea} {\begin{array}{rl}}
\newcommand{\eea} {\end{array}}
\newcommand{\bepa}{\left\{ \begin{array}{l}}
\newcommand{\eepa} {\end{array}\right.}

\newcommand{\qed}{{ \hfill
                       {\unskip\kern 6pt\penalty 500 \raise -2pt\hbox{\vrule\vbox to 6pt{\hrule width 6pt
                       \vfill\hrule}\vrule} \par}   \bigskip }}
\numberwithin{equation}{section}
\title{Evolutionary dynamics of glucose-deprived cancer cells: insights from experimentally-informed mathematical modelling}

\author{
Luis Almeida
\thanks{Sorbonne Universit{\'e}, CNRS, Universit{\'e} de Paris, Inria, Laboratoire Jacques-Louis Lions UMR 7598, 75005 Paris, France}
\and
J{\'e}r{\^o}me Alexandre Denis 
\thanks{Sorbonne Universit{\'e}, Cancer Biology and Therapeutics, INSERM, CNRS, Institut Universitaire de Canc{\'e}rologie, Saint-Antoine Research Center (CRSA), F-75012, Paris, France}
\thanks{Department of Endocrinology and Oncology Biochemistry, Piti{\'e}-Salpetri{\`e}re Hospital, 75013 Paris, France}
\and
Nathalie Ferrand
\footnotemark[2]
\and
Tommaso Lorenzi
\thanks{Department of Mathematical Sciences ``G. L. Lagrange'', Dipartimento di Eccellenza 2018-2022, Politecnico di Torino, 10129 Torino, Italy}
\and
\rev{Antonin Prunet}
\footnotemark[1]\,\,\footnotemark[2]
\and
Mich{\'e}le Sabbah
\footnotemark[2]
\and
Chiara Villa
\footnotemark[1]
}

\begin{document}
\maketitle
\pagestyle{plain}
%\tableofcontents
\pagenumbering{arabic}

%%%\begin{abstract}
%%%Several lines of evidence support the idea that lactate is continuously produced by glycolytic cells and that transporters such as MCT1 facilitate lactate uptake by cancer cells, allowing them to reuse this byproduct as an energy source, promoting further disease progression. 
%%%We employ a synergistic mathematical modelling and experimental approach to investigate the respective role of lactate-associated signalling pathways and non-genetic instability in facilitating the increase in MCT1 expression levels and survival of glucose-deprived breast cancer cells \textit{in vitro}. 
%%%Experimental observations highlight that short-term survival under glucose deprivation, associated with lactate uptake and an increase in MCT1 expression, is characteristic of aggressive tumours. 
%%%Data from the aggressive cell line are then used to calibrate the mathematical model, which comprises a phenotype-structured equation for the evolutionary dynamics of cancer cells under time-varying glucose and lactate concentrations, and the cell phenotypic distributions provided by the model are compared with MCT1 expression distributions obtained from flow cytometry experiments. 
%%%The analytical and numerical results of the calibrated model indicate that lactate-associated signalling pathways are responsible for fast phenotypic changes, while non-genetic instability plays a key role in increasing phenotypic diversity, speeding up natural selection and providing a bet-hedging strategy in the long-term. 
%%%\end{abstract}

\begin{abstract}
Glucose is a primary energy source for cancer cells. Several lines of evidence support the idea that monocarboxylate transporters, such as MCT1, elicit metabolic reprogramming of cancer cells in glucose-poor environments, allowing them to reuse lactate, a byproduct of glucose metabolism, as an alternative energy source with serious consequences for disease progression. We employ a synergistic experimental and mathematical modelling approach to explore the evolutionary processes at the root of cancer cell adaptation to glucose deprivation, with particular focus on the mechanisms underlying the increase in MCT1 expression observed in glucose-deprived aggressive cancer cells. Data from \textit{in vitro} experiments on breast cancer cells are used to inform and calibrate a mathematical model that comprises a partial integro-differential equation for the dynamics of a population of cancer cells structured by the level of MCT1 expression. Analytical and numerical results of this model \rev{suggest} that environment-induced changes in MCT1 expression mediated by lactate-associated signalling pathways enable a prompt adaptive response of glucose-deprived cancer cells, whilst \rev{fluctuations in MCT1 expression due to epigenetic changes} create the substrate for environmental selection to act upon, speeding up the selective sweep underlying cancer cell adaptation to glucose deprivation, and may constitute a long-term bet-hedging mechanism.

\end{abstract}

\section{Introduction}
Glucose is one of the primary nutrients used by cancer cells to produce energy, and glucose deficiency causes metabolic stress, cell dysfunction, and eventual death~\cite{romero2011tumor}. In fact, glucose consumption not only results in decreased nutrient availability, but also generally correlates with lactate production and the development of an acidic extracellular environment~\cite{doherty2013targeting,garcia2019tumor,keenan2015alternative,wang2015exploring}. Cancer cells can rely on a variety of mechanisms that activate protective functions under metabolic and environmental stress, including metabolic reprogramming~\cite{keenan2015alternative,hanahan2011hallmarks}. Accumulating evidence indicates that aggressive cancer cells may acquire the ability to absorb lactic acid and use it to synthesise pyruvate~\cite{garcia2019tumor,ippolito2019lactate,keenan2015alternative,zhao2020cancer}, thus converting harmful byproducts of glucose metabolism into alternative energy sources. Lactic acid is transported across cell membranes through a family of four reversible monocarboxylate transporters (MCTs) belonging to the SLC16/MCT family of solute carriers~\cite{halestrap2013monocarboxylic}. 
\rev{The intracellular uptake of lactate by cancer cells is primarily achieved via the MCT1 transporter protein on the cytoplasmic membrane~\cite{khan2020targeting,tasdogan2020metabolic}.} 
%Amongst these, MCT1 is the most widely expressed and facilitates lactate and pyruvate upload~\cite{park2018overview,sonveaux2008targeting}.
Such an increase in pyruvate and lactate metabolism has been associated with enhanced invasion and migration, and higher survival in the circulation, with overall consequences for metastasis~\cite{liu2021metabolic,tasdogan2020metabolic}, to the extent that MCT1 inhibition has been investigated as a potential therapeutic target~\cite{corbet2018interruption,wang2021lactate}. 

While it is evident that the overexpression of MCT1 plays a key role in the metabolic reprogramming of glucose-deprived aggressive cancer cells, allowing them to reuse lactate as an alternative energy source with serious consequences for disease progression, the mechanisms underlying such a change in MCT1 expression remain, to this day, poorly explored. On the one hand, lactate may function as a signalling molecule, triggering regulatory pathways that \rev{promote the expression of} MCT1~\cite{ippolito2019lactate,longhitano2022lactate}, thus mediating environment-induced changes in MCT1 expression~\cite{huang2013genetic}. 
On the other hand, \rev{lactic acid accumulation is known to cause reshaping of the tumour microenvironment and induce epigenetic changes in cancer cells~\cite{wang2023lactate}. } 
%On the other hand, cancer cells are known to undergo spontaneous heritable and reversible changes in gene expression due to non-genetic instability~\cite{brock2009non,nicolson1987tumor}. This is generally associated with noise in gene expression and may be due, for instance, to DNA methylation or histone modifications~\cite{hansen2011increased,sandoval2012cancer} occurring over the lifespan of a tumour cell. 
As lactate has been shown to be responsible for certain histone modifications, \rev{such as acetylation or lactylation}~\cite{bergers2021metabolism,zhang2019metabolic}, \rev{the effect of non-genetic mechanisms may even be perturbed under glucose-deprivation. Furthermore, MCT1 inhibition has been shown to prevent histone acetylation, suggesting that lactate intake induces transcriptional and metabolic reprogramming that involves epigenetic modifications~\cite{ippolito2022lactate}.
%As lactate has been shown to be responsible for certain histone modifications~\cite{bergers2021metabolism,zhang2019metabolic}, the effect of non-genetic instability may even be enhanced under glucose-deprivation. 
Therefore, protein expression levels, which vary in a systematic transcriptionally-regulated way, may undergo fluctuations due to spontaneous and lactate-induced epigenetic changes~\cite{huang2013genetic}.} 

In this work, a synergistic experimental and mathematical modelling approach is employed to explore the evolutionary processes at the root of cancer cell adaptation to glucose deprivation, with particular focus on the mechanisms underlying the increase in MCT1 expression observed in glucose-deprived aggressive cancer cells. Data from \textit{in vitro} experiments on breast cancer cells, which were specifically carried out for this study, are used to inform a mathematical model that comprises a partial integro-differential equation (PIDE) for the dynamics of a population of cancer cells structured by the level of MCT1 expression, which is coupled with a system of ordinary differential equations (ODEs) for the dynamics of glucose and lactate present in the extracellular environment. This model allows for predictions on the dynamics of the MCT1 expression distribution of cancer cells to be made and to be directly compared with the results of flow cytometry analyses, while making it also possible to dissect out the evolutionary processes underlying these dynamics. Related mathematical models have been employed to investigate cancer cell adaptation to hypoxia~\cite{ardavseva2019dissecting,celora2022dna,cho2017modeling,fiandaca2021mathematical,villa2021evolutionary}, but not to assess MCT1-associated changes in lactate uptake. Furthermore, alternative mathematical models have been proposed to study the role of MCT1 expression-regulated lactate uptake in the coexistence of different metabolic pathways within the same tumour~\cite{mcgillen2014glucose,mendoza2012mathematical,wang2015exploring}, but none of these proposes a causal mechanism for the reported increase in MCT1 expression. 

Experimental data are used to carry out model calibration, through a likelihood-maximising method~\cite{linden2022bayesian,martinez2014bayesopt,myung2003tutorial,spilker2001evaluation} \rev{and bootstrapping algorithm for uncertainty quantification~\cite{thompson1996tutorial,zhu1997making}}, and the results of numerical simulations of the calibrated model are complemented with analytical results on the qualitative and quantitative properties of the solution to the PIDE that governs the evolution of cancer cells. These results shed light on the evolutionary dynamics of glucose-deprived cancer cells by elucidating the respective roles that environment-induced changes in MCT1 expression mediated by lactate-associated signalling pathways and \rev{fluctuations in MCT1 expression due to epigenetic changes} play in the adaptation of cancer cell populations to glucose-poor environments.

\section{Methods}
\subsection{\textit{In vitro} experiments}\label{sec:met:exp}
A summary of the experimental set-up is provided below, and full details of experimental materials and methods can be found in Sup.Mat.S2.  

\subsubsection*{Cell lines}
Breast cancer cells of the MCF7 and MCF7-sh-WISP2 lines are considered, with the latter being obtained from the former upon inducing epithelial-to-mesenchymal transition through WISP2 gene silencing. The focus of this study is on the MCF7-sh-WISP2 cell line, which has been documented to be more invasive and aggressive than the MCF7 cell line~\cite{ferrand2014loss,sabbah2011ccn5}. We also report on the results from experiments conducted on MCF7 cells to corroborate the hypothesis that fast metabolic reprogramming associated with lactate uptake under glucose deprivation is characteristic of more aggressive cancer cells. 

\subsubsection*{`Glucose-deprivation' and `rescue' experiments}
Cells originally maintained in medium containing \rev{24.75mM} of glucose (i.e. \rev{4.5g/l,} a high level of glucose) are seeded, at high cell numbers, in a medium containing \rev{5.5mM} of glucose (i.e. \rev{1g/l,} a physiological level of glucose). In `glucose-deprivation' experiments, glucose is not re-added to the medium during cell culture (i.e. there is no glucose replenishment) so that, due to consumption by the cells, glucose levels drop during the course of the experiment and thus cells experience glucose deprivation. In `rescue' experiments, a similar protocol is followed for a few days and the culture medium is subsequently changed to a medium containing  \rev{24.75mM} of glucose, where cells are cultured for a few more days, so that cells first experience glucose deprivation and are then rescued from it. 
%without glucose replenishment

\subsubsection*{Measured quantities} 
Over the span of several days, we tracked: viable cell numbers and percentages of apoptotic cells; glucose and lactate concentrations in the cell culture medium; MCT1 expression distribution (i.e. fluorescence-intensity distributions), obtained through flow cytometry analysis and complemented with images from immunocytochemistry analysis; mRNA expression of different MCT proteins measured by RT-qPCR -- i.e. MCT1, MCT2 (an MCT very similar to MCT1 although it displays a higher affinity for L-lactic acid and pyruvate), and MCT4 (an efficient lactate exporter expressed in glycolytic cells that is not required for lactate uptake~\cite{park2018overview,sonveaux2008targeting}).  \rev{All data are available in the supplementary Excel file `ExperimentalData.xlsx'.}

\subsection{Mathematical modelling}
Building on the modelling strategies presented in~\cite{lorenzi2016tracking}, we develop a mathematical model that describes the evolutionary dynamics of a population of MCF7-sh-WISP2 cells, structured by the level of MCT1 expression, under the environmental conditions which are determined by the levels of glucose and lactate in the extracellular environment. An outline of the model is provided below, while a detailed description of the model equations alongside the main modelling assumptions, which are informed by the results of {\it in vitro} experiments underlying this study, is provided in Sup.Mat.S1. 
%The biological meaning of the model parameters is also summarised in Tab.S1 in Sup.Mat.S2. 

%\red{number density} 
\subsubsection*{Key model quantities}
The model comprises a PIDE for the dynamics of the cell population density function $n(t,y)$, which represents the number of MCF7-sh-WISP2 cells with level of MCT1 expression $y\in\mathbb{R}$ at time $t \in \mathbb{R}_+$ (i.e. the MCT1 expression distribution of MCF7-sh-WISP2 cells at time $t$). Such a PIDE is coupled with a system of ODEs for the dynamics of the concentrations of glucose and lactate in the extracellular environment $G(t)$ and $L(t)$. The cell number, the mean level of MCT1 expression and the related variance, which provides a possible measure for the level of intercellular variability in MCT1 expression, are computed, respectively, as
\beq
\label{erhomusigmay}
\rho(t) = \int_{\mathbb{R}} n(t,y) \; {\rm d}y, \quad \mu(t) = \frac{1}{\rho(t)} \int_{\mathbb{R}} y \, n(t,y) \; {\rm d}y, \quad \sigma^2(t)  = \frac{1}{\rho(t)} \int_{\mathbb{R}} y^2 \; n(t,y) \; {\rm d}y- \mu^2(t).
\eeq

\subsubsection*{Modelling cell proliferation and death under environmental selection on MCT1 expression}
%The results of {\it in vitro} experiments (cf. Sec.~\ref{sec:results2}) indicate that the MCT1 expression distribution of MCF7-sh-WISP2 cells is, to a first approximation, unimodal with a single peak at the centre of the distribution. The location of the centre of the distribution moves from lower to higher expression levels when cells experience glucose deprivation and from higher to lower expression levels when cells are rescued from glucose deprivation.
The results of {\it in vitro} experiments (cf. Sec.~\ref{sec:results2}) indicate that the \rev{mean of the MCT1 expression distribution of MCF7-sh-WISP2 cells moves from lower to higher expression levels when cells experience glucose deprivation, and from higher to lower expression levels when cells are rescued from glucose deprivation.} 
 Hence, we assume that there is a level of  MCT1 expression (i.e. the fittest level of MCT1 expression) endowing cells with the highest fitness depending on the environmental conditions determined by the concentrations of glucose and lactate. Moreover, the results of {\it in vitro} experiments (cf. Sec.~\ref{sec:results1}) support the idea that proliferation and survival of MCF7-sh-WISP2 cells correlate with glucose uptake when glucose levels are sufficiently high and with lactate uptake when glucose levels are low. \rev{This is also in line with the well-established notion of tumour metabolism indicating a preferential use of glucose for energy production when glucose is abundantly available~\cite{keenan2015alternative}.}
Therefore, we further assume that there are a level of MCT1 expression, $y_L$, endowing cells with the highest rate of proliferation via glycolysis and a higher level of MCT1 expression, $y_H>y_L$, endowing cells with the highest rate of proliferation via lactate reuse when glucose is scarce -- i.e. when the concentration of glucose in the extracellular environment is lower than a threshold level $G^*$ above which cells stop taking lactate from the extracellular environment in order to prioritise glucose uptake. 
Under these assumptions, in the framework of our model, the fittest level of MCT1 expression is represented by the function $Y(G,L)$, defined via Eq.~(S.14) in Sup.Mat.S1,  
{which is such that if $G \geq G^*$ then $Y(G,L)=y_L$ for any $L \geq 0$, whereas if $G < G^*$ then $Y(G,L) \to y_H$ as $G$ decreases and $L$ increases. }
%which is such that if $G \geq G^*$ then $Y(G,L)=y_L$ for any $L \geq 0$, whereas if $G < G^*$ then $Y(G,L) \to y_H$ as $G$ decreases and $L$ increases. 
Furthermore, the strength of environmental selection on MCT1 expression is linked to the value of the selection gradient $b(G,L)$ defined via Eq.~(S.12) in Sup.Mat.S1,

\subsubsection*{Modelling changes in MCT1 expression} %\rev{FECs and SPCs}  in MCT1 expression} 
The effects of changes in the level of cell expression of MCT1 are also incorporated into the model. In particular, we let \rev{fluctuations due to epigenetic changes -- henceforth referred to as FECs --} %spontaneous changes due to non-genetic instability
 occur at rate $\Phi$. Moreover, we assume that environment-induced changes mediated by lactate-associated signalling pathways  \rev{-- henceforth referred to as SPCs --} lead to an increase in MCT1 expression at rate $\Psi^+$ under glucose deprivation (i.e. when $G < G^*$) and to a decrease in MCT1 expression at rate $\Psi^-$ when the glucose level is sufficiently high (i.e. when $G \geq G^*$).

\subsection{Model calibration based on experimental data}
%Mathematical methodology
Experimental data on MCF7-sh-WISP2 cells are used to carry out model calibration through a likelihood-maximising method~\cite{linden2022bayesian,martinez2014bayesopt,myung2003tutorial}. 
\rev{We assume $G^*>5.5$mM (i.e. the threshold level of glucose above which cells interrupt lactate uptake to prioritise glucose uptake is above physiological levels) and calibrate the model using data from `glucose-deprivation' experiments.} % -- though discussion of model results under calibration with data from `glucose-deprivation' and `rescue' experiments is also discussed.}
In summary,  \rev{the optimal parameter set (OPS) is obtained, through an iterative process, by minimising the weighted sum of squared residuals $R({S}_P)$
 \begin{equation}\label{def:wssr}
R({S}_P)= \sum_{i=1}^{M} \frac{(\bar{u}^i_D-{u}^i_P)^2}{2{(s^i_D)}^2} \,,
\end{equation}
where $\bar{u}^i_D$ and $s^i_D$ are the average and standard deviation of summary statistic $i$ from the experimental data and ${u}^i_P$ is the value predicted by the model under the parameter set ${S}_P$, } exploiting the in-built {\sc Matlab} function \texttt{bayesopt}, which is based on Bayesian Optimisation.  
At each iteration, \rev{in order to retrieve ${u}^i_P$ for $i=1, \ldots, M$,} we solve numerically the PIDE-ODE system that constitutes the model, using the methods described in Sup.Mat.S2. 
\rev{Minimising~\eqref{def:wssr} is analogous to maximising the log-likelihood, having assumed that data for each summary statistic are normally distributed around their average values with variance $(s^i_D)^2$, to account for heteroscedasticity~\cite{spilker2001evaluation}.}
In order to explore a variety of evolutionary scenarios, calibration was carried out for the model in which both \rev{FECs and SPCs}  in MCT1 expression are included (i.e. $\Phi\not\equiv0$ and $\Psi^{\pm}\not\equiv0$), and for reduced models that take into account only \rev{FECs} (i.e. $\Phi\not\equiv0$ and $\Psi^{\pm}\equiv0$) or \rev{SPCs} (i.e. $\Phi\equiv0$ and $\Psi^{\pm}\not\equiv0$). The obtained OPSs are reported in Tab.S1 in Sup.Mat.S2.  \rev{Uncertainty quantification of the OPS was carried out on the full model by means of a bootstrapping algorithm~\cite{thompson1996tutorial,zhu1997making},  based on random sampling of data with replacement and particularly suited when only few data are available. The bootstrap statistics are reported in Tab.S2 in Sup.Mat.S2, and bootstrap sampling distributions of the parameter values are plotted in Fig.S6.}  
\rev{All} {\sc Matlab} source codes used for model calibration have been made available on GitHub\footnote{https://github.com/ChiaraVilla/AlmeidaEtAl2023Evolutionary}. 

\subsection{Simulation and analysis of the model}
To explore the mechanisms underlying the evolutionary dynamics of MCF7-sh-WISP2 cells under glucose deprivation, the results of numerical simulations of the calibrated model, %\rev{and within the parameter confidence intervals determined by undertainty quantification,} 
which are carried out using the numerical methods described in Sup.Mat.S2.3, are integrated with the analytical results presented in  Sup.Mat.S2.5, which build on the results presented in~\cite{chisholm2016evolutionary,villa2021evolutionary}. 
%We exploit the results of the calibrated mathematical model to deepen our understanding of the mechanisms responsible for the observed experimental outcomes. 

\section{Main results}\label{sec:results}
\subsection{Proliferation and survival of MCF7-sh-WISP2 cells correlate with lactate uptake under glucose deprivation}
\label{sec:results1}
Fig.\ref{Fig:data} and Fig.S1 in Sup.Mat.S3 summarise the dynamics of cell proliferation and glucose and lactate concentrations in the cell culture medium observed during `glucose-deprivation' experiments conducted for four days on MCF7-sh-WISP2 and MCF7 cells, respectively. The corresponding dynamics of cell death are summarised by Fig.S2 in Sup.Mat.S3. These results demonstrate that there is a stark difference in the proliferation dynamics of the MCF7-sh-WISP2 and MCF7 cell lines under glucose deprivation, with the former reaching numbers of viable cells over twice as high as the latter. Moreover, cell death in the MCF7-sh-WISP2 line does not significantly increase over time, as opposed to the MCF7 line for which the percentage of apoptotic cells undergoes a four-fold increase during the experiment. The dynamics of the concentration of glucose in the culture medium of the two cell lines are similar, though only cells of the MCF7-sh-WISP2 line consume all the glucose available. Furthermore, the concentration of lactate in the culture medium of MCF7 cells displays a steady increase mirroring glucose consumption, whilst a decline in lactate concentration in the culture medium of MCF7-sh-WISP2 cells is observed when little to no glucose is present in the medium, thus suggesting that lactate uptake occurs amongst MCF7-sh-WISP2 cells under glucose deprivation. Taken together, these experimental results support the idea that proliferation and survival of MCF7-sh-WISP2 cells correlate with glucose consumption when glucose levels are sufficiently high and with lactate uptake under glucose deprivation.

\begin{figure}[b!]
\centering
\includegraphics[width=0.94\linewidth]{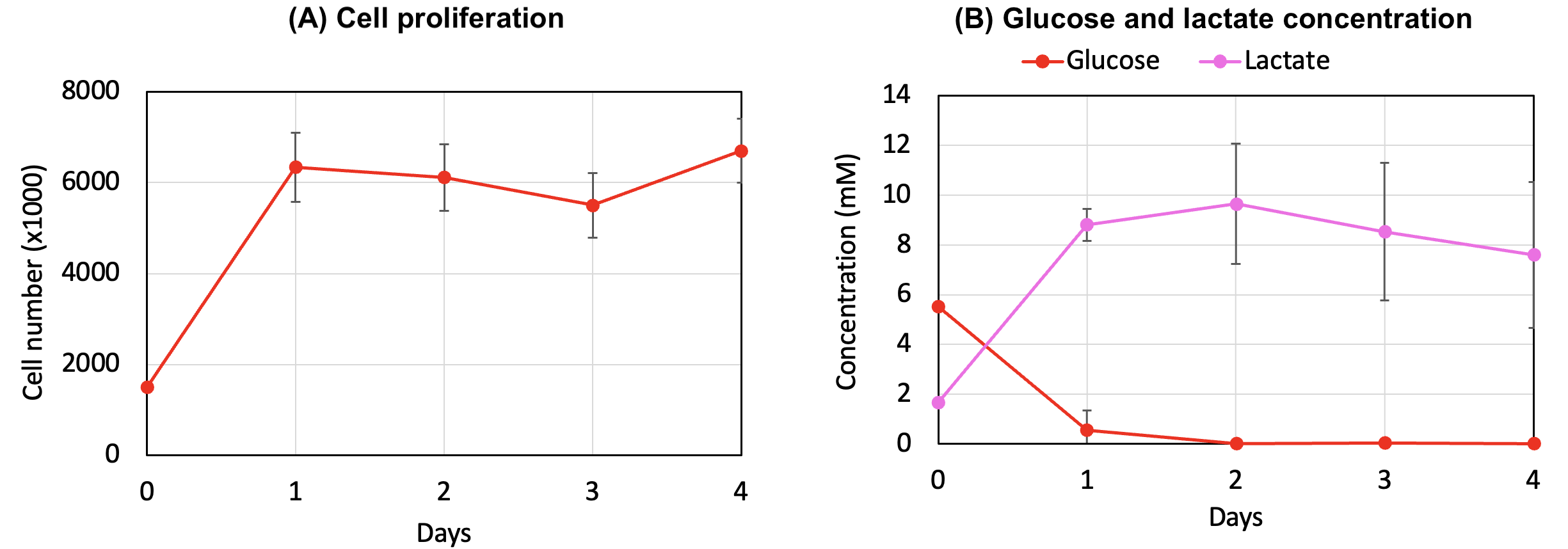}
\caption{\label{Fig:data}  \textbf{Dynamics of cell proliferation and glucose and lactate concentrations in `glucose-deprivation' experiments conducted on MCF7-sh-WISP2 cells.} Dynamics of  cell proliferation (panel {\bf (A)}), glucose concentration (panel {\bf (B)}, red line, left y-axis) and lactate concentration (panel {\bf (B)}, pink line, right y-axis) in `glucose-deprivation' experiments conducted on MCF7-sh-WISP2 cells for four days. Cell proliferation was assessed by counting the number of viable cells upon seeding (i.e. day 0) and at the end of each day of culture (i.e. days 1-4). Glucose and lactate concentrations were measured in the cell culture medium on days 0-4. \rev{These figures display the average (dots) and standard deviation (error bars) of two replicate experiments.} }
%A summary of the experimental set-up is provided in Sec.~\ref{sec:met:exp}, while full details of experimental materials and methods can be found in Sup.Mat.S2.1.
\end{figure}

\subsection{Glucose deprivation induces a reversible increase in MCT1 expression of MCF7-sh-WISP2 cells} 
\label{sec:results2}
The experimental results summarised by Fig.\ref{Fig:cytometry} show a steady increase in MCT1 expression of MCF7-sh-WISP2 cells throughout `glucose-deprivation' experiments. On the other hand, in `rescue' experiments, MCT1 expression levels of MCF7-sh-WISP2 cells increase during the glucose-deprivation phase of the experiment and then decrease again during the phase of rescue from glucose deprivation, which demonstrates reversibility of changes in MCT1 expression. 
Similar trends are observed in the MCT1 mRNA levels of MCF7-sh-WISP2 cells during `glucose-deprivation' and `rescue' experiments\rev{. Note that due to delayed transcription and translation, which have been reported by several authors and for various systems -- see for instance~\cite{gedeon2012delayed,greenbaum2003comparing} --, there is a delay between the surge in the mRNA level and the surge in the corresponding protein level. On the other hand,} an increase {resembling} the one detected in the MCT1 protein expression levels is not observed in the MCT2 mRNA levels, and no MCT4 mRNA is detected (cf. Fig.S4 in Sup.Mat.S3). 
%Similar trends are observed in the MCT1 mRNA levels of MCF7-sh-WISP2 cells during `glucose-deprivation'  and `rescue' experiments, whereas an increase {resembling} the one detected in the MCT1 protein expression levels is not observed in the MCT2 mRNA levels, and no MCT4 mRNA is detected (cf. Fig.S4 in Sup.Mat.S3). 
In contrast, \rev{our data give no indication of a significant change in MCT1 expression of MCF7 cells during both `glucose-deprivation' and `rescue' experiments} (cf. Fig.S3 in Sup.Mat.S3).
%{In contrast,} no appreciable change in MCT1 expression of MCF7 cells is observed during both `glucose-deprivation' and `rescue' experiments (cf. Fig.S3 in Sup.Mat.S3).

\begin{figure}[htb!]
\centering
\includegraphics[width=\linewidth]{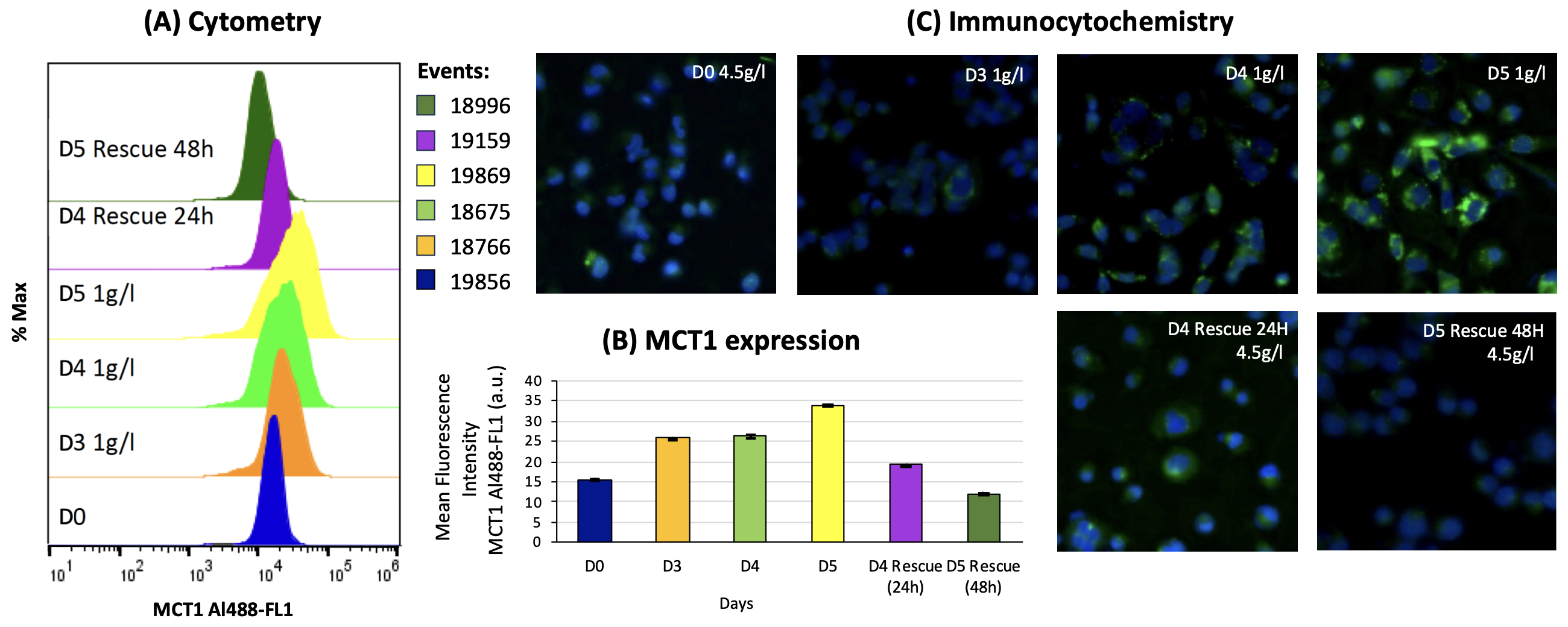}
\caption{\label{Fig:cytometry} \textbf{Dynamics of MCT1 expression in `glucose-deprivation' and `rescue' experiments conducted on MCF7-sh-WISP2 cells.} {\bf (A)},{\bf (C)} MCT1 protein expression of MCF7-sh-WISP2 cells, assessed through flow cytometry analysis (panel {\bf (A)}) and immunocytochemistry analysis using an MCT1 antibody (green staining in panel {\bf (C)}), upon seeding (i.e. on day 0) and on days 3-5 of `glucose-deprivation' experiments conducted for five days (sub-panel D0 and sub-panels D3-D5). MCT1 protein expression of MCF7-sh-WISP2 cells during the phase of rescue from glucose deprivation in the corresponding `rescue' experiments (i.e. on days 4 and 5) is also displayed (sub-panels D4 Rescue and D5 Rescue).  \rev{The `Events' legend indicates the number of events (i.e. the total number of cells analysed) for each distribution plotted in panel {\bf (A)} on a logarithmic scale.} {\bf (B)} Mean fluorescence intensity  of MCT1 labelling for MCF7-sh-WISP2 cells (in units of $10^3$), \rev{displaying average (coloured bars) and standard deviation (small black error bars) of two replicate experiments.}%of panel {\bf (A)} plotted on a logarithmic scale. 
}
\end{figure}

\subsection{Both \rev{FECs and SPCs} in MCT1-expression \rev{may} contribute to the adaptation of MCF7-sh-WISP2 cells to glucose deprivation\label{sec:res:math:1}} 
There is good quantitative agreement \rev{($R(S_P)\approx84$)} between numerical results obtained by simulating `glucose-deprivation' through the calibrated model in which both \rev{FECs and SPCs}  in MCT1 expression are included (i.e. $\Phi\not\equiv0$ and $\Psi^{\pm}\not\equiv0$) and experimental observations for MCF7-sh-WISP2 cells deprived of glucose, as shown by the plots in Fig.\ref{Fig:calib}. 
%and experimental observations for MCF7-sh-WISP2 cells deprived of glucose and rescued from glucose deprivation, as shown by the plots in Fig.\ref{Fig:calib}. %\red{, and as additionally highlighted by the high likelihood associated with the obtained OPS (cf. Tab.S1 in Sup.Mat.S2)}. 
Instead, levels of intercellular variability in MCT1 expression much higher or lower than those estimated from experimental data are observed in numerical simulations of the experiments carried out through calibrated reduced models that take into account only \rev{FECs or SPCs} in MCT1 expression (i.e. $\Phi\not\equiv0$ and $\Psi^{\pm}\equiv0$ or $\Phi\equiv0$ and $\Psi^{\pm}\not\equiv0$), respectively -- this is demonstrated by the dynamics of the variance of the MCT1 expression distribution, $\sigma^2$, displayed in Fig.S5(d) in Sup.Mat.S3. Furthermore, the results of numerically simulated `glucose-deprivation' experiments carried out over a time span longer than that of \textit{in vitro} experiments suggest that the synergy between these two forms of changes in MCT1 expression accelerates collective cell adaptation to glucose deprivation\rev{. This is} demonstrated by the fact that, when $\Phi\not\equiv0$ and $\Psi^{\pm}\not\equiv0$, the mean level of MCT1 expression, $\mu$, converges more quickly to the level $y_H$, which in our modelling framework is the level endowing MCF7-sh-WISP2 cells with the maximum capability of taking lactate from the extracellular environment and reusing it to produce the energy required for their proliferation when glucose is scarce (cf. Fig.S\rev{7} in Sup.Mat.S3).

Taken together, these results support the idea that both \rev{FECs and SPCs}  in MCT1 expression contribute to the adaptation of MCF7-sh-WISP2 cells to glucose deprivation. In particular, the  modelling assumptions underlying these numerical results provide the following theoretical explanation for the increase in the mean level of MCT1 expression experimentally observed amongst glucose-deprived MCF7-sh-WISP2 cells. Cells with different levels of MCT1 expression emerge as a consequence of \rev{fluctuations in MCT1 expression due to epigenetic changes}. On top of this, as the glucose concentration decreases and the lactate concentration increases during `glucose-deprivation' experiments, \rev{SPCs}  in protein expression mediated by lactate-associated signalling pathways lead cells to express MCT1 at a higher level. Cells with levels of MCT1 expression closer to the fittest one, which in glucose-poor environments is higher than in glucose-rich environments, are then dynamically selected. The interplay between these evolutionary processes results in a progressive increase in the mean level of MCT1 expression of MCF7-sh-WISP2 cells. 

\rev{Similar conclusions can be drawn by calibrating the model with data from both `glucose-deprivation' and `rescue' experiments -- cf. the results of numerical simulations displayed in Fig.S10 in Sup.Mat.S3 -- although the delay in the dynamics predicted by the model compared to experimental observations suggest additional evolutionary mechanisms are at play under high glucose levels, i.e. for $G\gg5.5$mM. %Nonetheless, it is clear that the model is capable of capturing the reversibility of MCT1 expression 
}

\begin{figure}[htb!]
\centering
\includegraphics[width=\linewidth]{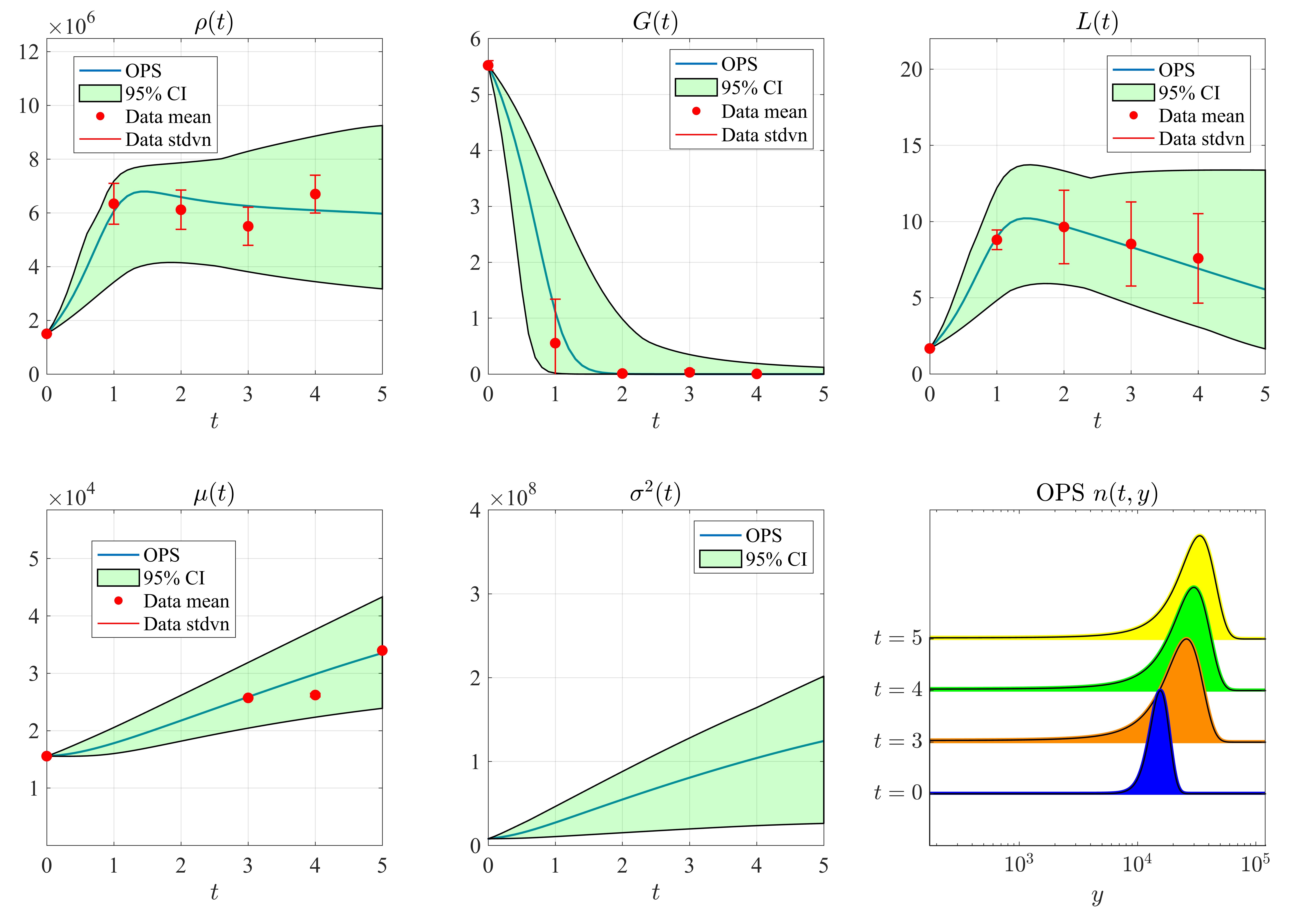}
\caption{\label{Fig:calib} \textbf{Numerical simulations of `glucose-deprivation' experiments conducted on MCF7-sh-WISP2 cells.}
Simulated dynamics of the cell number $\rho(t)$ (top-left panel), the glucose concentration $G(t)$ (top-central panel), the lactate concentration $L(t)$ (top-right panel), the mean level of MCT1 expression $\mu(t)$ (bottom-left panel, solid line), the related variance $\sigma^2(t)$ (bottom-central panel), and the MCT1 expression distribution $n(t,y)$ (bottom-right panel, $t=0$ - $t=5$) in `glucose-deprivation' experiments conducted on MCF7-sh-WISP2 cells. Numerical simulations were carried out for the calibrated model in which both \rev{FECs and SPCs}  in MCT1 expression are included (i.e. $\Phi\not\equiv0$ and $\Psi^{{\pm}}\not\equiv0$), under the OPS reported in Tab.S1 in Sup.Mat.S2 \rev{(blue lines), and under 200 parameter sets generated by random sampling from the empirical 95\% confidence interval (CI) of the bootstrap sampling distributions (green areas) -- see Fig.S6 in Sup.Mat.S3 and Tab.S2 in Sup.Mat.S2.}
The MCT1 expression distribution \rev{obtained under the OPS} is plotted on a logarithmic scale as for the outputs of flow cytometry analyses to facilitate visual comparison. 
%The MCT1 expression distribution during the phase of rescue from glucose deprivation in the corresponding simulations of `rescue' experiments is also displayed (bottom-right panel, $t=4 \, {\rm R}$ and $t=5 \, {\rm R}$) along with the mean level of MCT1 expression (bottom-left panel, dashed \rev{blue} line \rev{and light green area}). 
%The red markers highlight \rev{mean (scatter points) and standard deviation (error bars) of the} experimental data that are used to carry out model calibration, with circles and triangles corresponding to `glucose-deprivation' and `rescue' experiments, respectively. The values of $t$ are in days,  \rev{while the values of $G(t)$ and $L(t)$ are in mM}.}
The red markers highlight \rev{average (scatter points) and standard deviation (error bars) of the} experimental data that are used to carry out model calibration. The values of $t$ are in days,  \rev{while the values of $G(t)$ and $L(t)$ are in mM}.}
\end{figure}

\subsection{Respective contributions of \rev{FECs and SPCs} in MCT1-expression in the adaptation of MCF7-sh-WISP2 cells to glucose deprivation\label{sec:res:math:2}} 
The analytical results of Proposition~S2.1 in Sup.Mat.S2 (cf. Eqs.~(S28)$_1$ and~(S28)$_2$ along with the relations given by~Eq.(S4)) clarify how\rev{, under the assumptions on which our model is built,  FECs and SPCs} in MCT1 expression, along with environmental selection on MCT1 expression, affect the dynamics of the mean level of MCT1 expression, $\mu$, and the corresponding variance, $\sigma^2$, in MCF7-sh-WISP2 cells. In summary, larger values of the rate of \rev{FECs}  in MCT1 expression, $\Phi$, accelerate the growth of $\sigma^2$, while a stronger environmental selection on MCT1 expression (i.e. a larger selection gradient $b$) leads to reduced values of $\sigma^2$. In turn, larger values of $\sigma^2$ enhance the rate at which $\mu$ approaches the fittest level of MCT1 expression, $Y$. Such a rate also increase with the strength of environmental selection on MCT1 expression (i.e. the selection gradient $b$). Moreover, under glucose deprivation, larger values of the rate at which \rev{SPCs} lead to an increase in MCT1 expression, $\Psi^+$, promote the growth of $\mu$. These analytical results are confirmed by the results of numerical simulations of `glucose-deprivation' experiments presented in Fig.S\rev{8} in Sup.Mat.S3, which show that larger values of $\Phi$ and $\Psi^+$ correlate with a faster increase of $\sigma^2$ and $\mu$. 

Taken together, these results \rev{may} clarify the roles played by \rev{FECs and SPCs}  in MCT1 expression in the evolutionary dynamics of glucose-deprived MCF7-sh-WISP2 cells. \rev{Under our model's assumptions, t}he former promote intercellular variability in MCT1 expression, which creates the substrate for environmental selection to act upon and speed up the selective sweep underlying collective cell adaptation to glucose deprivation, while the latter triggers a prompt adaptive response of glucose-deprived MCF7-sh-WISP2 cells by promoting overexpression of MCT1. These conclusions are also supported by the fact that estimation of the model parameters from experimental data (cf. the OPS reported in Tab.S1 in Sup.Mat.S2) \rev{ and uncertainty quantification via bootstrapping (cf. the bootstrap statistics reported in Tab.S2  in Sup.Mat.S2 and the bootstrap sampling distributions plotted in Fig.S6 in Sup.Mat.S3)} indicate that the rate of \rev{SPCs} in the level of MCT1 expression is approximately three orders of magnitude larger than the rate of \rev{FECs}  \rev{-- see Sup.Mat.S2.4 for more details}.
% of the order of $10^{-1}$ per day (i.e $10^{-6}$s$^{-1}$) approximately, while  is much lower, being of the order of $10^{-4}$ per day (i.e. $10^{-9}$s$^{-1}$) .

\subsection{\rev{FECs} in MCT1-expression  may constitute a long-term bet-hedging mechanism for MCF7-sh-WISP2 cells under glucose deprivation} 
The mathematical model makes it possible to explore the cell evolutionary dynamics beyond timescales and scenarios which can be investigated through experiments. In particular, the analytical results of Theorem~S2.2 in Sup.Mat.S2 (cf. Remark S2.3 in Sup.Mat.S2) provide a complete characterisation of the equilibrium values of the number, the mean level of MCT1 expression and the related variance of MCF7-sh-WISP2 cells under virtual scenarios where the glucose and lactate concentrations are kept constant, i.e. a complete characterisation of the limits $\rho(t) \to \rho_{\infty}$, $\mu(t) \to \mu_{\infty}$ and $\sigma^2(t) \to \sigma^2_{\infty}$ as $t \to \infty$ when $G(t) = \overline{G}$ and $L(t) = \overline{L}$ for all $t \geq 0$. These analytical results are confirmed by the results of numerical simulations of the calibrated model summarised by Fig.S\rev{9} in Sup.Mat.S3, which demonstrate that the equilibrium values $\rho_{\infty}$, $\mu_{\infty}$ and $\sigma^2_{\infty}$ are ultimately attained when $(G(t),L(t)) \equiv (\overline{G},\overline{L})$. 

The results of Theorem~S2.2 complement the results discussed in Sec.~\ref{sec:res:math:1} by showing that the equilibrium value of the variance of the MCT1 expression distribution, $\sigma^2_{\infty}$, increases with the rate of \rev{FECs}  in MCT1 expression, $\Phi$, and decreases with the strength of environmental selection on MCT1 expression (i.e. the selection gradient $b$). The results of Theorem~S2.2 also demonstrate that when glucose is scarce and lactate is present (i.e. when $\overline{G} < G^*$ and $\overline{L}>0$), and thus under glucose deprivation, the distance between the equilibrium value of the mean level of MCT1 expression, $\mu_{\infty}$, and the fittest level of MCT1 expression, $Y$, increases with the rate at which \rev{SPCs}  lead to an increase in MCT1 expression, $\Psi^+$, and decreases with both the rate of \rev{FECs} , $\Phi$, and the strength of environmental selection on MCT1 expression (i.e. the selection gradient $b$). This supports the idea that, whilst enabling a faster adaptive response to glucose deprivation, as discussed in Sec.~\ref{sec:res:math:2}, \rev{SPCs}  in MCT1 expression may ultimately lead to suboptimal adaptation, whereas \rev{FECs}  may constitute a long-term bet-hedging mechanism. 

Moreover, the heat maps in Fig.~\ref{Fig:ss} illustrate how the equilibrium values $\rho_{\infty}$, $\mu_{\infty}$ and $\sigma^2_{\infty}$ vary with the glucose and lactate concentrations $\overline{G}$ and $\overline{L}$, under the OPS reported in Tab.S1 in Sup.Mat.S2.  
In summary, when glucose is scarce (i.e. for $\overline{G} < G^*$ with $G^*\approx \rev{5.5}$\rev{mM} in the obtained OPS \rev{and bootstrap sampling distribution mean reported in Tab.S3 in Sup.Mat.S2}): $\mu_{\infty}$ decreases with $\overline{G}$ and increases with $\overline{L}$; $\sigma^2_{\infty}$ increases as $\overline{G}$ decreases and reaches maximum levels when $\overline{L}$ is also small; $\rho_{\infty}$ increases with \rev{$\overline{G}$ and decreases with $\overline{L}$ when $\overline{G}$ is sufficiently large, while it increases with $\overline{L}$ when $\overline{G}$ is closer to zero. In particular, the values of $\rho_{\infty}$ obtained for $\overline{G}=0$mM and $\overline{L}=9.645$mM are only one order of magnitude smaller than those obtained for $\overline{G}=5.52$mM and $\overline{L}=0$mM (cf. Fig.S9, first panel in Sup.Mat.S3).} 
%$\rho_{\infty}$ increases with both $\overline{G}$ and $\overline{L}$. 
%On the other hand, when the glucose level is sufficiently high (i.e. for $\overline{G} \geq G^*$), $\mu_{\infty}$ drops to $y_L$ and both $\rho_{\infty}$ and $\sigma^2_{\infty}$ vary very little with $\overline{G}$ and $\overline{L}$ -- i.e. $\rho_{\infty}$ remains relatively high and $\sigma^2_{\infty}$ remains relatively low. 

%\red{These findings recapitulate the results of numerical simulations of `glucose-deprivation' and `rescue' experiments displayed in Fig.\ref{Fig:calib} by corroborating the idea that, whereas lower mean levels of MCT1 expression emerge when the concentration of glucose in the extracellular environment is sufficiently high, as in `rescue' experiments, glucose deprivation leads to the selection for cells that are capable of exploiting lactate as an alternative energy source, which results in higher mean levels of MCT1 expression amongst MCF7-sh-WISP2 cells and allows for relatively high cell numbers in spite of glucose scarcity.}
These findings recapitulate the results of numerical simulations of `glucose-deprivation' experiments displayed in Fig.\ref{Fig:calib} by corroborating the idea that, whereas lower mean levels of MCT1 expression emerge when the concentration of glucose in the extracellular environment is sufficiently high, as in `rescue' experiments \rev{(cf. Fig.S10 in Sup.Mat.S3)}, glucose deprivation leads to the selection for cells that are capable of exploiting lactate as an alternative energy source, which results in higher mean levels of MCT1 expression amongst MCF7-sh-WISP2 cells and allows for relatively high cell numbers in spite of glucose scarcity.

\begin{figure}[htb]
\centering
\includegraphics[width=\linewidth]{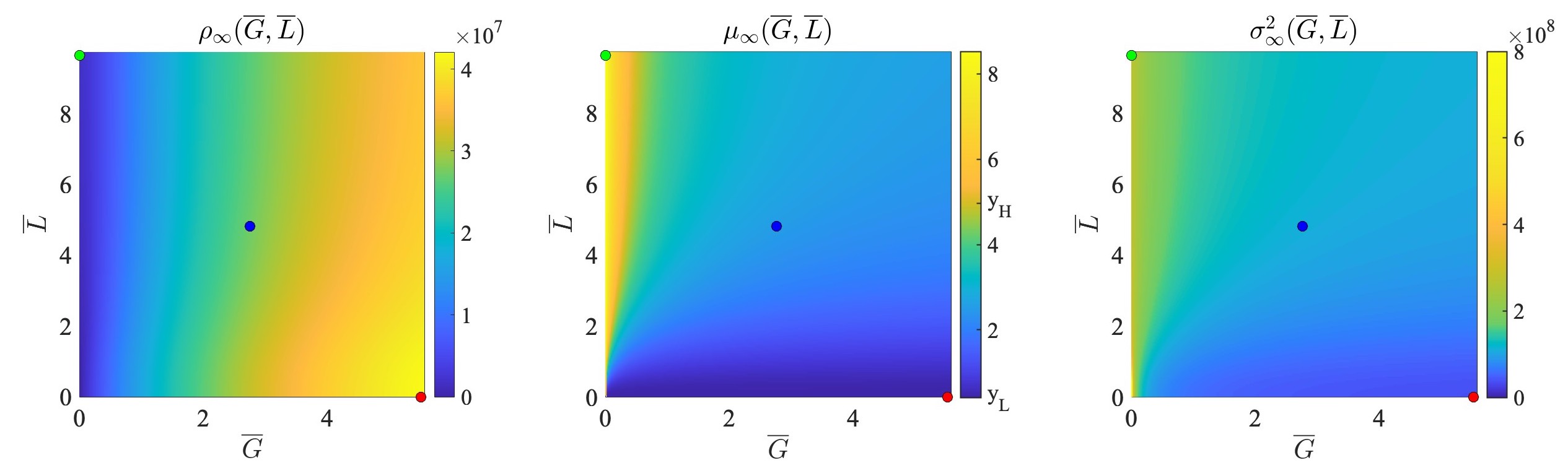}
\caption{\label{Fig:ss}\textbf{Equilibrium values of the number, the mean level of MCT1 expression and the related variance of MCF7-sh-WISP2 cells predicted by the mathematical model under constant concentrations of glucose and lactate.} Plots of the equilibrium number $\rho_\infty$ (left panel), the mean level of MCT1 expression $\mu_\infty$ (central panel), and the related variance $\sigma^2_\infty$ (right panel) of MCF7-sh-WISP2 cells given by Theorem~S2.2 in Sup.Mat.S2 (cf. Eq.(S51)) as functions of constant concentrations of glucose and lactate $(\overline{G}, \overline{L})$, under the OPS reported in Tab.S1 in Sup.Mat.S2. The green, blue and red dots highlight the values of $(G(t),L(t)) \equiv (\overline{G},\overline{L})$ that are used to obtain the numerical results of Fig.S\rev{9} in Sup.Mat.S3. The values of $\overline{G}$ \rev{and $\overline{L}$} are in \rev{mM}.
}
\end{figure}

\section{Discussion}\label{sec:discussion}
%\subsection{Key results}
We adopted an experimentally-informed mathematical modelling approach to investigate the evolutionary dynamics of glucose-deprived cancer cells.

\textit{In vitro} experiments were conducted on two breast cancer cell lines, MCF7 and MCF7-sh-WISP2, seeded at high cell numbers and quickly consuming the glucose initially available at physiological levels. Experimental outputs revealed that the more aggressive MCF7-sh-WISP2 cells have the ability to survive and sustain substantial proliferation in low-glucose conditions, as opposed to the less aggressive MCF7 cells. Changes in lactate levels \textit{in situ} suggested lactate uptake by MCF7-sh-WISP2 cells, and flow cytometry and {immunocytochemistry} analyses indicated an associated increase in MCT1 expression, which was then reversed when cells were rescued and exposed again to higher glucose levels. 

Experimental data \rev{from `glucose-deprivation' experiments} on the MCF7-sh-WISP2 cell line were used to calibrate the proposed mathematical model of cell evolutionary dynamics, \rev{as well as to conduct uncertainty quantification of the model calibration results,} and the MCT1 expression distributions obtained through flow cytometry analyses were compared with those predicted by the mathematical model. We found that the calibrated model, whose numerical simulation results are in good quantitative agreement with experimental data, best reproduces experimental observations when the effects of both \rev{FECs and SPCs}  in MCT1 expression are taken into account. This finding \rev{suggests} that cognate studies considering only one of these two types of changes in protein expression may be overestimating the rates at which the considered type of change occurs, overall disregarding the combined effect of the two of them. 

The analytical and numerical results of the calibrated model presented here \rev{suggest} that environment-induced changes in MCT1 expression mediated by lactate-associated signalling pathways enable a prompt adaptive response of glucose-deprived cancer cells. Furthermore, \rev{fluctuations in MCT1 expression due to epigenetic changes may} create the substrate for natural selection to act upon, speeding up the selective sweep underlying cancer cell adaptation to glucose deprivation, and may constitute a long-term bet-hedging mechanism. These results on the respective roles played by \rev{FECs and SPCs}  in MCT1 expression in the evolutionary dynamics of cancer cells, whilst having been obtained for glucose-deprived cells of the MCF7-sh-WISP2 line, may extend to other cell lines and scenarios whereby changes in protein expression elicit metabolic reprogramming of cancer cells under nutrient deprivation -- e.g. HIF1 favouring anaerobic energy pathways or CD36 promoting fatty acid uptake~\cite{liu2021metabolic,zaoui2019breast}. 

%\subsection{Relation to \textit{in vivo} tumours}
The optimal parameter set recovered from model calibration\rev{, and related bootstrap sampling distributions obtained through the uncertainty quantification procedure,} suggest that the MCT1 fluorescence intensity levels recorded at the end of the \textit{in vitro} experiments on MCF7-sh-WISP2 cells do not correspond to maximal levels of lactate uptake, and MCT1 expression levels may continue to increase over \rev{time}.  
In practice, performing the experiments over a longer timeframe we expect cells to die out faster than as predicted by the model, due to external factors, demographic stochasticity at low cell numbers or additional byproducts of cell metabolism that are not incorporated into the modelling framework proposed here.  
\rev{Therefore, our study supports the idea that alternative experimental conditions may need to be considered in studies aimed at investigating the fitness of aggressive breast cancer cells at maximum MCT1 expression levels in glucose-deprived conditions -- for instance by periodically, or continuously,  adding lactate to the medium and recording data for a longer period of time. This would also reproduce the inflow of lactate we expect the cells to be exposed to \textit{in vivo}, as cells closer to blood vessels may still perform glycolysis due to exposure to high glucose concentrations and thus produce lactate, which could then reach glucose-deprived areas via spatial diffusion~\cite{fiandaca2021mathematical}.} 
Nevertheless, it is evident from our experimental and numerical results that the observed increase in MCT1 expression of MCF7-sh-WISP2 cells over the span of a few days is sufficient to ensure survival and sustain proliferation under glucose deprivation, maintaining the population at high cell numbers for about a week, by the end of which we expect cells to have initiated alternative survival mechanisms associated with disease progression \textit{in vivo}~\cite{liu2021metabolic,tasdogan2020metabolic}.  

\rev{Analogous results have been obtained calibrating the model with data from `glucose-deprivation' and `rescue' experiments, although with a worse quantitative agreement between numerical simulation results and experimental data. 
% although with a rather qualitative than quantitative agreement. 
This suggests that additional evolutionary mechanisms may need to be considered in regimes of glucose abundance wherein, for instance, carrying capacities and net proliferation rates may differ from those recorded at physiological levels of glucose~\cite{friedl2003tumour}.  Nonetheless, such regimes are outside the scope of this work, and the good qualitative agreement of the model results with experimental data still highlights the model's ability to predict the reversibility of MCT1 expression observed in `rescue' experiments.}

While we recorded an increase in MCT1 expression of glucose-deprived MCF7-sh-WISP2 cells, no MCT4 mRNA was detected, which suggests that the \textit{in vitro} environmental conditions here investigated do not influence the expression of such a monocarboxylate transporter. It would be relevant to perform similar experiments under hypoxic conditions, as hypoxia-regulated signalling pathways may explain the increase in MCT4 expression observed \textit{in vivo} far from tumour blood vessels. In this regard, it would also be significant to formulate a spatially-explicit extension of the present model where oxygen dependency of various dynamics, here ignored as experiments were carried out in normoxic conditions, was modelled explicitly. Such an extended model would allow for theoretical studies on the still debated role of hypoxia in MCT1 expression at tissue level, which might reconcile reported discrepancies between oxygen and pH profiles~\cite{helmlinger1997interstitial,miranda2016hypoxia,ullah2006plasma,wang2015exploring}, and could inform anti-cancer therapeutic approaches based on MCT1 blockers~\cite{corbet2018interruption,wang2021lactate}.

\section*{Acknowledgements} The authors acknowledge the valuable assistance of Romain Morichon and Annie Munier of the Sorbonne Universit{\'e}-INSERM, UMR\_S938, Centre de Recherche Saint-Antoine Imagery and cytometry platform. T.L. thanks the Istituto Nazionale di Alta Matematica (INdAM) and the Gruppo Nazionale per la Fisica Matematica (GNFM) for their support.

\section*{Funding} C.V. has received funding from the European Research Council (ERC) under the European Union's Horizon 2020 research and innovation programme (grant agreement No 740623). T.L. gratefully acknowledges support from the Italian Ministry of University and Research (MUR) through the grant ``Dipartimenti di Eccellenza 2018-2022'' (Project no. E11G18000350001) and the PRIN 2020 project (No. 2020JLWP23) ``Integrated Mathematical Approaches to Socio-Epidemiological Dynamics'' (CUP: E15F21005420006). T.L. and C.V. gratefully acknowledge support of the Institut Henri Poincar\'e (UAR 839 CNRS-Sorbonne Universit\'e), and LabEx CARMIN (ANR-10-LABX-59-01). L.A., T.L. and C.V. gratefully acknowledge support from the CNRS International Research Project `Mod\'elisation de la biom\'ecanique cellulaire et tissulaire' (MOCETIBI).

\bibliographystyle{unsrt}

\bibliography{Lactate.bib}

\begin{thebibliography}{10}

\bibitem{romero2011tumor}
Susana Romero-Garcia, Jose~Sullivan Lopez-Gonzalez, Jos{\'e}~Luis B{\'{}}~ez
  Viveros, Dolores Aguilar-Cazares, and Heriberto Prado-Garcia.
\newblock Tumor cell metabolism: an integral view.
\newblock {\em Cancer Biology \& Therapy}, 12(11):939--948, 2011.

\bibitem{doherty2013targeting}
Joanne~R Doherty, John~L Cleveland, et~al.
\newblock Targeting lactate metabolism for cancer therapeutics.
\newblock {\em The Journal of Clinical Investigation}, 123(9):3685--3692, 2013.

\bibitem{garcia2019tumor}
Juan~C Garc{\'\i}a-Ca{\~n}averas, Li~Chen, and Joshua~D Rabinowitz.
\newblock The tumor metabolic microenvironment: Lessons from lactate.
\newblock {\em Cancer Research}, 79(13):3155--3162, 2019.

\bibitem{keenan2015alternative}
Melissa Keenan and Jen-Tsan Chi.
\newblock Alternative fuels for cancer cells.
\newblock {\em Cancer Journal}, 21(2):49, 2015.

\bibitem{wang2015exploring}
Qian Wang, Peter Vaupel, Sibylle~I Ziegler, and Kuangyu Shi.
\newblock Exploring the quantitative relationship between metabolism and
  enzymatic phenotype by physiological modeling of glucose metabolism and
  lactate oxidation in solid tumors.
\newblock {\em Physics in Medicine \& Biology}, 60(6):2547, 2015.

\bibitem{hanahan2011hallmarks}
Douglas Hanahan and Robert~A Weinberg.
\newblock Hallmarks of cancer: the next generation.
\newblock {\em Cell}, 144(5):646--674, 2011.

\bibitem{ippolito2019lactate}
Luigi Ippolito, Andrea Morandi, Elisa Giannoni, and Paola Chiarugi.
\newblock Lactate: a metabolic driver in the tumour landscape.
\newblock {\em Trends in Biochemical Sciences}, 44(2):153--166, 2019.

\bibitem{zhao2020cancer}
Chongru Zhao, Min Wu, Ning Zeng, Mingchen Xiong, Weijie Hu, Wenchang Lv, Yi~Yi,
  Qi~Zhang, and Yiping Wu.
\newblock Cancer-associated adipocytes: Emerging supporters in breast cancer.
\newblock {\em Journal of Experimental \& Clinical Cancer Research},
  39(1):1--17, 2020.

\bibitem{halestrap2013monocarboxylic}
Andrew~P Halestrap.
\newblock Monocarboxylic acid transport.
\newblock {\em Comprehensive Physiology}, 3(4):1611--1643, 2013.

\bibitem{khan2020targeting}
Aaminah Khan, Emanuele Valli, Hayley Lam, David~A Scott, Jayne Murray,
  Kimberley~M Hanssen, Georgina Eden, Laura~D Gamble, Rupinder Pandher,
  Claudia~L Flemming, et~al.
\newblock Targeting metabolic activity in high-risk neuroblastoma through
  monocarboxylate transporter 1 (mct1) inhibition.
\newblock {\em Oncogene}, 39(17):3555--3570, 2020.

\bibitem{tasdogan2020metabolic}
Alpaslan Tasdogan, Brandon Faubert, Vijayashree Ramesh, Jessalyn~M Ubellacker,
  Bo~Shen, Ashley Solmonson, Malea~M Murphy, Zhimin Gu, Wen Gu, Misty Martin,
  et~al.
\newblock Metabolic heterogeneity confers differences in melanoma metastatic
  potential.
\newblock {\em Nature}, 577(7788):115--120, 2020.

\bibitem{liu2021metabolic}
Baoyi Liu and Xin Zhang.
\newblock Metabolic reprogramming underlying brain metastasis of breast cancer.
\newblock {\em Frontiers in Molecular Biosciences}, 8, 2021.

\bibitem{corbet2018interruption}
Cyril Corbet, Estelle Bastien, Nihed Draoui, Bastien Doix, Lionel Mignion,
  B{\'e}n{\'e}dicte~F Jordan, Arnaud Marchand, Jean-Christophe Vanherck,
  Patrick Chaltin, Olivier Schakman, et~al.
\newblock Interruption of lactate uptake by inhibiting mitochondrial pyruvate
  transport unravels direct antitumor and radiosensitizing effects.
\newblock {\em Nature Communications}, 9(1):1--11, 2018.

\bibitem{wang2021lactate}
Zi-Hao Wang, Wen-Bei Peng, Pei Zhang, Xiang-Ping Yang, and Qiong Zhou.
\newblock Lactate in the tumour microenvironment: From immune modulation to
  therapy.
\newblock {\em EBioMedicine}, 73:103627, 2021.

\bibitem{longhitano2022lactate}
Lucia Longhitano, Nunzio Vicario, Daniele Tibullo, Cesarina Giallongo, Giuseppe
  Broggi, Rosario Caltabiano, Giuseppe Maria~Vincenzo Barbagallo, Roberto
  Altieri, Marta Baghini, Michelino Di~Rosa, et~al.
\newblock Lactate induces the expressions of {MCT1} and {HCAR1} to promote
  tumor growth and progression in glioblastoma.
\newblock {\em Frontiers in Oncology}, 12, 2022.

\bibitem{huang2013genetic}
Sui Huang.
\newblock Genetic and non-genetic instability in tumor progression: link
  between the fitness landscape and the epigenetic landscape of cancer cells.
\newblock {\em Cancer and Metastasis Reviews}, 32(3):423--448, 2013.

\bibitem{wang2023lactate}
Ting Wang, Zeng Ye, Zheng Li, De-sheng Jing, Gui-xiong Fan, Meng-qi Liu,
  Qi-feng Zhuo, Shun-rong Ji, Xian-jun Yu, Xiao-wu Xu, et~al.
\newblock Lactate-induced protein lactylation: A bridge between epigenetics and
  metabolic reprogramming in cancer.
\newblock {\em Cell proliferation}, page e13478, 2023.

\bibitem{bergers2021metabolism}
Gabriele Bergers and Sarah-Maria Fendt.
\newblock The metabolism of cancer cells during metastasis.
\newblock {\em Nature Reviews Cancer}, 21(3):162--180, 2021.

\bibitem{zhang2019metabolic}
Di~Zhang, Zhanyun Tang, He~Huang, Guolin Zhou, Chang Cui, Yejing Weng, Wenchao
  Liu, Sunjoo Kim, Sangkyu Lee, Mathew Perez-Neut, et~al.
\newblock Metabolic regulation of gene expression by histone lactylation.
\newblock {\em Nature}, 574(7779):575--580, 2019.

\bibitem{ippolito2022lactate}
Luigi Ippolito, Giuseppina Comito, Matteo Parri, Marta Iozzo, Assia Duatti,
  Francesca Virgilio, Nicla Lorito, Marina Bacci, Elisa Pardella, Giada
  Sandrini, et~al.
\newblock Lactate rewires lipid metabolism and sustains a metabolic--epigenetic
  axis in prostate cancer.
\newblock {\em Cancer Research}, 82(7):1267--1282, 2022.

\bibitem{ardavseva2019dissecting}
Aleksandra Arda{\v{s}}eva, Robert~A. Gatenby, Alexander R.~A. Anderson,
  Helen~M. Byrne, Philip~K. Maini, and Tommaso Lorenzi.
\newblock A mathematical dissection of the adaptation of cell populations to
  fluctuating oxygen levels.
\newblock {\em Bulletin of Mathematical Biology}, 82(6):81, 2020.

\bibitem{celora2022dna}
Giulia~L Celora, Samuel~B Bader, Ester~M Hammond, Philip~K Maini, Joe~M
  Pitt-Francis, and Helen~M Byrne.
\newblock A {DNA}-structured mathematical model of cell-cycle progression in
  cyclic hypoxia.
\newblock {\em Journal of Theoretical Biology}, 545:111104, 2022.

\bibitem{cho2017modeling}
Heyrim Cho and Doron Levy.
\newblock {Modeling the dynamics of heterogeneity of solid tumors in response
  to chemotherapy}.
\newblock {\em Bulletin of Mathematical Biology}, 79(12):2986--3012, 2017.

\bibitem{fiandaca2021mathematical}
Giada Fiandaca, Marcello Delitala, and Tommaso Lorenzi.
\newblock A mathematical study of the influence of hypoxia and acidity on the
  evolutionary dynamics of cancer.
\newblock {\em Bulletin of Mathematical Biology}, 83(7):1--29, 2021.

\bibitem{villa2021evolutionary}
Chiara Villa, Mark~AJ Chaplain, and Tommaso Lorenzi.
\newblock Evolutionary dynamics in vascularised tumours under chemotherapy:
  Mathematical modelling, asymptotic analysis and numerical simulations.
\newblock {\em Vietnam Journal of Mathematics}, 49(1):143--167, 2021.

\bibitem{mcgillen2014glucose}
Jessica~B McGillen, Catherine~J Kelly, Alicia Mart{\'\i}nez-Gonz{\'a}lez,
  Natasha~K Martin, Eamonn~A Gaffney, Philip~K Maini, and V{\'\i}ctor~M
  P{\'e}rez-Garc{\'\i}a.
\newblock Glucose--lactate metabolic cooperation in cancer: Insights from a
  spatial mathematical model and implications for targeted therapy.
\newblock {\em Journal of Theoretical Biology}, 361:190--203, 2014.

\bibitem{mendoza2012mathematical}
Berta Mendoza-Juez, Alicia Mart{\'\i}nez-Gonz{\'a}lez, Gabriel~F Calvo, and
  V{\'\i}ctor~M P{\'e}rez-Garc{\'\i}a.
\newblock A mathematical model for the glucose-lactate metabolism of in vitro
  cancer cells.
\newblock {\em Bulletin of Mathematical Biology}, 74(5):1125--1142, 2012.

\bibitem{linden2022bayesian}
Nathaniel~J Linden, Boris Kramer, and Padmini Rangamani.
\newblock Bayesian parameter estimation for dynamical models in systems
  biology.
\newblock {\em bioRxiv}, 2022.

\bibitem{martinez2014bayesopt}
Ruben Martinez-Cantin.
\newblock {BayesOpt: a Bayesian optimization library for nonlinear
  optimization, experimental design and bandits}.
\newblock {\em Journal of Machine Learning Research}, 15(1):3735--3739, 2014.

\bibitem{myung2003tutorial}
In~Jae Myung.
\newblock {Tutorial on maximum likelihood estimation}.
\newblock {\em Journal of Mathematical Psychology}, 47(1):90--100, 2003.

\bibitem{spilker2001evaluation}
Mary~E Spilker and Paolo Vicini.
\newblock An evaluation of extended vs weighted least squares for parameter
  estimation in physiological modeling.
\newblock {\em Journal of Biomedical Informatics}, 34(5):348--364, 2001.

\bibitem{thompson1996tutorial}
Paul~A Thompson and OH44106 Cleveland.
\newblock A tutorial on bootstrapping in the sas system.
\newblock {\em SAS Institute Inc}, 1996.

\bibitem{zhu1997making}
Weimo Zhu.
\newblock Making bootstrap statistical inferences: A tutorial.
\newblock {\em Research Quarterly for Exercise and Sport}, 68(1):44--55, 1997.

\bibitem{ferrand2014loss}
Nathalie Ferrand, Anne Gnanapragasam, Guillaume Dorothee, G{\'e}rard Redeuilh,
  Annette~K Larsen, and Mich{\`e}le Sabbah.
\newblock {Loss of WISP2/CCN5 in estrogen-dependent MCF7 human breast cancer
  cells promotes a stem-like cell phenotype}.
\newblock {\em PloS One}, 9(2):e87878, 2014.

\bibitem{sabbah2011ccn5}
Mich{\`e}le Sabbah, C{\'e}line Prunier, Nathalie Ferrand, Virginie
  Megalophonos, Kathleen Lambein, Olivier De~Wever, Nicolas Nazaret, Jo{\"e}l
  Lachuer, Sylvie Dumont, and G{\'e}rard Redeuilh.
\newblock {CCN5, a novel transcriptional repressor of the transforming growth
  factor $\beta$ signaling pathway}.
\newblock {\em Molecular and Cellular Biology}, 31(7):1459--1469, 2011.

\bibitem{park2018overview}
Simon~J Park, Chase~P Smith, Ryan~R Wilbur, Charles~P Cain, Sankeerth~R Kallu,
  Srijan Valasapalli, Arpit Sahoo, Maheedhara~R Guda, Andrew~J Tsung, and
  Kiran~K Velpula.
\newblock {An overview of MCT1 and MCT4 in GBM: small molecule transporters
  with large implications}.
\newblock {\em American Journal of Cancer Research}, 8(10):1967, 2018.

\bibitem{sonveaux2008targeting}
Pierre Sonveaux, Fr{\'e}d{\'e}rique V{\'e}gran, Thies Schroeder, Melanie~C
  Wergin, Julien Verrax, Zahid~N Rabbani, Christophe~J De~Saedeleer, Kelly~M
  Kennedy, Caroline Diepart, B{\'e}n{\'e}dicte~F Jordan, et~al.
\newblock Targeting lactate-fueled respiration selectively kills hypoxic tumor
  cells in mice.
\newblock {\em The Journal of Clinical Investigation}, 118(12):3930--3942,
  2008.

\bibitem{lorenzi2016tracking}
Tommaso Lorenzi, Rebecca~H Chisholm, and Jean Clairambault.
\newblock Tracking the evolution of cancer cell populations through the
  mathematical lens of phenotype-structured equations.
\newblock {\em Biology Direct}, 11(1):43, 2016.

\bibitem{chisholm2016evolutionary}
Rebecca~H Chisholm, Tommaso Lorenzi, Laurent Desvillettes, and Barry~D Hughes.
\newblock Evolutionary dynamics of phenotype-structured populations: from
  individual-level mechanisms to population-level consequences.
\newblock {\em Zeitschrift f{\"u}r angewandte Mathematik und Physik},
  67(4):1--34, 2016.

\bibitem{gedeon2012delayed}
Tom{\'a}{\v{s}} Gedeon and Pavol Bokes.
\newblock Delayed protein synthesis reduces the correlation between mrna and
  protein fluctuations.
\newblock {\em Biophysical Journal}, 103(3):377--385, 2012.

\bibitem{greenbaum2003comparing}
Dov Greenbaum, Christopher Colangelo, Kenneth Williams, and Mark Gerstein.
\newblock Comparing protein abundance and mrna expression levels on a genomic
  scale.
\newblock {\em Genome Biology}, 4(9):1--8, 2003.

\bibitem{zaoui2019breast}
Maurice Zaoui, Mehdi Morel, Nathalie Ferrand, Soraya Fellahi, Jean-Philippe
  Bastard, Antonin Lamazi{\`e}re, Annette~Kragh Larsen, V{\'e}ronique
  B{\'e}r{\'e}ziat, Michael Atlan, and Mich{\`e}le Sabbah.
\newblock {Breast-associated adipocytes secretome induce fatty acid uptake and
  invasiveness in breast cancer cells via CD36 independently of body mass
  index, menopausal status and mammary density}.
\newblock {\em Cancers}, 11(12):2012, 2019.

\bibitem{friedl2003tumour}
Peter Friedl and Katarina Wolf.
\newblock Tumour-cell invasion and migration: diversity and escape mechanisms.
\newblock {\em Nature Reviews Cancer}, 3(5):362--374, 2003.

\bibitem{helmlinger1997interstitial}
Gabriel Helmlinger, Fan Yuan, Marc Dellian, and Rakesh~K Jain.
\newblock {Interstitial pH and pO2 gradients in solid tumors in vivo:
  high-resolution measurements reveal a lack of correlation}.
\newblock {\em Nature Medicine}, 3(2):177--182, 1997.

\bibitem{miranda2016hypoxia}
Vera Miranda-Gon{\c{c}}alves, Sara Granja, Olga Martinho, Mrinalini Honavar,
  Marta Pojo, Bruno~M Costa, Manuel~M Pires, C{\'e}lia Pinheiro, Michelle
  Cordeiro, Gil Bebiano, et~al.
\newblock {Hypoxia-mediated upregulation of MCT1 expression supports the
  glycolytic phenotype of glioblastomas}.
\newblock {\em Oncotarget}, 7(29):46335, 2016.

\bibitem{ullah2006plasma}
Mohammed~S Ullah, Andrew~J Davies, and Andrew~P Halestrap.
\newblock {The plasma membrane lactate transporter MCT4, but not MCT1, is
  up-regulated by hypoxia through a HIF-1$\alpha$-dependent mechanism}.
\newblock {\em Journal of Biological Chemistry}, 281(14):9030--9037, 2006.

\end{thebibliography}


\begin{thebibliography}{10}

\bibitem{alfaro2014explicit}
{\sc M.~Alfaro and R.~Carles}, {\em Explicit solutions for replicator-mutator
  equations: extinction versus acceleration}, SIAM Journal on Applied
  Mathematics, 74 (2014), pp.~1919--1934.

\bibitem{almeida2019evolution}
{\sc L.~Almeida, P.~Bagnerini, G.~Fabrini, B.~D. Hughes, and T.~Lorenzi}, {\em
  Evolution of cancer cell populations under cytotoxic therapy and treatment
  optimisation: insight from a phenotype-structured model}, ESAIM: Mathematical
  Modelling and Numerical Analysis, 53 (2019), pp.~1157--1190.

\bibitem{ardavseva2020comparative}
{\sc A.~Arda{\v{s}}eva, A.~R.~A. Anderson, R.~A. Gatenby, H.~M. Byrne, P.~K.
  Maini, and T.~Lorenzi}, {\em Comparative study between discrete and continuum
  models for the evolution of competing phenotype-structured cell populations
  in dynamical environments}, Physical Review E, 102 (2020), p.~042404.

\bibitem{ardavseva2020evolutionary}
{\sc A.~Arda{\v{s}}eva, R.~A. Gatenby, A.~R. Anderson, H.~M. Byrne, P.~K.
  Maini, and T.~Lorenzi}, {\em Evolutionary dynamics of competing
  phenotype-structured populations in periodically fluctuating environments},
  Journal of Mathematical Biology, 80 (2020), pp.~775--807.

\bibitem{bergers2021metabolism}
{\sc G.~Bergers and S.-M. Fendt}, {\em The metabolism of cancer cells during
  metastasis}, Nature Reviews Cancer, 21 (2021), pp.~162--180.

\bibitem{celora2021phenotypic}
{\sc G.~L. Celora, H.~M. Byrne, C.~E. Zois, and P.~G. Kevrekidis}, {\em
  Phenotypic variation modulates the growth dynamics and response to
  radiotherapy of solid tumours under normoxia and hypoxia}, Journal of
  Theoretical Biology, 527 (2021), p.~110792.

\bibitem{chisholm2016evolutionary}
{\sc R.~H. Chisholm, T.~Lorenzi, L.~Desvillettes, and B.~D. Hughes}, {\em
  Evolutionary dynamics of phenotype-structured populations: from
  individual-level mechanisms to population-level consequences}, Zeitschrift
  f{\"u}r angewandte Mathematik und Physik, 67 (2016), pp.~1--34.

\bibitem{chisholm2015emergence}
{\sc R.~H. Chisholm, T.~Lorenzi, A.~Lorz, A.~K. Larsen, L.~Almeida,
  A.~Escargueil, and J.~Clairambault}, {\em Emergence of drug tolerance in
  cancer cell populations: an evolutionary outcome of selection, non-genetic
  instability and stress-induced adaptation}, Cancer Research, 75 (2015),
  pp.~930--939.

\bibitem{ferrand2014loss}
{\sc N.~Ferrand, A.~Gnanapragasam, G.~Dorothee, G.~Redeuilh, A.~K. Larsen, and
  M.~Sabbah}, {\em {Loss of WISP2/CCN5 in estrogen-dependent MCF7 human breast
  cancer cells promotes a stem-like cell phenotype}}, PloS One, 9 (2014),
  p.~e87878.

\bibitem{fiandaca2021mathematical}
{\sc G.~Fiandaca, M.~Delitala, and T.~Lorenzi}, {\em A mathematical study of
  the influence of hypoxia and acidity on the evolutionary dynamics of cancer},
  Bulletin of Mathematical Biology, 83 (2021), pp.~1--29.

\bibitem{freischel2021frequency}
{\sc A.~R. Freischel, M.~Damaghi, J.~J. Cunningham, A.~Ibrahim-Hashim, R.~J.
  Gillies, R.~A. Gatenby, and J.~S. Brown}, {\em Frequency-dependent
  interactions determine outcome of competition between two breast cancer cell
  lines}, Scientific Reports, 11 (2021), pp.~1--18.

\bibitem{fritah2008role}
{\sc A.~Fritah, C.~Saucier, O.~De~Wever, M.~Bracke, I.~Bi{\`e}che, R.~Lidereau,
  C.~Gespach, S.~Drouot, G.~Redeuilh, and M.~Sabbah}, {\em {Role of WISP-2/CCN5
  in the maintenance of a differentiated and noninvasive phenotype in human
  breast cancer cells}}, Molecular and Cellular Biology, 28 (2008),
  pp.~1114--1123.

\bibitem{huang2013genetic}
{\sc S.~Huang}, {\em Genetic and non-genetic instability in tumor progression:
  link between the fitness landscape and the epigenetic landscape of cancer
  cells}, Cancer and Metastasis Reviews, 32 (2013), pp.~423--448.

\bibitem{ippolito2019lactate}
{\sc L.~Ippolito, A.~Morandi, E.~Giannoni, and P.~Chiarugi}, {\em Lactate: a
  metabolic driver in the tumour landscape}, Trends in Biochemical Sciences, 44
  (2019), pp.~153--166.

\bibitem{keenan2015alternative}
{\sc M.~Keenan and J.-T. Chi}, {\em Alternative fuels for cancer cells}, Cancer
  Journal, 21 (2015), p.~49.

\bibitem{khan2020targeting}
{\sc A.~Khan, E.~Valli, H.~Lam, D.~A. Scott, J.~Murray, K.~M. Hanssen, G.~Eden,
  L.~D. Gamble, R.~Pandher, C.~L. Flemming, et~al.}, {\em Targeting metabolic
  activity in high-risk neuroblastoma through monocarboxylate transporter 1
  (mct1) inhibition}, Oncogene, 39 (2020), pp.~3555--3570.

\bibitem{linden2022bayesian}
{\sc N.~J. Linden, B.~Kramer, and P.~Rangamani}, {\em Bayesian parameter
  estimation for dynamical models in systems biology}, bioRxiv,  (2022).

\bibitem{longhitano2022lactate}
{\sc L.~Longhitano, N.~Vicario, D.~Tibullo, C.~Giallongo, G.~Broggi,
  R.~Caltabiano, G.~M.~V. Barbagallo, R.~Altieri, M.~Baghini, M.~Di~Rosa,
  et~al.}, {\em Lactate induces the expressions of {MCT1} and {HCAR1} to
  promote tumor growth and progression in glioblastoma}, Frontiers in Oncology,
  12 (2022).

\bibitem{lorenzi2016tracking}
{\sc T.~Lorenzi, R.~H. Chisholm, and J.~Clairambault}, {\em Tracking the
  evolution of cancer cell populations through the mathematical lens of
  phenotype-structured equations}, Biology Direct, 11 (2016), p.~43.

\bibitem{lorenzi2015dissecting}
{\sc T.~Lorenzi, R.~H. Chisholm, L.~Desvillettes, and B.~D. Hughes}, {\em
  Dissecting the dynamics of epigenetic changes in phenotype-structured
  populations exposed to fluctuating environments}, Journal of Theoretical
  Biology, 386 (2015), pp.~166--176.

\bibitem{lorenzi2020discrete}
{\sc T.~Lorenzi, F.~R. Macfarlane, and C.~Villa}, {\em Discrete and continuum
  models for the evolutionary and spatial dynamics of cancer: a very short
  introduction through two case studies}, Trends in Biomathematics: Modeling
  Cells, Flows, Epidemics, and the Environment: Selected Works from the BIOMAT
  Consortium Lectures, Szeged, Hungary, 2019 19,  (2020), pp.~359--380.

\bibitem{martinez2014bayesopt}
{\sc R.~Martinez-Cantin}, {\em {BayesOpt: a Bayesian optimization library for
  nonlinear optimization, experimental design and bandits}}, Journal of Machine
  Learning Research, 15 (2014), pp.~3735--3739.

\bibitem{mcgillen2014glucose}
{\sc J.~B. McGillen, C.~J. Kelly, A.~Mart{\'\i}nez-Gonz{\'a}lez, N.~K. Martin,
  E.~A. Gaffney, P.~K. Maini, and V.~M. P{\'e}rez-Garc{\'\i}a}, {\em
  Glucose--lactate metabolic cooperation in cancer: Insights from a spatial
  mathematical model and implications for targeted therapy}, Journal of
  Theoretical Biology, 361 (2014), pp.~190--203.

\bibitem{mendoza2012mathematical}
{\sc B.~Mendoza-Juez, A.~Mart{\'\i}nez-Gonz{\'a}lez, G.~F. Calvo, and V.~M.
  P{\'e}rez-Garc{\'\i}a}, {\em A mathematical model for the glucose-lactate
  metabolism of in vitro cancer cells}, Bulletin of Mathematical Biology, 74
  (2012), pp.~1125--1142.

\bibitem{molavian2009fingerprint}
{\sc H.~R. Molavian, M.~Kohandel, M.~Milosevic, and S.~Sivaloganathan}, {\em
  Fingerprint of cell metabolism in the experimentally observed interstitial ph
  and po2 in solid tumors}, Cancer Research, 69 (2009), pp.~9141--9147.

\bibitem{myung2003tutorial}
{\sc I.~J. Myung}, {\em {Tutorial on maximum likelihood estimation}}, Journal
  of Mathematical Psychology, 47 (2003), pp.~90--100.

\bibitem{phipps2015microscale}
{\sc C.~Phipps, H.~Molavian, and M.~Kohandel}, {\em {A microscale mathematical
  model for metabolic symbiosis: Investigating the effects of metabolic
  inhibition on ATP turnover in tumors}}, Journal of Theoretical Biology, 366
  (2015), pp.~103--114.

\bibitem{santillan2008use}
{\sc M.~Santill{\'a}n}, {\em {On the use of the Hill functions in mathematical
  models of gene regulatory networks}}, Mathematical Modelling of Natural
  Phenomena, 3 (2008), pp.~85--97.

\bibitem{spilker2001evaluation}
{\sc M.~E. Spilker and P.~Vicini}, {\em An evaluation of extended vs weighted
  least squares for parameter estimation in physiological modeling}, Journal of
  Biomedical Informatics, 34 (2001), pp.~348--364.

\bibitem{tasdogan2020metabolic}
{\sc A.~Tasdogan, B.~Faubert, V.~Ramesh, J.~M. Ubellacker, B.~Shen,
  A.~Solmonson, M.~M. Murphy, Z.~Gu, W.~Gu, M.~Martin, et~al.}, {\em Metabolic
  heterogeneity confers differences in melanoma metastatic potential}, Nature,
  577 (2020), pp.~115--120.

\bibitem{thompson1996tutorial}
{\sc P.~A. Thompson and O.~Cleveland}, {\em A tutorial on bootstrapping in the
  sas system}, SAS Institute Inc,  (1996).

\bibitem{vaupel1989blood}
{\sc P.~Vaupel, F.~Kallinowski, and P.~Okunieff}, {\em Blood flow, oxygen and
  nutrient supply, and metabolic microenvironment of human tumors: a review},
  Cancer Research, 49 (1989), pp.~6449--6465.

\bibitem{villa2021evolutionary}
{\sc C.~Villa, M.~A. Chaplain, and T.~Lorenzi}, {\em Evolutionary dynamics in
  vascularised tumours under chemotherapy: Mathematical modelling, asymptotic
  analysis and numerical simulations}, Vietnam Journal of Mathematics, 49
  (2021), pp.~143--167.

\bibitem{wang2023lactate}
{\sc T.~Wang, Z.~Ye, Z.~Li, D.-s. Jing, G.-x. Fan, M.-q. Liu, Q.-f. Zhuo, S.-r.
  Ji, X.-j. Yu, X.-w. Xu, et~al.}, {\em Lactate-induced protein lactylation: A
  bridge between epigenetics and metabolic reprogramming in cancer}, Cell
  proliferation,  (2023), p.~e13478.

\bibitem{zhang2019metabolic}
{\sc D.~Zhang, Z.~Tang, H.~Huang, G.~Zhou, C.~Cui, Y.~Weng, W.~Liu, S.~Kim,
  S.~Lee, M.~Perez-Neut, et~al.}, {\em Metabolic regulation of gene expression
  by histone lactylation}, Nature, 574 (2019), pp.~575--580.

\bibitem{zhu1997making}
{\sc W.~Zhu}, {\em Making bootstrap statistical inferences: A tutorial},
  Research Quarterly for Exercise and Sport, 68 (1997), pp.~44--55.

\end{thebibliography}

\end{document}

\maketitle
\pagestyle{plain}
%\tableofcontents
\pagenumbering{arabic}
\renewcommand{\thefigure}{S\arabic{figure}}
\setcounter{figure}{0}
\renewcommand{\theequation}{S\arabic{equation}}
\setcounter{equation}{0}
\renewcommand{\thesection}{S\arabic{section}}
\setcounter{section}{0}

\section{Mathematical model}\label{sec:model}
Building on the modelling strategies presented in~\cite{lorenzi2016tracking}, we develop a mathematical model that describes the evolutionary dynamics of a population of MCF7-sh-WISP2 cells, structured by the level of MCT1 expression, under the environmental conditions determined by the levels of glucose and lactate in the extracellular environment.  \rev{The model relies on the following assumptions, justified by the literature or our experimental observations:
\begin{itemize}
\item[{\bf A1}] There is a low level of MCT1 expression endowing cells with the highest proliferation rate via glycolysis, and a higher level of MCT1 expression endowing cells with the highest rate of proliferation via lactate reuse when glucose is scarce -- cf. the experimental results underlying our study.
\item[{\bf A2}] There is a threshold level of glucose above which cells prioritise glucose uptake~\cite{keenan2015alternative}.
\item[{\bf A3}] Lactate binding to the membrane of cancer cells triggers regulatory pathways increasing the transcriptional activity -- and thus level of expression -- of MCT1 and, conversely, the interruption of lactate signalling induces a reduction in MCT1 expression levels~\cite{ippolito2019lactate,longhitano2022lactate}.  
\item[{\bf A4}] MCT1 expression levels may undergo fluctuations due to epigenetic changes interfering with the transcriptional activity detailed in assumption {\bf A3}~\cite{huang2013genetic}, and the rate at which these changes occur increases with lactate uptake~\cite{wang2023lactate}.
\item[{\bf A5}] Cancer cells proliferate and die according to their fitness in relation to the environmental conditions they are exposed to, and may also die due to competition for space~\cite{lorenzi2016tracking}.
\item[{\bf A6}] The rate of cell proliferation via glycolysis and the corresponding rates of glucose consumption and lactate production are proportional to the rate of glucose uptake, whereas the rate of cell proliferation via lactate reuse and the corresponding rate of lactate consumption are proportional to the rate of lactate uptake. Furthermore, lactate consumption is mediated by the cells' MCT1 expression level~\cite{khan2020targeting,tasdogan2020metabolic}.
\item[{\bf A7}] Glucose and lactate uptake by cancer cells are mediated by ligand-receptor dynamics~~\cite{santillan2008use}.
\item[{\bf A8}] The MCT1 expression distribution at day 0 is in a Gaussian form -- cf. the experimental results underlying our study.
\end{itemize}
%More details on the justification of these assumptions are found in the Introduction of the Main Manuscript and in the rest of this section. 
}

\subsection{Preliminaries}\label{sec:model:prel}
We introduce the cell population density function $n(t,y)$, which represents the number of MCF7-sh-WISP2 cells with level of MCT1 expression $y\in\mathbb{R}$ at time $t \in \mathbb{R}_+$ (i.e. the MCT1 expression distribution of MCF7-sh-WISP2 cells at time $t$). The cell number, the mean level of MCT1 expression and the related variance, which provides a possible measure for the level of intercellular variability in MCT1 expression, are then computed, respectively, as
\beq
\label{erhomusigmay}
\rho(t) = \int_{\mathbb{R}} n(t,y) \; {\rm d}y, \quad \mu(t) = \frac{1}{\rho(t)} \int_{\mathbb{R}} y \, n(t,y) \; {\rm d}y, \quad \sigma^2(t)  = \frac{1}{\rho(t)} \int_{\mathbb{R}} y^2 \; n(t,y) \; {\rm d}y- \mu^2(t).
\eeq
We also introduce the functions $G(t)$ and $L(t)$, which model, respectively, the concentrations of glucose and lactate in the extracellular environment.  

%The results of {\it in vitro} experiments (cf. Sec.~3.2 in the Main Manuscript) indicate that the MCT1 expression distribution of MCF7-sh-WISP2 cells is, to a first approximation, unimodal with a single peak at the centre of the distribution. The location of the centre of the distribution moves from lower to higher expression levels when cells experience glucose deprivation and from higher to lower expression levels when cells are rescued from glucose deprivation. 
The results of {\it in vitro} experiments (cf. Sec. %~\ref{sec:results2}) 
3.2 in the Main Manuscript) 
indicate that the \rev{mean of the MCT1 expression distribution of MCF7-sh-WISP2 cells moves from lower to higher expression levels when cells experience glucose deprivation, and from higher to lower expression levels when cells are rescued from glucose deprivation.} Hence, we assume that there is a level of  MCT1 expression (i.e. the fittest level of MCT1 expression) endowing cells with the highest fitness depending on the environmental conditions determined by the concentrations of glucose and lactate. Moreover, the results of {\it in vitro} experiments (cf. Sec.~3.1 in the Main Manuscript) support the idea that proliferation and survival of MCF7-sh-WISP2 cells correlate with glucose uptake when glucose levels are sufficiently high and with lactate uptake when glucose levels are low. Therefore, we further assume that there are a level of MCT1 expression, $y_L$, endowing cells with the highest rate of proliferation via glycolysis and a higher level of MCT1 expression, $y_H>y_L$, endowing cells with the highest rate of proliferation via lactate reuse when glucose is scarce~\rev{(cf. assumption {\bf A1})} -- i.e. if the concentration of glucose in the extracellular environment is lower than a threshold level $G^*$ above which cells stop taking lactate from the extracellular environment in order to prioritise glucose uptake~\rev{(cf. assumption {\bf A2})}. We then introduce the following change of variable
\begin{equation}\label{def:x}
x=\frac{y-y_L}{y_H-y_L}\,,
\end{equation}
so that the rescaled level of MCT1 expression $x=0$ corresponds to the level of MCT1 expression $y=y_L$ and the rescaled level of MCT1 expression $x=1$ corresponds to the level of MCT1 expression $y=y_H$. Under the change of variable defined by Eq.~\eqref{def:x}, representing the rescaled MCT1 expression distribution of MCF7-sh-WISP2 cells at time $t$ by the cell population density function ${n_r(t,x) =  (y_H - y_L) \, n(t,y)}$, we compute the cell number, the mean rescaled level of MCT1 expression and the related variance, respectively, as
\beq
\label{erhomusigma}
\rho_r(t) = \int_{\mathbb{R}} n_r(t,x) \; {\rm d}x, \quad \mu_r(t) = \frac{1}{\rho_r(t)} \int_{\mathbb{R}} x \, n_r(t,x) \; {\rm d}x, \quad \sigma^2_r(t)  = \frac{1}{\rho_r(t)} \int_{\mathbb{R}} x^2 \; n_r(t,x) \; {\rm d}x- \mu^2_r(t).
\eeq
%n_r(t,x) = (y_H - y_L) \, n(t, y_L + x  (y_H-y_L)) = (y_H - y_L) \, n(t,y)

\begin{remark}
Note that the following relations hold between the quantities defined via Eq.~\eqref{erhomusigma} and Eq.~\eqref{erhomusigmay}:
\beq
\label{erhomusigmaxy}
{\rho(t) = \rho_r(t), \quad \mu(t) = y_L + \mu_r(t) (y_H - y_L), \quad \sigma^2(t) = \sigma^2_r(t) (y_H - y_L)^2}.
\eeq
\end{remark}

\subsection{Cell dynamics}\label{sec:model:cells}
The dynamics of the population density function $n_r(t,x)$ is governed by the following partial integro-differential equation (PIDE)
\begin{equation}
\label{ePDEn}
\begin{cases}
\displaystyle{\frac{\partial n_r}{\partial t}  \; - \; \Phi(G(t),L(t))\,\frac{\partial^2 n_r}{\partial x^2} \;+\; \Psi(G(t),L(t),\mu_r(t))\,\frac{\partial n_r}{\partial x} \; = \; R(x,G(t),L(t),\rho_r(t)) \, n_r,  \quad x \in \mathbb{R}}
\\\\
\displaystyle{\rho_r(t) = \int_{\mathbb{R}} n_r(t,x) \; {\rm d}x, \quad \mu_r(t) = \frac{1}{\rho_r(t)} \int_{\mathbb{R}} x \, n_r(t,x) \; {\rm d}x},
\end{cases}
\end{equation} 
where: 
\begin{itemize}[leftmargin=*]
\item the advection term $\displaystyle{\Psi(G,L,\mu_r)\,\frac{\partial n_r}{\partial x}}$ models the effects of environment-induced changes in MCT1 expression mediated by lactate-associated signalling pathways, i.e. SPCs~\rev{(cf. assumption {\bf A3})};
\item the diffusion term $\displaystyle{\Phi(G,L)\,\frac{\partial^2 n_r}{\partial x^2}}$ models the effects of \rev{fluctuations in MCT1 expression due to epigenetic changes, i.e. FECs~(cf. assumption {\bf A4})};
\item the reaction term $\displaystyle{R(x,G,L,\rho_r) \, n_r}$ models the effects of cell proliferation and death under environmental selection on MCT1 expression~\rev{(cf. assumption {\bf A5})}.
\end{itemize}
The modelling strategies underlying the PIDE~\eqref{ePDEn} are detailed in the following. 
%\begin{equation}
%\label{ePDEn}
%\displaystyle{\frac{\partial n}{\partial t}  \; - \underbrace{ \Phi(t)\,\frac{\partial^2 n}{\partial x^2}}_{\substack{\mbox{\scriptsize{spontaneous changes in MCT1 expression}}\\\mbox{\scriptsize{due to non-genetic instability}}}} \;\; +  \underbrace{ \Psi(t)\,\frac{\partial n}{\partial x}}_{\substack{\mbox{\scriptsize{environment-induced changes in MCT1 expression}}\\\mbox{\scriptsize{mediated by lactate-associated signalling pathways}}}} = \; \underbrace{R\big(x,G(t),L(t),\rho(t)\big) \, n,}_{\text{natural selection}}  \quad x \in \mathbb{R}},
%\end{equation} 

\subsubsection*{Modelling cell proliferation and death under environmental selection on MCT1 expression} 
%, under the environmental conditions given by the glucose concentration $G(t)$, the lactate concentration $L(t)$ and the cell number $\rho_r(t)$,
The fitness of cells with rescaled level of MCT1 expression $x$ at time $t$ is modelled by the function 
\begin{equation}\label{eR}
R\big(x,G,L,\rho_r\big) = p(x,G,L)  \;-\; d \, \rho_r.
\end{equation}
Definition~\eqref{eR} translates in mathematical terms to the following biological ideas: all else being equal, cells die due to intracellular competition at rate $d \, \rho_r$, with the parameter $d > 0$ being related to the carrying capacity of the {\it in vitro} system in which the cells are contained~\rev{(cf. assumption {\bf A5})}; cells with the rescaled level of MCT1 expression $x$ proliferate and die under environmental selection on MCT1 expression at rate $p(x,G,L)$ (i.e. the function $p$ is a net proliferation rate). Based on the considerations and assumptions introduced so far, we define
\begin{equation}\label{ep1}
p\big(x,G,L\big) = p_G\big(x,G\big) \; + p_L\big(x,G,L\big) 
\end{equation}
with
\begin{equation}\label{ep2}
p_G\big(x,G\big) =  \gamma_G \, U_G(G)\, \left(1-x^2\right), \quad p_L\big(x,G,L\big) = \gamma_L  \, U_L(G,L) \,  \left[1-\left(1-x\right)^2\right],
\end{equation}
\begin{equation}\label{eu}
U_G(G) = \frac{G^m}{(\alpha_G)^m + G^m} \quad \text{and} \quad  U_L(G,L) = \left(1 - H(G-G^*) \right)\, \frac{L^c}{(\alpha_L)^c+L^c}.
\end{equation}
The function $p_G$ models the net rate of cell proliferation via glycolysis and the function $p_L$ models the net rate of cell proliferation via lactate reuse. The fact that these functions are negative for values of $x$ sufficiently far from $0$ (i.e. the rescaled level of MCT1 expression endowing cells with the highest rate of proliferation via glycolysis) and $1$ (i.e. the rescaled level of MCT1 expression endowing cells with the highest rate of proliferation via lactate reuse when glucose is scarce) captures the idea that cells with less fit levels of MCT1 expression are driven to extinction by environmental selection~\rev{(cf. assumptions {\bf A1}, {\bf A2} and {\bf A5})}. Moreover, the functions $U_G$ and $U_L$ model glucose and lactate uptake, respectively \rev{(cf. assumptions~\rev{{\bf A6}} and {\bf A7})}. In the definitions given by Eqs.~\eqref{ep2} and~\eqref{eu}: 
\begin{itemize}[leftmargin=*]
\item $H(G-G^*)$ is the Heaviside step function centred at the threshold %physiological
 level of glucose $G^*$ (i.e. the level of glucose above which cells stop taking lactate from the extracellular environment and reusing it to produce energy for fuelling their proliferation), that is,
\begin{equation}\label{eheavi}
H(G-G^*) = 
\begin{cases}
0, \quad \text{if } G < G^*
\\
1, \quad \text{if } G \geq G^*;
\end{cases}
\end{equation}
\item $\gamma_G>0$ and $\gamma_L>0$ are the maximum rates of cell proliferation via glycolysis and lactate reuse; 
\item $\alpha_G>0$ and $\alpha_L>0$ are the glucose and lactate concentrations at half receptor occupancy~\cite{santillan2008use};
\item $m>0$ and $c>0$ are the Hill coefficients for glucose and lactate ligand-receptor dynamics~\cite{santillan2008use}. 
\end{itemize}
Under the definitions given by Eq.~\eqref{ep2}, after a little algebra, the definition given by Eq.~\eqref{ep1} can be rewritten as
\begin{equation}\label{eprev}
p\big(x,G,L\big) = a(G,L)-b(G,L) \left(x-X(G,L) \right)^2
\end{equation}
with
\begin{equation}\label{eab}
a(G,L) = \gamma_G \,U_G(G) + \, \frac{\displaystyle \big( \gamma_L  \, U_L(G,L)\big)^2}{\displaystyle \gamma_G \, U_G(G) + \gamma_L  \, U_L(G,L)}, \quad b(G,L)=  \gamma_L  \, U_L(G,L)+ \gamma_G \,U_G(G)
\end{equation}
and
\begin{equation}\label{eX}
X(G,L)= \frac{\displaystyle \gamma_L  \, U_L(G,L)}{\displaystyle \gamma_G \,U_G(G) + \gamma_L  \, U_L(G,L) }.
\end{equation}
Under the environmental conditions defined by the concentrations of glucose $G$ and lactate $L$: $X(G,L)$ represents the fittest rescaled level of MCT1 expression; $a(G,L)$ is the corresponding maximum fitness (i.e. the proliferation rate of cells exhibiting the fittest rescaled level of MCT1 expression); $b(G,L)$ can be seen as a nonlinear selection gradient that quantifies the strength of environmental selection on MCT1 expression. 

\begin{remark}
\label{rem:X}
Note that the definition given by Eq.~\eqref{eX} implies that $0 \leq X(G,L) \leq 1$ for all $G \geq 0$ and $L \geq 0$. In particular, coherently with the considerations and assumptions introduced in Sec.~\ref{sec:model:prel}, under this definition and the definitions given by Eq.~\eqref{eu}, we have that if $G=0$ then $X(G,L)=1$ for any $L>0$, while if $G\geq G^*$ then $X(G,L)=0$ for any $L \geq 0$.
\end{remark}

\begin{remark}
\label{rem:Y}
The fittest level of MCT1 expression, $Y(G,L)$, is obtained from the definition of the fittest rescaled level of MCT1 expression, $X(G,L)$, given by Eq.~\eqref{eX} through the change of variable defined by Eq.~\eqref{def:x}, that is,
\begin{equation}\label{eY}
Y(G,L) = y_L + X(G,L) (y_H - y_L) = y_L +  \frac{\displaystyle \gamma_L  \, U_L(G,L)}{\displaystyle \gamma_G \,U_G(G) + \gamma_L  \, U_L(G,L) } (y_H - y_L).
\end{equation}
\end{remark}

\subsubsection*{Modelling  \rev{FECs and SPCs}  in MCT1 expression}
The rate of \rev{FECs in MCT1 expression} is modelled by the function 
\begin{equation}\label{ephi}
\Phi(G,L)=\beta\big(1+\zeta\,U_L(G,L)\big),
\end{equation}
where the  lactate uptake function $U_L$ is defined via Eq.~\eqref{eu}. Under the definition given by Eq.~\eqref{ephi}, the minimum rate of FECs, $\beta>0$, is increased proportionally to lactate uptake with constant of proportionality $\zeta>0$~\rev{(cf. assumption {\bf A4})}. This translates in mathematical terms to the idea that, since lactate has been shown to be responsible for histone modifications~\cite{bergers2021metabolism,zhang2019metabolic}, the \rev{rate of FECs in MCT1 expression} may be enhanced under glucose-deprivation. 

Moreover, the rate of environment-induced changes in MCT1 expression mediated by lactate-associated signalling pathways~\rev{(cf. assumption {\bf A3})} is modelled by the function
\begin{equation}\label{epsi}
\Psi(G,L,\mu_r) = \Psi^+(G,L) - \Psi^-(G,\mu_r),
\end{equation}
with
\begin{equation}\label{epsi2}
\Psi^+(G,L)  = \lambda_{L} \, U_L(G,L) \quad \text{and} \quad \Psi^-(G,\mu_r)= \lambda_{G} \, H(G-G^*) \, \left(\mu_r\right)_+ \,,
\end{equation}
where the lactate uptake function $U_L$ is defined via Eq.~\eqref{eu}, while $H(G-G^*)$ is the Heaviside step function defined via Eq.~\eqref{eheavi}. The definitions given by Eqs.~\eqref{epsi} and~\eqref{epsi2} translate in mathematical terms to the idea that environment-induced changes mediated by lactate-associated signalling pathways lead to: an increase in MCT1 expression at rate $\Psi^+$, which is proportional to lactate uptake, under glucose deprivation (i.e. when $G < G^*$); to a decrease in MCT1 expression at rate $\Psi^-$ when the glucose level is sufficiently high (i.e. when $G \geq G^*$). The parameters ${\lambda_{L}}>0$ and  ${\lambda_{G}}>0$ model the corresponding maximum rates of environment-induced increase and decrease in MCT1 expression. Moreover, the dependence of $\Psi^-$ on $\left(\mu_r\right)_+ = \max \left\{0, \mu_r\right\}$ captures the fact that interruption of lactate-associated signalling pathways may occur when $G \geq G^*$ if the mean rescaled level of MCT1 expression of the cells is below the fittest level $x=0$ (cf. Remark~\ref{rem:X}).

\subsection{Glucose and lactate dynamics}
%\paragraph{Dynamics of abiotic factors.} 
%where  models the maximum rate of consumption of glucose by the cells, and the glucose uptake While cell proliferation via glycolysis is phenotype dependent as modelled by $p_G$ in Eq.\eqref{ep2},  that the rate of glucose consumption by the cells is independent of their phenotypic state. 
%On the other hand, we let the evolution of the lactate concentration $L(t)$ be governed by the following ODE
%\begin{equation}\label{eL}
%\frac{dL}{dt} =  \eta_L \, \frac{G}{\alpha_G+G} \, \rho(t) \,-k_L\, \frac{L}{\alpha_L+L} \, \int_\mathbb{R} \left[1-(1-x)^2\right]\,n\,{\rm d}x \,,
%\end{equation}
%where $\eta_L$ models the maximum rate of lactate production by the cells and $k_L$ models the maximum rate of consumption of lactate by the cells. In analogy with equation~\eqref{eG}, in equation~\eqref{eL} we assume the rate of lactate production via glycolysis is independent of the cell's phenotypic state. On the other hand, we assume the rate of lactate consumption by the cells be phenotype dependent and let the maximum consumption rate is achieved by cells in the phenotypic state $x=1$ coherently with the model assumptions introduced in the previous section.
% in Eq.\eqref{eL} we model lactate release in analogy with Eq.\eqref{eG}, with the parameter $k_L>0$ modelling the maximum rate of lactate production by the cells, which is assumed to be independent of the cell's phenotypic state. 
The dynamic of the glucose concentration $G(t)$ is governed by the following ordinary differential equation (ODE)
\begin{equation}\label{eG}
\frac{{\rm d}G}{{\rm d}t} =-k_G \, U_G(G)\, \rho_r(t) \,,
\end{equation}
where the glucose uptake function $U_G$ is defined via Eq.\eqref{eu}. The ODE~\eqref{eG} relies on the assumption that glucose is consumed by the cells at a rate proportional to glucose uptake~\rev{(cf. assumption {\bf A6})}, with constant of proportionality $k_G>0$. Moreover, the dynamic of the lactate concentration $L(t)$ is governed by the following ODE
\begin{equation}\label{eL}
\frac{{\rm d}L}{{\rm d}t} =  k_L \, U_G(G) \, \rho_r(t) \,-\eta_L \, \int_\mathbb{R}\big(p_L(x,G,L)\big)_{+}\,n_r(t,x)\,{\rm d}x \,,
\end{equation}
where $\left(p_L\right)_+ = \max \left\{0, p_L\right\}$, with $p_L$ being the function that models the net rate of cell proliferation via lactate reuse, which is defined via Eq.~\eqref{ep2}. Based on earlier studies indicating that most tumours release lactate in quantities linearly related to glucose consumption~\cite{vaupel1989blood}, and coherently with the way in which the effect of glucose consumption is incorporated into the ODE~\eqref{eG}, the ODE~\eqref{eL} relies on the assumption that lactate is produced by the cells at a rate proportional to glucose uptake~\rev{(A6)}, with constant of proportionality $k_L>0$. Moreover, the ODE~\eqref{eL} relies on the additional assumption that lactate is absorbed only by the cells whose rescaled levels of MCT1 expression make them capable of reusing lactate to produce the energy required for their proliferation when glucose is scarce (i.e. cells with rescaled levels of MCT1 expression $x$ corresponding to positive values of $p_L(x,G,L)$), which absorb lactate at a rate proportional to their net proliferation rate~\rev{(cf. assumption {\bf A6})}, with constant of proportionality (i.e. conversion factor for lactate consumption) $\eta_L>0$.
%On the other hand, we make the assumption that  phenotypic states endowing them with the ability to proliferate in the absence of glucose, i.e. phenotypic variants corresponding to values of $x$ for which the fitness $p_L$, defined in Eq.\eqref{ep2}, is positive. Therefore, we let the rate at which cells consume lactate be proportional to the function $(p_L)_+$, where $(\cdot)_+\equiv\max(0,\cdot)$, and introduce the conversion factor $\eta_L>0$.
%On the other hand, we assume the rate of lactate consumption by the cells to be phenotype dependent, as dictated by the function $p_L$ defined in Eq.\eqref{ep2}, and introduce the conversion factor $\eta_L>0$.

\subsection{Initial conditions}
Informed by the experimental data reported in Fig.2(A) in the Main Manuscript, we define the initial MCT1 expression distribution of MCF7-sh-WISP2 cells as~\rev{(cf. assumption {\bf A8})}
\begin{equation}\label{IC:ny}
n(0,y) = n_{0}(y) \quad \text{with} \quad  n_{0}(y) = \, \frac{\rho_{0}}{\sqrt{2\pi\sigma_{0}^2}} \, \exp \left( -\frac{(y-\mu_{0})^2}{2\sigma_{0}^2}  \right),
\end{equation}
where the initial cell number, $\rho_{0}$, the initial mean level of MCT1 expression, $\mu_{0}$, and the related variance, $\sigma^2_{0}$, are defined as
\begin{equation}\label{IC:parsy}
\rho_0=1.5\times 10^6, \quad \mu_0 = 15.57 \times 10^3, \quad \sigma^2_0= 8 \times 10^6.
\end{equation}
Hence, under the change of variable defined by Eq.~\eqref{def:x}, the initial rescaled MCT1 expression distribution of MCF7-sh-WISP2 cells is
\begin{equation}\label{IC:n}
n_{r}(0,x) = n_{r0}(x) \quad \text{with} \quad  n_{r0}(x) = \frac{\rho_{r0}}{\sqrt{2\pi\sigma_{r0}^2}} \, \exp \left( -\frac{(x-\mu_{r0})^2}{2\sigma_{r0}^2}  \right),
\end{equation}
where (cf. the relations given by Eq.~\eqref{erhomusigmaxy})
\begin{equation}\label{IC:pars}
\rho_{r0} = \rho_{0}, \quad \mu_{r0} = \frac{\mu_0-y_L}{y_H-y_L}, \quad \sigma^2_{r0}= \dfrac{\sigma^2_0}{(y_H-y_L)^2}.
\end{equation}
%\rho_{r0}=\rho_0, \quad
\begin{remark}
Note that under the relations given by Eq.~\eqref{IC:pars}, we have that $n_{r0}(x)$ in Eq.~\eqref{IC:n} and $n_{0}(y)$ in Eq.~\eqref{IC:ny} are related by $n_{r0}(x)=(y_H-y_L)n_{0}(y)$, i.e. we retrieve the relation between $n_{r}(t,x)$ and $n(t,y)$ introduced in Sec.\ref{sec:model:prel}.
\end{remark}
Moreover, in order to match the experimental data reported in Fig.1(B) in the Main Manuscript, we define the initial concentrations of glucose and lactate, respectively, as
\begin{equation}\label{IC:GL}
G(0) =G_0, \quad L(0)=L_0 ,
\end{equation}
\rev{with $G_0 =5.52$mM and $L_0 =1.67$mM being the average values recorded at day 0 of the glucose-deprivation experiments.}

\subsection{Parameter values}
The values of the model parameters obtained, through the calibration procedure detailed in Section~\ref{sec:methods:calib}, \rev{using data from `glucose-deprivation' experiments conducted on MCF7-sh-WISP2 cells} are summarised in Tab.~\ref{tab:parameters2}\rev{, with the associated bootstrap statistics reported in Tab.~\ref{tab:bootstrap}.} 
\rev{Moreover, Tab.~\ref{tab:bootstrap:ALT} displays the parameter values, along with the corresponding bootstrap statistics, recovered by repeating the calibration procedure using data from both `glucose-deprivation' and `rescue' experiments, which are employed to obtain the numerical results in Fig.~\ref{Fig:calib_ALT}.}

%\newpage
\section{Further details of materials and methods}\label{sec:methods}

\subsection{\textit{In vitro} experiments}\label{sec:methods:exp}
%To investigate the impact of microenvironment-induced metabolic stress on aggressive cancer cells, 
Two breast cancer cell lines are considered: MCF7 (human breast cancer cell line, epithelial phenotype) and MCF7-sh-WISP2 (MCF7 cells invalidated for WISP2 by sh-RNA plasmid, mesenchymal phenotype)~\cite{ferrand2014loss,fritah2008role}.
%the Epithelial-Mesenchymal Transition-induced breast cancer cell line 

\subsubsection{Cell proliferation and death}
Cells were routinely maintained in Dulbecco's modified Eagle medium containing 4.5g/l of glucose supplemented with 10\% fetal bovine serum (FBS), L-Glutamine, and antibiotics. For assessing cell proliferation and death, cells were cultured for four days in a medium initially containing 1g/l of glucose. Viable cells were identified via trypan blue exclusion and counted using Beckman Coulter, while cell death was quantified via annexin V-FITC apoptosis staining.

\subsubsection{Flow cytometry analysis}
%cultured in media \red{\bf INITIALLY?!} containing 4.5g/l or 1g/l of glucose 
Cells were stained with fluorochrome-conjugated monoclonal antibodies against human MCT1-FITC (Beckman Coulter) at room temperature in the dark for 20 minutes. Cells were then washed with PBS containing 0.5\% serum and flow cytometry analysis was carried out. The labelled cells were analysed on a FACS Gallios (Beckman Coulter) and data analysis was performed using the Kaluza software.

\subsubsection{Immunofluorescence staining}
Cells were plated on chamber slides and fixed in 4\% paraformaldehyde. Cells were stained with anti-MCT1 antibody and secondary anti-Rabbit FITC-conjugated antibody (Jackson ImmunoResearch, Cambridgeshire, UK). After immunolabelling, cells were washed, stained with 1${\mu}$g/mL DAPI (Sigma), and observed by fluorescence microscopy (BX61, Olympus). 

\subsubsection{Real-time RT-qPCR}
Total RNA was extracted from cell samples using the TRIzol\textsuperscript{\textregistered} RNA purification reagent. RNA quantity and purity were assessed by using a Spectrophotometer DS-11 (Denovix, Wilmington, DE, USA). One microgram of total RNA from each sample was reverse transcribed, and real-time RT-qPCR measurements were performed as described in~\cite{ferrand2014loss}, using an apparatus Aria MX (Agilent Technologies, Santa Clara, CA, USA) with the corresponding SYBR\textsuperscript{\textregistered} Green kit, according to the PROMEGA manufacturer's recommendations. 

\subsection{Model calibration with experimental data} \label{sec:methods:calib}
Experimental data on MCF7-sh-WISP2 cells are used to carry out model calibration through a likelihood-\rev{maximisation} method~\cite{linden2022bayesian,martinez2014bayesopt,myung2003tutorial,spilker2001evaluation}. 
\rev{The likelihood of each parameter set is defined implementing statistical measures obtained from data of replicate experiments, to account for average behaviour.}
The optimal parameter set (OPS) is obtained by \rev{minimising the weighted sum of squared residuals, which corresponds to} maximising the likelihood, through an iterative process described in Sec.\ref{sec:methods:validation}. At each iteration, we solve numerically the model \rev{comprising the PIDE-ODE system~\eqref{ePDEn}, \eqref{eG}, \eqref{eL} subject to the initial conditions defined via Eqs.~\eqref{IC:n}-\eqref{IC:GL}}, using methods analogous to those described in Sec.\ref{sec:methods:numerics}.  
\rev{Uncertainty quantification is conducted using a bootstrapping algorithm~\cite{thompson1996tutorial,zhu1997making}, described in Sec.\ref{sec:methods:bootstrapping}, to obtain emprical 95\% confidence intervals of each parameter in the OPS. }
The {\sc Matlab} source codes along with the data used for model calibration have been made available on GitHub\footnote{https://github.com/ChiaraVilla/AlmeidaEtAl2023Evolutionary}. The obtained OPSs are reported in Tab.S1. \rev{The bootstrapping statistics are reported in Tab.S2, and the empirical probability distributions of the parameters obtained via the bootstrapping procedure are plotted in Fig.\ref{Fig:bootstrap}.}

\setcounter{equation}{24}

\subsubsection{Calibration procedure}\label{sec:methods:validation}

\paragraph{\rev{The experimental data used for model calibration.}} \rev{Let
$S_D=\{{u}^{i,k}_D,\, i=1,\ldots,M\,,k=1,\ldots,K\}$ indicate the set of $M\times K$ data points ${u}^{i,k}_D$, i.e. the $M$ experimentally obtained summary statistics from each of the $K$ replicate experiments. From these, the average $\bar{u}^i_D$ and standard deviation $s^i_D$ of each summary statistic ($i=1,\ldots,M$) are calculated using the standard formulas
\begin{equation}\label{def:meansd}
\bar{u}^i_D = \frac{1}{K} \sum_{k=1}^K {u}^{i,k}_D \,, \qquad s^i_D = \frac{1}{\sqrt{K-1}}\sum_{k=1}^K | {u}^{i,k}_D - \bar{u}^i_D | \,.
\end{equation}
}

\paragraph{\rev{The likelihood and the weighted sum of square residuals.}} \rev{
Let ${S}_P \in \Omega \subset \mathbb{R}^N_{\geq 0}$ indicate the set of parameter values in the $N$-dimensional and bounded parameter space $\Omega$.} %, and let $\{\hat{u}^i_P, i=1,...,M\}$ indicate the summary statistics predicted by the model under ${S}_P$.  } 
Assuming Gaussian measurement noise with zero mean~\cite{linden2022bayesian,spilker2001evaluation}, the likelihood \rev{of ${S}_P$ is given by
\begin{equation}\label{def:ld}
\mathcal{L}({S}_P) = \mathds{P}({S}_D\, |\, {S}_P) =   \prod_{i=1}^{M}  \, \frac{1}{\sqrt{2\pi}s^i_D}\, \exp \left( -\frac{(\bar{u}^i_D-{u}^i_P)^2}{2\rev{(s^i_D)}^2}  \right)\,,
\end{equation}
where $\{{u}^i_P, i=1,\ldots,M\}$ indicates the summary statistics predicted by the model under ${S}_P$.  
Since PIDE models provide a mean-field representation of the underlying cellular dynamics, the Gaussian likelihood for each summary statistic ${u}^i_P$ is centred at the experimental data average $\bar{u}^i_D$. The variance of the Gaussian error measurement --  and thus of the Gaussian likelihood -- for each summary statistic ${u}^i_P$ is assumed to be equal to the variance of the experimental data $(s^i_D)^2$, in order to account for the heteroscedasticity~\cite{spilker2001evaluation} suggested by experimental observations. Then the logarithm of the likelihood~\eqref{def:ld} is
\begin{equation}\label{def:logld}
\log \mathcal{L}({S}_P) = C_D - R({S}_P) \,,
\end{equation}
where $\displaystyle{C_D = - \sum_{i=1}^{M} \log \left(\sqrt{2\pi}s^i_D\right)}$ is a constant and 
\begin{equation}\label{def:wssr}
R({S}_P) = \sum_{i=1}^{M} \frac{(\bar{u}^i_D-{u}^i_P)^2}{2\rev{(s^i_D)}^2} 
\end{equation}
is the weighted sum of squared residuals, whereby higher/lower variability in the observed data will result in lighter/heavier weights~\cite{spilker2001evaluation}.
}

\paragraph{Likelihood-\rev{maximisation} method.} 
From Bayes Theorem, the posterior distribution of $\rev{S}_P$ given the data set $\rev{S}_D$ -- i.e. the distribution $\mathcal{P}\left( \rev{S}_P \,|\, \rev{S}_D \right)$ -- is such that
\begin{equation}
\mathcal{P}\left( \rev{S}_P \,|\, \rev{S}_D \right) \propto {\mathcal{P}(\rev{S}_P)\,\mathcal{L}(\rev{S}_P)},
\end{equation}
where $\mathcal{P}(\rev{S}_P)$ is the prior distribution of $\rev{S}_P$, and $\mathcal{L}(\rev{S}_P)$ is the likelihood of $\rev{S}_P$~\cite{myung2003tutorial}. 
Due to little knowledge on the prior distribution of the parameters, we assume each of them to be uniformly distributed in a bounded domain, and seek $\rev{S}_P$ maximising the likelihood.  
In practice, \rev{for numerical reasons~\cite{spilker2001evaluation},} we search for the minimum point of the \rev{weighted sum of squared residuals~\eqref{def:wssr},  which corresponds to the maximum of the log likelihood~\eqref{def:logld},} in the domain assumed for the prior distributions, exploiting the in-built {\sc Matlab} function \texttt{bayesopt}, which is based on Bayesian Optimisation~\cite{martinez2014bayesopt}. 
Due to little knowledge on the parameter values, we take the assumed domain of the prior distributions of most parameters to span several orders of magnitude (see details provided below). These ranges of values are then iteratively updated to ensure that we obtain a good agreement with the experimentally observed MCT1 expression distributions of MCF7-sh-WISP2 cells reported in Fig.2(A) of the Main Manuscript. 

\paragraph{Ranges of parameter values considered in the calibration algorithm.} 
\rev{We consider values of the maximum rates of proliferation in the range ${\gamma}_G,\,{\gamma}_L\in[0.001,3]$ /day,  to ensure we capture all values recorded in~\cite{freischel2021frequency} in a variety of environmental conditions for two breast cancer cells lines.
 Considering this and values of \textit{in vitro} tumour carrying capacities in the range of $10^5-10^7$ cells, as observed here (cf. Fig.1(A) of the Main Manuscript) and in~\cite{mendoza2012mathematical},  the value of the death rate due to competition for space is taken in the range $d\in[10^{-8},10^{-5}]$ /day /cell.} % covering also values of $d$ taken in other spatially homogeneous phenotype structured PDE models in the literature -- e.g.~\cite{bosque2023metabolic,lorenzi2020discrete}.}
\rev{ We take values of the glucose consumption and lactate production rates to be in the range ${\kappa}_G,\,{\kappa}_L\in[10^{-7},10^{-5}]$ mM /day /cell,  considering the values of glucose consumption rates in~\cite{mcgillen2014glucose,mendoza2012mathematical} and those of the tumour carrying capacity range introduced above, knowing that values of ${\kappa}_G$ and ${\kappa}_L$ are of the same order of magnitude~\cite{mcgillen2014glucose}.}
%\rev{In~\cite{mcgillen2014glucose} they note lactate production rates are about twice as high as glucose consumption rates during glycolysis.
%In~\cite{mendoza2012mathematical} the rates of glucose consumption, for cultures averaging $8\time10^5$ cells, were around 1 mM/day, corresponding to a $\kappa_G$ of a little over $1\times10^{-6}$ mmol/l /cell /day. In~\cite{molavian2009fingerprint} they considered glucose consumption rates to be around $0.002$ mmol/l /s per cell, corresponding to a little over $170$ mmol/l /day. (!!) 
% In~\cite{mcgillen2014glucose} they also consider the production rate of glucose in the range $[1,10^5]$ mmol/l /day. 
% If we consider all these to bemultiplied by the cell density normalised by the carrying capacity, we have to divide the parameters by the carrying capacity, which we consider in the range $\time10^5-10^7$, giving us $\kappa_G$ in the range $10^{-7}-10^{-5}$.
%}
\rev{In the absence of empirically-informed estimates, the value of the conversion factor for lactate consumption is taken in $\eta_L\in[10^{-12},10^{-4}]$ mM /cell, covering a wide range of orders of magnitude, including those used in~\cite{fiandaca2021mathematical} and references therein.  } 
The values of the Hill coefficients ${c}$ and ${m}$ are assumed to be in the interval $[0.9,4]$, since most studies \rev{-- e.g.~\cite{fiandaca2021mathematical,mcgillen2014glucose,molavian2009fingerprint,phipps2015microscale} --} assume Michaelis Menten kinetics (i.e. Hill coefficient equal to 1) but recent works assume positive cooperative binding for glucose uptake~\cite{celora2021phenotypic}. 
\rev{The glucose and lactate concentrations at half receptor occupancy are taken to have values in the range ${\alpha}_G,\,{\alpha}_L \in [0.01,10]$ mM as in~\cite{mcgillen2014glucose,mendoza2012mathematical}.} 
\rev{We take value of the minimum rate of FECs in MCT1 expression $\beta\in[10^{-4},10^{-1}]$ /day, covering the range of values used in~\cite{ardavseva2020comparative,celora2021phenotypic,lorenzi2020discrete} and references therein, and consider the value of the lactate-dependency coefficient $\zeta\in[0,100]$ to avoid an unrealistic blow up of the rate of FECs.  In the absence of further knowledge, we take the maximum rates of SPCs in MCT1 expression to have values in the range ${\lambda}_L,\, {\lambda}_G \in [0,1]$ /day, which includes the range of values considered for $\beta$ as well as the phenotypic drift magnitude considered in previous phenotype-structured PIDE models for cancer evolution~\cite{celora2021phenotypic}. } 
% \cite{ardavseva2020comparative} $\beta\in[10^-2,10^-3]$ /day
For consistency with the mean MCT1 expression levels recorded in the experiments (cf. Fig.2(B) of the Main Manuscript), we consider $y_L\in[0,15 \times 10^3]$ and $y_H\in[35 \times 10^3,100 \times 10^3]$. 
Finally, the value of the threshold glucose concentration for lactate uptake \rev{is assumed to be above physiological levels, i.e. $G^*>5.5$mM, when calibrating the model with data from `glucose-deprivation' experiments, and later taken} in the interval of experimentally-considered glucose concentrations ${G}^* \in [0,\rev{25}]$~\rev{mM to test our initial assumption by calibrating the model with data from both `glucose-deprivation' and `rescue' experiments.}

\subsubsection{\rev{Uncertainty quantification}}\label{sec:methods:bootstrapping}
\rev{
Given the little amount of data available for model calibration, we make use of a bootstrapping algorithm~\cite{thompson1996tutorial,zhu1997making} to quantify uncertainty in the maximum likelihood estimates obtained from fitting the model to the average values of each summary statistic.  
 The algorithm is composed of the following steps:
\begin{enumerate}
%\item Create a new data set $\tilde{S}_D^j = \{u^{i}_D,\, i=1,...,M \} := \{u^{i,k_i}_D,\, i=1,...,M,\, k_i\in\{1,...,K\}\}$ by randomly resampling with replacement from the original dataset, i.e. by selecting the value of each of the $M$ summary statistic randomly from one of the $K$ replicate experiments.
\item Create the $j^{th}$ bootstrap data set $\tilde{S}_D^j = \{u^{ij}_D,\, i=1,\ldots,M \} $ by randomly resampling with replacement from the original dataset, i.e. by selecting the value of each of the $M$ data points randomly from one of the $K$ replicate experiments ($u^{ij}_D\equiv u^{i,k_i}_D$,  where $k_i\in\{1,\ldots,K\}$ for each $i=1,\ldots,M$).
%\item Find the $j^{th}$ bootstrap optimal parameter set ${S}^j_B$ to match the resampled data set $\tilde{S}_D^j$, by repeating the calibration procedure described in Section~\ref{sec:methods:calib} maximising -- in place of the likelihood~\eqref{def:ld} -- the bootstrap likelihood
%$$
%\mathcal{L}^j_B({S}_P) = \mathds{P}(\tilde{S}_D^j\, |\, {S}^j_B) =   \prod_{i=1}^{M}  \, \frac{1}{\sqrt{2\pi}s^i_D}\, \exp \left( -\frac{(\bar{u}^i_D-{u}^i_P)^2}{2\rev{(s^i_D)}^2}  \right)\,.
%$$
\item Find the $j^{th}$ bootstrap optimal parameter set ${S}^j_B$ maximising the bootstrap likelihood $\mathcal{L}^j_B({S}_P)$, i.e.
\begin{equation}\label{def:lh:boot}
\mathcal{L}^j_B({S}_P) = \mathds{P}({S}_P\, |\, \tilde{S}_D^j) =   \prod_{i=1}^{M}  \, \frac{1}{\sqrt{2\pi}s^i_D}\, \exp \left( -\frac{({u}^{ij}_D-{u}^i_P)^2}{2{(s^i_D)}^2}  \right)\,,
\end{equation}
by repeating the calibration procedure described in Sec.~\ref{sec:methods:validation} to match the data set $\tilde{S}_D^j$.
\item Repeat Points 1 and 2 for $j=1,\ldots,J$ to obtain $J$ bootstrap samples of the maximum likelihood estimate of each parameter -- say,  $\hat{\theta}^j_B$.
\item Calculate bootstrap statistics, such as bootstrap mean $\bar{\theta}_B$, standard deviation $s^\theta_B$, and bias of the maximum likelihood estimate obtained during the main calibration procedure (denoted by $\hat{\theta}_{mle}$)
\begin{equation}\label{def:mean:boot}
\bar{\theta}_B=\frac{1}{J}\sum_{j=1}^J \hat{\theta}^j_B, \quad s^\theta_B= \frac{1}{\sqrt{J-1}}\sum_{j=1}^J | \hat{\theta}^j_B - \bar{\theta}_B |, \quad \text{BIAS}= \hat{\theta}_{mle} - \bar{\theta}_B,
\end{equation}
as well as the empirical 95\% confidence interval, i.e.  the range of values containing the intermediate 95\% bootstrap sample values (removing the first and last 2.5\% quartiles). Note that positive/negative bias suggests over/under-estimation of the parameter in the optimal parameter set of the main calibration procedure.
\end{enumerate}}

\subsection{Numerical methods for the simulations of the mathematical model}\label{sec:methods:numerics}
%To obtain the summary statistics $u^i_P$,
%To simulate the mathematical model 
Numerical solutions of the PIDE-ODE system~\eqref{ePDEn}, \eqref{eG}, \eqref{eL} subject to the initial conditions defined via Eqs.\eqref{IC:n}-\eqref{IC:GL} are constructed using a uniform discretisation of the interval $[0,\rm{T}]$, chosen as computational domain of the variable $t$, with uniform step $\Delta t=10^{-5}$, and a uniform discretisation of the interval $[-3,3]$, chosen as computational domain of the variable $x$, with uniform step $\Delta x=0.002$. Suitable values of the final time of simulations $\rm{T}>0$ are chosen depending on the scenarios under study. 

To solve numerically the PIDE~\eqref{ePDEn}, we impose the following zero-flux boundary conditions
$$
\begin{cases}
\Psi(G(t), L(t))\,n_r(t,-3) - \Phi(G(t), L(t), \mu_r(t))\,\partial_x n_r(t,-3) = 0,
\\\\
\Psi(G(t), L(t))\,n(t,3) - \Phi(G(t), L(t), \mu_r(t))\,\partial_x n_r(t,3) = 0,
\end{cases}
\;\; \forall t \in(0,\rm{T}),
$$
which are implemented by means of first-order forward (at $x=-3$) and backward (at $x=3$) finite difference approximations. We make use of first-order forward difference approximation for the time derivative, second-order central difference approximation for the diffusion term, and a first-order upwind scheme to approximate the advection term. Integral terms are approximated by the corresponding left Riemann sums. Given the numerical values of $n_r(t,x)$, $\rho_r(t)$, $\mu_r(t)$ and $\sigma^2_r(t)$, the corresponding values of $n(t,y)$, $\rho(t)$, $\mu(t)$ and $\sigma^2(t)$ are obtained through the change of variable $n(t,y) = (y_H - y_L)^{-1} n_r(t,x)$ and the relations given by Eq.~\eqref{erhomusigmaxy}, respectively.

To solve numerically the ODEs~\eqref{eG} and~\eqref{eL}, we make use of first-order forward difference approximation for the time derivatives, while integral terms are approximated by the corresponding left Riemann sums.

\subsection{Optimal parameter sets obtained through model calibration}
\renewcommand{\thetable}{S\arabic{table}}
\setcounter{table}{0}
The OPSs $\hat{S}_P$ (\rev{with maximum likelihood estimates for each parameter indicated} up to 4 d.p.) for the mathematical model defined by the PIDE-ODE system~\eqref{ePDEn}, \eqref{eG}, \eqref{eL}, subject to the initial conditions defined via Eqs.~\eqref{IC:n}-\eqref{IC:GL}, in which both  \rev{FECs and SPCs}  in MCT1 expression are included (i.e. the model with $\Phi\not\equiv0$, $\Psi^{\pm}\not\equiv0$) and for reduced models in which only \rev{FECs} in MCT1 expression are included (i.e. the model with $\Phi\not\equiv0$, $\Psi^{\pm}\equiv0$) or only \rev{SPCs} in MCT1 expression are included (i.e. the model with $\Phi\equiv0$, $\Psi^{\pm}\not\equiv0$) are reported in Tab.~\ref{tab:parameters2}. The value of the \rev{weighted sum of squared residuals $R(S_P)$, defined via Eq.\eqref{def:wssr},} related to each parameter set is provided in the last row of Tab.~\ref{tab:parameters2}. The units of measure of the parameters are reported in the last column, where `-' is reported for dimensionless parameters.  \rev{The bootstrap statistics~\eqref{def:mean:boot} and empirical 95\% confidence interval of each parameter obtained during the uncertainty quantification procedure are reported in Tab.~S2, where the BIAS is calculated using the OPS of the full model, i.e. for the values listed in the third column of Tab~S1.  The bootstrap sampling distributions are plotted, along with bootstrap statistics and the OPS used to calculate the BIAS,  in Fig.~\ref{Fig:bootstrap}. Maximum likelihood estimates and bootstrap statistics obtained fitting data from `glucose-deprivation' and `rescue' experiments are reported in Tab.~\ref{tab:bootstrap:ALT}.}

\begin{table}[htb!]
\centering
%Optimal dimensional
\caption{\textbf{Optimal parameter sets $\hat{S}_P$ obtained through model calibration \rev{with data from `glucose deprivation' experiments. }} \rev{Note that since we assume $G^*>5.5$mM, both the exact value of this parameter and the value of the parameter $\lambda_G$ are not relevant for predicting dynamics under glucose deprivation and, therefore, they are not provided here -- estimates for the values of these parameters are provided in Tab.\ref{tab:bootstrap:ALT}. Bootstrap sampling distributions are plotted in Fig.~\ref{Fig:bootstrap}}}\label{tab:parameters2}
\footnotesize
\begin{tabular}{|c | c| c| c| c| c|} 
 \hline
%\parbox[c]{1.7cm}{\textbf{Parameter}} & \parbox[c]{2.5cm}{ \textbf{Main model}  }  & \parbox[c]{4.5cm}{ \textbf{Unit of measure}  } \\ [0.5ex] 
\textbf{Parameter} & \textbf{Biological meaning} & \makecell{\textbf{Model with} \\ \textbf{$\Phi\not\equiv0$, $\Psi^{\pm}\not\equiv0$}} & \makecell{\textbf{Model with} \\ \textbf{$\Phi\equiv0$, $\Psi^{\pm}\not\equiv0$}} & \makecell{\textbf{Model with} \\\textbf{$\Phi\not\equiv0$, $\Psi^{\pm}\equiv0$}} & \makecell{\textbf{Units}\\\textbf{of measure}} \\[0.5ex]
 \hline
   $y_L$ & \makecell{MCT1 level corresponding \\ to the maximum rate of \\ proliferation via glycolysis} &  \rev{4.1751} $\times 10^3$ &  \rev{6.3798} $\times 10^3$ &  \rev{9.9326} $\times 10^3$ & - \\ [0.5ex] %%2.9289
 \hline
   $y_H$  &  \makecell{MCT1 level corresponding \\ to the maximum rate of \\ proliferation via lactate reuse} &   \rev{49.5822} $\times 10^3$ &  \rev{71.3331} $\times 10^3$ &  \rev{48.6315} $\times 10^3$ & - \\ [0.5ex] %%69.6893
 \hline
 ${d}$ & \makecell{Rate of death due to  \\ intracellular competition} &   $ \rev{4.5232 \times 10^{-8}}$ & $ \rev{1.6174\times 10^{-8}}$ & $ \rev{1.4032\times 10^{-7}}$& /day /cell  \\ [0.5ex] %%  0.8913 0.7905 0.7511
 \hline
${\gamma}_G$   & \makecell{Maximum rate of  \\ proliferation via glycolysis} &     \rev{2.8898} &  \rev{2.8307} &  \rev{2.9924}& /day 	\\ [0.5ex] %%1.9898
 \hline
 ${\gamma}_L$ & \makecell{Maximum rate of \\ proliferation via lactate reuse} &  \rev{0.4278} &  \rev{0.1148} &  \rev{0.2921} &/day 	 \\ [0.5ex] %%1.9391
 \hline
 ${\alpha}_G$   & \makecell{Glucose concentration \\ at half receptor occupancy} &  \rev{2.7500} &  \rev{3.2362} &  \rev{3.6466} & g/l	  \\ [0.5ex] %%0.5937
 \hline
  ${\alpha}_L$ & \makecell{Lactate concentration \\ at half receptor occupancy} &  \rev{3.6131} &  \rev{5.9826} &  \rev{3.0933} & \rev{mM}	 \\ [0.5ex] %%0.5912 
 \hline
  $m$   & \makecell{Hill coefficient for glucose \\ ligand-receptor dynamics} &   \rev{1.0066} &  \rev{1.0140} & \rev{1.0294}  & -\\ [0.5ex] 
 \hline
  $c$ & \makecell{Hill coefficient for lactate \\ ligand-receptor dynamics} & \rev{1.9997} &  \rev{1.6730} &  \rev{2.2783}  & - \\ [0.5ex] 
 \hline 
%     $G^*$  & \makecell{Threshold level of glucose \\ above which lactate uptake stops} &  \rev{5.7999} &   \rev{5.3369} &  \rev{5.7060} & \rev{mM} \\ [0.5ex]  
% \hline
   ${\beta}$ & \makecell{Minimum rate of \rev{FECs} \\ in MCT1 expression} & $\rev{6.2992 \times 10^{-4}}$ & / & \rev{$0.0152$}  &  /day	 \\ [0.5ex] 
 \hline
${\zeta}$ & \makecell{Lactate-dependency coefficient \\ of the rate of \rev{FECs}\\ in MCT1 expression} & \rev{10.8609} & / & \rev{6.9453}& -  \\ [0.5ex] 
 \hline
  ${\lambda_{L}}$ & \makecell{Maximum rate of \rev{SPCs} \\ \rev{increasing MCT1 expression} } & $\rev{0.0894}$ & \rev{0.0905} & / & /day  \\ [0.5ex] 
 \hline
% ${\lambda_{L}}$ & \makecell{Maximum \rev{MCT1 expression} \\ \rev{increase rate due to SPCs} } & $\rev{4.5879 \times 10^{-2}}$ & \rev{0.1035} & / & /day  \\ [0.5ex] 
% \hline
% ${\lambda_{G}}$ & \makecell{Maximum rate of \\  \rev{transcriptionally-regulated} \\ decrease in MCT1 expression} & \rev{0.1012} & \rev{0.2161} & / & /day \\ [0.5ex] 
% \hline
  ${\kappa}_G$   &  \makecell{Rate of glucose \\ consumption}  & $\rev{2.4618}\times10^{-6}$ & $\rev{2.6930}\times10^{-6}$ & $\rev{2.8446}\times 10^{-6}$ &  \rev{mM} /day /cell 	\\ [0.5ex] 
 \hline
  ${\kappa}_L$  & \makecell{Rate of lactate \\ production}  & $\rev{4.3323}\times 10^{-6}$ &  $\rev{4.4738}\times 10^{-6}$ & $\rev{4.7253}\times 10^{-6}$ &  \rev{mM} /day /cell  \\[0.5ex]  
 \hline
  ${\eta}_L$  & \makecell{Conversion factor for \\  lactate consumption} & $\rev{8.1164\times 10^{-7}}$ & $ \rev{1.4316\times 10^{-6}}$ & $ \rev{3.0079\times 10^{-7}}$&  \rev{mM} /cell	   \\ [0.5ex] 
 \hline \hline
\rule{0pt}{10pt}
 \rev{$R(\hat{S}_P)$} & \rev{ \makecell{Weighted sum of  \\   squared residuals}} %\rev{Weighted sum of squared residuals} 
 & \rev{$84.1286$} & \rev{$113.8088$} & \rev{$195.4532$} & - \\ [0.5ex] 
 \hline
\end{tabular}
\end{table}	

%\newpage

\paragraph{\rev{Calibration results fitting data from `glucose-deprivation' experiments.}}
\rev{The maximum likelihood estimates of the parameters present in each model variation: (i) are consistent across models; (ii) are found in parameter ranges and orders of magnitude consistent with the current modelling and biological literature; and (iii) all provide a good qualitative agreement with the experimental data.  Furthermore, the bootstrap sampling distributions obtained via the uncertainty quantification procedure, conducted on the full model, are in agreement with the estimate of each parameter in the OPS, further supporting the validity of the OPS.  
In particular,  values in the OPS are consistently found in the interval $\bar{\theta}^i_B\pm s^i_B$ (cf.  green vertical lines and red error bars in Fig.\ref{Fig:bootstrap}), i.e. significantly close to the bootstrap mean. The only exceptions are found in the values of the parameters $\gamma_L$ and $d$, for both of which we recorded a relatively large negative bias, suggesting the values of these parameter were simultaneously overestimated in the main calibration algorithm. Nonetheless, the value of $\gamma_G$ is consistently higher than $\gamma_L$ -- of one or 2 orders of magnitude -- which is in line with the biologically coherent notion that cell proliferation via glycolysis is more efficient than via alternative metabolic pathways~\cite{keenan2015alternative}, with maximum net proliferation rates via glycolysis being amongst the largest values recorded for cancer cells in~\cite{freischel2021frequency}.  
Uncertainty quantification of the maximum likelihood estimates of the Hill coefficients $m$ and $c$ suggests mostly Michaelis-Menten dynamics are at play for glucose uptake, a result supported by many works in the literature --see, for instance, ~\cite{fiandaca2021mathematical,mcgillen2014glucose,molavian2009fingerprint,phipps2015microscale} --, and stronger positive cooperative binding for lactate uptake. This can be regarded as an additional evolutionary mechanism of cancer cells to survive glucose-deprivation once they acquire the ability to reuse lactate. Calibration results on the parameters $\beta$, $\lambda_L$, and $\zeta$ confirm the following trends: the rate of FECs in MCT1 expressions can become 10 times larger in the presence of a high concentration of lactate; FECs in MCT1 expression occur at a rate that is 2 orders (or 1 order, under high lactate concentrations) of magnitude smaller than that of SPCs. Finally, we remark that the maximum likelihood estimate for $\kappa_G$ is consistently about twice as large as $\kappa_L$, as supported by the literature~\cite{mcgillen2014glucose}.
}

\begin{table}[htb!]
\centering
%Optimal dimensional
\caption{\textbf{\rev{Bootstrap statistics~\eqref{def:mean:boot} ($J=200$) obtained for uncertainty quantification with data from `glucose deprivation' experiments.}} \rev{The BIAS is calculated with $\hat{\theta}_{mle}$ in the optimal parameter set reported in the third column of Tab.\ref{tab:parameters2}.}}\label{tab:bootstrap} %  procedure for the model with $\Phi\not\equiv0$, $\Psi^{\pm}\not\equiv0$, using $J=200$ bootstrap samples
\footnotesize
\renewcommand{\arraystretch}{1.5}
\rev{\begin{tabular}{|c | c| c| c| c| c|} 
 \hline
%\parbox[c]{1.7cm}{\textbf{Parameter}} & \parbox[c]{2.5cm}{ \textbf{Main model}  }  & \parbox[c]{4.5cm}{ \textbf{Unit of measure}  } \\ [0.5ex] 
\textbf{Parameter} & \textbf{Mean ($\bar{\theta}_B$)} & \makecell{\textbf{Standard} \\\textbf{deviation ($s^\theta_B$)}} 
& \textbf{BIAS} & \makecell{\textbf{Empirical 95\%} \\\textbf{Confidence Interval}} & \makecell{\textbf{Units}\\\textbf{of measure}} \\[0.5ex]
 \hline
   $y_L$ & 5.5442$\times 10^3$  &  3.2132 $\times 10^3$ &  -1.3690 $\times 10^3$ &  [0.9997,10.7084]$\times 10^3$ & - \\ [0.5ex] %%2.9289
 \hline
   $y_H$  &  52.0074$\times 10^3$  &  4.6521 $\times 10^3$ &  0.5425 $\times 10^3$ &  [42.0513,59.7910]$\times 10^3$  & - \\ [0.5ex] %%69.6893
 \hline
 ${d}$ &   $ 2.7315 \times 10^{-8}$ & $1.1327 \times 10^{-8}$ & $ 1.7917 \times 10^{-8}$ & [0.1043 ,0.4881]$\times 10^{-7}$ & /day /cell  \\ [0.5ex] %%  0.8913 0.7905 0.7511
 \hline
${\gamma}_G$   & 2.6901 &     0.2103 & 0.1996 &  [2.2339, 2.9833] & /day 	\\ [0.5ex] %%1.9898
 \hline
 ${\gamma}_L$ &0.1802  &  0.1098 &  0.2476 & [0.0539,0.4290] &/day 	 \\ [0.5ex] %%1.9391
 \hline
 ${\alpha}_G$   & 3.6320 &  1.4511 &  -0.8820 & [1.1456,6.6794] & g/l	  \\ [0.5ex] %%0.5937
 \hline
  ${\alpha}_L$ & 3.0706 &  1.4633 &  0.5425 &  [1.1488,6.4890] & \rev{mM}	 \\ [0.5ex] %%0.5912 
 \hline
  $m$  & 1.1424 & 0.1414 & -0.1358  &  [0.9568,1.4324]  & - \\ [0.5ex]  
 \hline
  $c$ & 2.6921 & 0.8918 & -0.6924 & [ 1.1876,3.9769]  & -\\ [0.5ex] 
 \hline 
%     $G^*$  & 6.0598 &  0.7747 &  -0.2599 & [4.3033,7.5067] & \rev{mM} \\ [0.5ex]  
% \hline
   ${\beta}$ & 5.0096$\times10^{-4}$ &  1.3779$\times10^{-4}$ &  1.2896$\times10^{-4}$ &  [0.2309,0.6940]$\times10^{-3}$ &  /day	 \\ [0.5ex] 
 \hline
${\zeta}$ & 8.1224 & 3.2471 & 2.7385 & [ 2.2757,14.0522] & -  \\ [0.5ex] 
 \hline
 ${\lambda_{L}}$ & 0.0885 & 0.0112 & 9.0230$\times10^{-4}$ & [0.0657,  0.1086] & /day  \\ [0.5ex] 
% \hline
% ${\lambda_{G}}$ & 0.2786 & 0.0456 & -0.1774 & [0.1908,0.3646] & /day \\ [0.5ex] 
 \hline
  ${\kappa}_G$   & 2.9733$\times 10^{-6}$ & 4.8463$\times 10^{-7}$ & -5.1151$\times 10^{-7}$ & [0.2066 ,0.4057]$\times 10^{-5}$&  \rev{mM} /day /cell 	\\ [0.5ex] 
 \hline
  ${\kappa}_L$  & 4.5190$\times 10^{-6}$ & 6.1788$\times 10^{-7}$ & -1.8673$\times 10^{-7}$ & [0.3295,0.5476]$\times 10^{-5}$&  \rev{mM} /day /cell  \\[0.5ex]  
 \hline
  ${\eta}_L$  &  1.1982$\times 10^{-6}$ & 6.2939$\times 10^{-7}$ & -2.9643$\times 10^{-7}$ & [0.5311,2.7340]$\times 10^{-6}$&  \rev{mM} /cell	   \\ [0.5ex] 
 \hline 
\end{tabular}}
\end{table}	

%\newpage

\paragraph{\rev{Calibration results fitting data from both `glucose-deprivation' and `rescue' experiments.}}
\rev{First of all, we note that the value of $G^*$ in the OPS and the bootstrap sample distributions is close to 5.5mM, i.e. physiological levels of glucose, supporting the assumption made during calibration using `glucose-deprivation' experiments.  
The estimated value of the parameter $\lambda_G$ appears to be as large as, if not more than, the estimated value of $\lambda_L$ reported in Tab.~\ref{tab:bootstrap}, which is consistent with the analogous nature of the biological mechanisms to which these parameters are linked.
Similar conclusions to those reported above can be drawn, with the remarkable exception of the relation between values of $\gamma_G$ and $\gamma_L$. Nonetheless, the delay in the increase in cell numbers observed using the parameter set in Tab.~\ref{tab:bootstrap:ALT} suggests additional evolutionary mechanisms may be at play when glucose levels are around 20mM. Interestingly, a decrease in net proliferation rates at such large glucose concentrations has been reported in~\cite{freischel2021frequency}, which would explain the inconsistencies between the calibration carried out using data from `glucose-deprivation' experiments alone and the calibration relying on data from both `glucose-deprivation' and `rescue' experiments, and the poorer quantitative fit in Fig.~\ref{Fig:calib_ALT}.
}

\begin{table}[htb!]
\centering
%Optimal dimensional
\caption{\textbf{\rev{Bootstrap statistics~\eqref{def:mean:boot} ($J=200$) obtained for uncertainty quantification with data from both `glucose deprivation' and `rescue' experiments. }}\rev{The optimal parameter set correlates with weighted sum of squared residuals $R(\hat{S}_P)=4.8147\times 10^5$. % $R(\hat{S}_P)=2.5357\times 10^3 $. 
The units of measure of each parameter value are as reported in Tab.\ref{tab:parameters2}, values of the parameter $G^*$ are in units of mM and those of the parameter $\lambda_G$ are in units of (day)$^{-1}$.  Bootstrap sampling distributions are plotted in Fig.~\ref{Fig:bootstrap_ALT}.}}\label{tab:bootstrap:ALT} %  procedure for the model with $\Phi\not\equiv0$, $\Psi^{\pm}\not\equiv0$, using $J=200$ bootstrap samples
\footnotesize
\renewcommand{\arraystretch}{1.5}
\rev{\begin{tabular}{|c |c | c| c| c| c| c|} 
 \hline
%\parbox[c]{1.7cm}{\textbf{Parameter}} & \parbox[c]{2.5cm}{ \textbf{Main model}  }  & \parbox[c]{4.5cm}{ \textbf{Unit of measure}  } \\ [0.5ex] 
\textbf{Parameter} & \makecell{\textbf{Optimal} \\\textbf{parameter} \\\textbf{set ($\hat{\theta}_{mle}$)}} & \textbf{Mean ($\bar{\theta}_B$)} & \makecell{\textbf{Standard} \\\textbf{deviation ($s^\theta_B$)}} 
& \textbf{BIAS} & \makecell{\textbf{Empirical 95\%} \\\textbf{Confidence Interval}}  \\[0.5ex]
 \hline
   $y_L$ & 2.6711$\times 10^3$ & 2.9773$\times 10^3$  &  0.9773 $\times 10^3$ &  -0.3061 $\times 10^3$ &  [1.2317,4.7918]$\times 10^3$  \\ [0.5ex] %%2.9289
 \hline
   $y_H$  & 70.8005$\times10^3$ &  69.9811$\times 10^3$  &  2.2316 $\times 10^3$ &  0.8194 $\times 10^3$ &  [65.363,7.3986]$\times 10^3$   \\ [0.5ex] %%69.6893
 \hline
 ${d}$ & 1.2852$\times 10^{-7}$ &   $ 1.6046 \times 10^{-7}$ & $2.2275 \times 10^{-8}$ & $ -3.1935 \times 10^{-8}$ & [1.1003,1.9221]$\times 10^{-7}$  \\ [0.5ex] %%  0.8913 0.7905 0.7511
 \hline
${\gamma}_G$  & 2.0329 & 1.9956 &     0.0913 & 0.0373 &  [1.8326,2.1603] 	\\ [0.5ex] %%1.9898
 \hline
 ${\gamma}_L$ & 2.0426  &2.0147  &  0.0928 &  0.0279 & [1.8258,2.1710] 	 \\ [0.5ex] %%1.9391
 \hline
 ${\alpha}_G$ & 3.0815  & 3.8164 &  0.9890 &  -0.7349 & [2.1793,5.8222] 	  \\ [0.5ex] %%0.5937
 \hline
  ${\alpha}_L$ & 6.6278 & 6.3020 &  1.4101 &  0.3258 &  [3.5671,8.6386] 	 \\ [0.5ex] %%0.5912 
 \hline
  $m$  & 0.9150 & 1.0024 & 0.0427 & -0.0874 & [0.9206,1.0729]  \\ [0.5ex] 
 \hline
  $c$ & 0.9165 & 0.9895 & 0.0454 &  -0.0731 &  [0.9112,1.0710]   \\ [0.5ex] 
 \hline 
     $G^*$ & 5.7999  & 6.0598 &  0.7747 &  -0.2599 & [4.3033,7.5067]  \\ [0.5ex]  
 \hline
   ${\beta}$ & 3.4736$\times 10^{-4}$ & 5.7841$\times10^{-4}$ &  2.2511$\times10^{-4}$ &  -2.3095$\times10^{-4}$ &  [1.6953,9.1985]$\times10^{-4}$ 	 \\ [0.5ex] 
 \hline
${\zeta}$ & 9.7698 & 7.4438 & 3.3634 & 2.3261 & [1.6020,13.5169]   \\ [0.5ex] 
 \hline
 ${\lambda_{L}}$ & 4.5879$\times 10^{-2}$ & 0.1056 & 0.0326 & -0.0597 & [0.0477,0.1831]   \\ [0.5ex] 
 \hline
 ${\lambda_{G}}$ & 0.1012 & 0.2786 & 0.0456 & -0.1774 & [0.1908,0.3646]  \\ [0.5ex] 
 \hline
  ${\kappa}_G$ & 1.3394$\times 10^{-6}$   & 1.6594$\times 10^{-6}$ & 3.8485$\times 10^{-7}$ & -3.2002$\times 10^{-7}$ & [0.8632,2.2686]$\times 10^{-6}$	\\ [0.5ex] 
 \hline
  ${\kappa}_L$ & 2.4345$\times 10^{-6}$ & 2.1349$\times 10^{-6}$ & 8.2256$\times 10^{-7}$ & 2.9955$\times 10^{-7}$ & [0.6541,3.7297]$\times 10^{-6}$  \\[0.5ex]  
 \hline
  ${\eta}_L$ & 2.5557$\times10^{-7}$  &  5.4953$\times 10^{-7}$ & 2.1185$\times 10^{-7}$ & -2.9396$\times 10^{-7}$ & [1.4937,9.4075]$\times 10^{-7}$   \\ [0.5ex] 
 \hline 
\end{tabular}}
\end{table}	

\newpage

\newpage

\subsection{Analysis of the mathematical model}
\rev{We build on the analytical methods and results presented in~\cite{alfaro2014explicit,chisholm2016evolutionary,villa2021evolutionary}. 
}
We first characterise the qualitative and quantitative properties of the solution to the PIDE~\eqref{ePDEn} subject to the initial condition~\eqref{IC:n} (cf. Proposition~\ref{Prop1}) and then study its convergence to equilibrium under fixed concentrations of glucose and lactate (cf. Theorem~\ref{Theo1}).

%The equilibrium solution presented in Theorem~\ref{Theo1} relies on the assumption that the glucose and lactate concentrations are fixed and given, i.e.
%\begin{equation}\label{barGL}
%G(t)\equiv \overline{G}\geq0 \quad \text{and} \quad L(t)\equiv \overline{L}\geq0\,.
%\end{equation}
%Proofs of Proposition~\ref{Prop1} and Theorem~\ref{Theo1} can be found in Sup.Mat.S3. These analytical results %of Proposition~\ref{Prop1} and Theorem~\ref{Theo1} 
%have been verified numerically under fixed glucose and lactate concentrations, see Supplementary Fig.S8.

\begin{proposition}
\label{Prop1}
Let assumptions~\eqref{eR} and \eqref{eprev} hold. Then, the PIDE~\eqref{ePDEn} subject to the initial condition~\eqref{IC:n} admits the exact solution
%\begin{equation}\label{gaussian}
%n_r(t,x) = \red{\dfrac{\rho_r(t)}{y^H-y^L}} \, \sqrt{\frac{v_r(t)}{2 \pi}} \, \exp\left[-\frac{v_r(t)}{2} \left(x-\mu_r(t) \right)^2\right], 
%\end{equation}
\begin{equation}\label{gaussian}
n_r(t,x) = \frac{\rho_r(t)}{\sqrt{2 \pi \, \sigma^2_r(t)}} \, \exp\left[-\frac{\left(x-\mu_r(t) \right)^2}{2 \, \sigma^2_r(t)}\right], 
\end{equation}
with $\rho_r(t)$, $\mu_r(t)$ and $v_r(t) = 1/\sigma^{2}_r(t)$ being the components of the solution to the following Cauchy problem
\begin{equation}
\label{eq:rhomuv}
\left\{
\begin{array}{ll}
\displaystyle{\frac{{\rm d}v_r}{{\rm d}t}= 2 \left(b(G,L)-\Phi(G,L) v^2_r\right), }\\\\
\displaystyle{\frac{{\rm d}\mu_r}{{\rm d}t} = \frac{2 \, b(G,L)}{v_r} \left(X(G,L)-\mu_r\right) + \Psi(G,L,\mu_r), }\\\\
\displaystyle{\frac{{\rm d}\rho_r}{{\rm d}t} = \left[\left(a(G,L)-\frac{b(G,L)}{v_r}-b(G,L) \left(X(G,L) - \mu_r \right)^2 \right)-d \rho_r \right]\rho_r, }\\\\
v_r(0) = 1/\sigma^{2}_{r0}, \quad \mu_r(0) = \mu_{r0}, \quad \rho_r(0) = \rho_{r0}, 
\end{array}
\right.
\quad 
t \in (0,\infty).
\end{equation} 
%where $G$ and $L$ are solutions to the ODE Eqs.\eqref{eG} and~\eqref{eL}.
%where
%\begin{equation} \label{Fit}
%F(S,C,v,\mu) =  a(S,C)-\frac{b(S,C)}{v}-b(S,C) \left(\mu-h(S,C) \right)^2 \,.
%\end{equation}
\end{proposition}

%\subsection{Proof of Proposition B.1 in Appendix B} 
\begin{proof}
In the remainder of the proof we use the abridged notation
\begin{equation*}
a\equiv a(G,L)\,, \;\; b\equiv b(G,L)\,, \;\; X\equiv X(G,L)\,, \;\; \Phi\equiv \Phi(G,L)\,, \;\; \Psi\equiv \Psi(G,L,\mu_r)
\end{equation*}
and drop the subscripts $r$ for brevity.

Substituting the definitions given by Eqs.\eqref{eR} and \eqref{eprev} into the PIDE~\eqref{ePDEn} yields
\begin{equation}
\label{e1-multi}
\frac{\partial n}{\partial t} = \Phi \,\frac{\partial^2 n}{\partial x^2} \; - \Psi \frac{\partial n}{\partial x} \; + \; \left[a - b \, (x -  X)^2 - d \ \rho(t) \right] n, \quad x \in \mathbb{R}.
\end{equation} 
Building upon the results presented in~\cite{chisholm2016evolutionary,villa2021evolutionary}, we make the ansatz~\eqref{gaussian}. %\eqref{gaussian}. 
Substituting this ansatz into Eq.~\eqref{e1-multi} and introducing the notation $v(t) = 1/\sigma^{2}(t)$ we find
\begin{equation}
\label{eAnsatzTest}
\begin{split}
\frac{1}{\rho}\frac{{\rm d}\rho}{{\rm d} t}+\frac{1}{2v}\frac{{\rm d}v}{{\rm d} t}\; =\; & \frac{1}{2}\frac{{\rm d}v}{{\rm d} t}\left(x-\mu\right)^2 - \frac{{\rm d} \mu}{{\rm d}t} \, v \left(x-\mu\right) + \,  \Phi\left[v^2 \left(x-\mu\right)^2 -v \right] +  \\[5pt]
&+ \Psi v(x-\mu)\, +\, a  - b \, \left(x-X\right)^2 - d \rho.
\end{split}
\end{equation} 
Equating the second-order terms in $x$ gives the following differential equation for $v$
\begin{equation}
\label{odevi}
\frac{{\rm d}v}{{\rm d} t} + 2 \Phi v^2= 2 \, b.
\end{equation} 
Equating the coefficients of the first-order terms in $x$, and eliminating $\dfrac{{\rm d}v}{{\rm d} t}$ from the resulting equation, yields
\begin{equation}
\label{odemui}
\frac{{\rm d}\mu}{{\rm d} t} = \frac{2  b (X - \mu)}{v} + \Psi.
\end{equation} 
Choosing $x=\mu$ in Eq.~\eqref{eAnsatzTest}, and eliminating $\dfrac{{\rm d}v}{{\rm d} t}$ from the resulting equation, we obtain
\begin{equation} 
\label{oderhoi}
\frac{{\rm d}\rho}{{\rm d} t} = \left[\left(a - \frac{b}{v} - b \left(X - \mu \right)^2 \right) - d \rho \right]\rho.
\end{equation}
Under the initial condition given by Eq.~\eqref{IC:n}, we have
$$
v(0) = 1/\sigma^2_0, \quad \mu(0) = \mu_0, \quad \rho(0) = \rho_0,
$$
and imposing these initial conditions for the ODEs~\eqref{odevi}-\eqref{oderhoi} yields the Cauchy problem~\eqref{eq:rhomuv}.
\qed
\end{proof}

\begin{theorem} \label{Theo1}
Let assumptions~\eqref{eR}, \eqref{eu}, \eqref{eprev}, \eqref{eX}, \eqref{epsi} and~\eqref{epsi2} hold. Let also 
\begin{equation}\label{barGL}
G(t)\equiv \overline{G}\geq0 \quad \text{and} \quad L(t)\equiv \overline{L}\geq0\,.
\end{equation}
Then, the solution of the PIDE~\eqref{ePDEn} subject to the initial condition~\eqref{IC:n} is such that 
\beq
\label{Th1i}
\rho_r(t) \longrightarrow \rho_{r \infty}(\overline{G},\overline{L}), \quad \mu_r(t) \longrightarrow \mu_{r \infty}(\overline{G},\overline{L}), \quad \sigma^2_r(t) \longrightarrow \sigma^2_{r \infty}(\overline{G},\overline{L}) \quad \text{as } t \to \infty,
\eeq
with
%%\beq
%%\label{Th1ii}
%%\rho_{\infty}(\overline{G},\overline{L}) = \max\Bigg(0,\frac{a(\overline{G},\overline{L})-\sqrt{\Phi(\overline{L})\,b(\overline{G},\overline{L})}}{d }\Bigg), \quad \mu_{\infty}(\overline{G},\overline{L}) = h(\overline{G},\overline{L}), \quad \sigma^2_{\infty}(\overline{G},\overline{L}) =\sqrt{\frac{\Phi(\overline{L})}{b(\overline{G},\overline{L})}}.
%%\eeq
\beq
\label{Th1ii}
\begin{split}
&\rho_{r \infty}(\overline{G},\overline{L}) = \max\Bigg(0,\, \frac{1}{d} \left[a(\overline{G},\overline{L})-\sqrt{\Phi(\overline{G},\overline{L})\,b(\overline{G},\overline{L})}\, - \frac{\big(\Psi^+(\overline{G},\overline{L})\big)^2}{4\Phi(\overline{G},\overline{L})}\right]\Bigg), \\[5pt]
&\mu_{r \infty}(\overline{G},\overline{L}) = X(\overline{G},\overline{L}) + \frac{\Psi^+(\overline{G},\overline{L})}{2\sqrt{\Phi(\overline{G},\overline{L})\,b(\overline{G},\overline{L})}}, \quad \sigma^2_{r \infty}(\overline{G},\overline{L}) =\sqrt{\frac{\Phi(\overline{G},\overline{L})}{b(\overline{G},\overline{L})}}.
\end{split}
\eeq
\end{theorem}

%\subsection{Proof of Theorem B.2 in Appendix B }
% (Section~\ref{P2:Ch2:analytic:Dn0}
\begin{proof}
Proposition~\ref{Prop1} ensures that for any $t\in[0,\infty)$ the solution of the PIDE~\eqref{ePDEn} subject to the initial condition~\eqref{IC:n} is of the Gaussian form~\eqref{gaussian}. Building on the method of proof presented in~\cite{chisholm2016evolutionary,villa2021evolutionary}, we thus prove Theorem~\ref{Theo1} 
by studying the asymptotic behaviour of the components of the solution to the Cauchy problem~\eqref{eq:rhomuv} for $t \to \infty$ under the additional assumption~\eqref{barGL}. In the remainder of the proof use the abridged notation
\begin{equation*}
a\equiv a(\overline{G},\overline{L})\,, \;\; b\equiv b(\overline{G},\overline{L})\,, \;\; X\equiv X(\overline{G},\overline{L})\,, \;\; \Phi\equiv \Phi(\overline{G},\overline{L})\,, \;\; \Psi^+\equiv \Psi^+(\overline{G},\overline{L})\,, \;\; \Psi^-\equiv \Psi^-(\overline{G},\mu_r)
\end{equation*}
and drop the subscript $r$ for brevity.

\paragraph{Asymptotic behaviour of $v(t) = 1/\sigma^2(t)$ for $t \to \infty$.} Solving the ODE~\eqref{eq:rhomuv}$_1$ subject to the initial condition $v(0)=v_0$ gives
\begin{equation}
\label{vi}
v(t) = \sqrt{\frac{b}{\Phi}} \, \frac{\sqrt{\dfrac{b}{\Phi}} + v_0 - \left(\sqrt{\dfrac{b}{\Phi}} - v_0\right)\exp\left(-4\sqrt{b \, \Phi} \, t\right)}
{\sqrt{\dfrac{b}{\Phi}}+ v_0 + \left(\sqrt{\dfrac{b}{\Phi}} - v_0\right)\exp\left(-4\sqrt{b \Phi} \, t\right)},
\end{equation}  
which implies that
\begin{equation}
\label{viinf}
v(t) \longrightarrow \sqrt{\frac{b}{\Phi}} \quad \text{exponentially fast} \text{ as } t \to \infty.
\end{equation} 

\paragraph{Asymptotic behaviour of $\mu(t)$ for $t \to \infty$.} Solving the ODE~\eqref{eq:rhomuv}$_2$ subject to the initial condition $\mu(0)=\mu_0$ with the integrating factor method yields 
\begin{equation}
\label{mui}
\begin{split}
\mu(t) =\,&\, h + (\mu_0 -X)\exp\left[ -\int_0^t \left(\frac{2b}{v(z)}+\Psi^-\right){\rm d}z \right] + \\[5pt]
&+ \left(\Psi^+-X\Psi^-\right)\left\{ \int_0^t \exp \left[ \int_0^z \left(\frac{2b}{v(\tau)}+\Psi^-\right) {\rm d}\tau \right]{\rm d}z \,\right\}  \exp\left[ -\int_0^t \left(\frac{2b}{v(z)}+\Psi^-\right){\rm d}z \right]\,.
\end{split}
\end{equation} 
We compute the integrals in Eq.\eqref{mui} using the solution of the ODE~\eqref{eq:rhomuv}$_1$ given by Eq.\eqref{vi}. Introducing the notation
\begin{equation}
\label{del}
\delta = \frac{\sqrt{b/\Phi} - v_0}{\sqrt{b/\Phi} + v_0}\,,
\end{equation}
we obtain
\begin{equation*}
\begin{split}
\mu(t) =\,&\, X + \frac{(1-\delta)( \mu_0 -X )}{\exp\left(2\sqrt{b \Phi}\, t\right) -\delta \exp\left(-2\sqrt{b \Phi} \,t\right) + (1-\delta)\Psi^-t} \, + \\[5pt]
&+ \frac{(\Psi^+-X\Psi^-)}{2\sqrt{b \Phi}} \, \frac{\left[ \exp\left(2\sqrt{b \Phi} \,t\right) +\delta \exp\left(-2\sqrt{b \Phi} \,t\right) - (1+\delta) + (1-\delta)\Psi^-\sqrt{b \Phi}\,t^2\right]}{\left[ \exp\left(2\sqrt{b \Phi} \,t\right) -\delta \exp\left(-2\sqrt{b \Phi} \,t\right) + (1-\delta)\Psi^-t\right]}\,.
\end{split}
\end{equation*} 
Since, under assumptions~\eqref{eu}, \eqref{eX}, \eqref{epsi} and~\eqref{epsi2}, we have $X \, \Psi^- \equiv X(\overline{G},\overline{L})\,\Psi^-(\overline{G},\mu)=0$ for any $\overline{G} \geq 0$, the latter expression of $\mu(t)$ allows us to conclude that
\begin{equation}
\label{muiinf}
\mu(t)  \longrightarrow X + \frac{\Psi^+}{2\sqrt{b \Phi}}\quad \text{exponentially fast} \text{ as } t \to \infty.
\end{equation}

\paragraph{Asymptotic behaviour of $\rho(t)$ for $t \to \infty$.} We define
$$
w \equiv w(t) \equiv w(v(t),\mu(t),\overline{G},\overline{L}) = \left(\sqrt{b \Phi } - \frac{b}{v} \right) - b \left(\mu - X - \frac{\Psi^+}{2\sqrt{ b\Phi}} \right)^2
$$
and rewrite the ODE~\eqref{eq:rhomuv}$_3$ as
\begin{equation}
\label{oderhoLs0}
\frac{{\rm d}\rho}{{\rm d} t}  = \left[\left(a +\frac{(\Psi^+)^2}{4\Phi} - \sqrt{b\Phi  } -\Psi^+\sqrt{\frac{b}{\Phi}}(\mu-X) + w\right) - d \rho\right] \rho.
\end{equation}
Solving Eq.\eqref{oderhoLs0} subject to the initial condition $\rho(0)=\rho_0$ yields
\begin{equation}
\label{rhoLs0}
\rho(t) = \frac{\displaystyle \rho_0 \exp\left[\left(a +\frac{(\Psi^+)^2}{4\Phi} - \sqrt{b \, \Phi}\right) t - \Psi^+\sqrt{\frac{b}{\Phi}}\int_0^t (\mu(z)-X){\rm d}z + \int_0^t  w(z) \, {\rm d}z\right]}{\displaystyle 1 + d\, \rho_0 \int_0^t \exp\left[\left(a +\frac{(\Psi^+)^2}{4\Phi} - \sqrt{b \, \Phi}\right) z - \Psi^+\sqrt{\frac{b}{\Phi}}\int_0^z (\mu(\tau)-X){\rm d}\tau+ \int_0^z w(\tau) \, {\rm d}\tau\right]{\rm d}z}.
\end{equation}
The asymptotic results~\eqref{viinf} and \eqref{muiinf} ensure that
\begin{equation}
\label{newresetatozero}
w(t)  \longrightarrow 0 \quad \text{exponentially fast} \text{ as } t \to \infty \,.
\end{equation}
Furthermore, the asymptotic results~\eqref{muiinf} and~\eqref{newresetatozero} imply that in the asymptotic regime $t \to \infty$ we have
%$$
%\exp\left[\left(a +\frac{(\Psi^+)^2}{4\Phi} - \sqrt{b \, \Phi}\right) t - \Psi^+\sqrt{\frac{b}{\Phi}}\int_0^t (\mu(z)-h){\rm d}z + \int_0^t  w(z) \, {\rm d}z\right] \sim A(\overline{G},\overline{L}) \, \exp\left[\left(a - \sqrt{b \, \Phi} -\frac{(\Psi^+)^2}{4\Phi} \right) t \right]\,.
%$$
\begin{equation}\label{as1}
\begin{split}
&\exp\left[\left(a +\frac{(\Psi^+)^2}{4\Phi} - \sqrt{b \, \Phi}\right) t - \Psi^+\sqrt{\frac{b}{\Phi}}\int_0^t (\mu(z)-X){\rm d}z + \int_0^t  w(z) \, {\rm d}z\right]\\[5pt]
& \hspace{8cm} \sim A\, \exp\left[\left(a - \sqrt{b \, \Phi} -\frac{(\Psi^+)^2}{4\Phi} \right) t \right]\,,
\end{split}
\end{equation}
for some positive constant factor $A\equiv A(\overline{G},\overline{L})$. Therefore, Eq.~\eqref{rhoLs0} allows us to conclude that
\begin{equation}
\label{rhto0}
\text{if } \quad \sqrt{b\,\Phi}\,+\frac{\left(\Psi^+\right)^2}{4\Phi}\geq a \quad \text{ then } \quad \rho(t)\longrightarrow 0 \;\; \text{ as } \;\;t \to \infty.
\end{equation}
On the other hand, the asymptotic results~\eqref{muiinf} and~\eqref{newresetatozero} imply that, if $\sqrt{b \, \Phi} + (\Psi^+)^2/(4\Phi)< a$, in the asymptotic regime $t \to \infty$ we also have
\begin{equation}\label{as2}
\displaystyle{
\begin{split}
&\int_0^t\exp\left[\left(a +\frac{(\Psi^+)^2}{4\Phi} - \sqrt{b \, \Phi}\right) z - \Psi^+\sqrt{\frac{b}{\Phi}}\int_0^z (\mu(\tau)-X){\rm d}\tau + \int_0^z  w(\tau) \, {\rm d}\tau\right]{\rm d}z\\[5pt]
& \hspace{8cm} \sim B\, \frac{\exp\left[\left(a - \sqrt{b \, \Phi} -\frac{(\Psi^+)^2}{4\Phi} \right) t \right]}{\left(a - \sqrt{b \, \Phi} -\frac{(\Psi^+)^2}{4\Phi} \right)}\,,
\end{split}}
\end{equation}
for some positive constant factor $B\equiv B(\overline{G},\overline{L})$. The asymptotic relations~\eqref{as1} and~\eqref{as2}, along with Eq.~\eqref{rhoLs0}, allow us to conclude that
\begin{equation}
\label{rhtogr0}
\text{if }\quad \sqrt{b\,\Phi}\,+\frac{\left(\Psi^+\right)^2}{4\Phi}< a \quad \text{then} \quad \rho(t) \longrightarrow \frac{1}{d}\left[ a - \sqrt{b \, \Phi} - \frac{\left(\Psi^+\right)^2}{4\Phi} \right]\;\; \text{ as } \;\; t \to \infty.
\end{equation}
%\begin{equation}
%\label{rhtogr0}
%\begin{split}
%&\text{if } \sqrt{b(\overline{G},\overline{L})\,\Phi(\overline{G},\overline{L})}\,+\frac{\left(\Psi^+(\overline{G},\overline{L})\right)^2}{4\Phi(\overline{G},\overline{L})}< a(\overline{G},\overline{L}) \; \\[5pt] &\text{then } \;\rho(t,) \longrightarrow \frac{1}{d}\left[ a(\overline{G},\overline{L}) - \sqrt{b(\overline{G},\overline{L}) \, \Phi(\overline{G},\overline{L})} - \frac{\left(\Psi^+(\overline{G},\overline{L})\right)^2}{4\Phi(\overline{G},\overline{L})} \right]\; \text{ as } t \to \infty.
%\end{split}
%\end{equation}
Taken together, the asymptotic results~\eqref{rhto0} and \eqref{rhtogr0} yield
\begin{equation}
\label{rhoinf}
\rho(t)  \longrightarrow \max\left(0, \frac{1}{d}\left[a - \sqrt{b \, \Phi} - \frac{(\Psi^+)^2}{4\Phi}\right]\right) \quad \text{as } t \to \infty.
\end{equation}
Claims~\eqref{Th1i} and \eqref{Th1ii} follow from the asymptotic results~\eqref{viinf}, \eqref{muiinf} and \eqref{rhoinf}.
\end{proof}

\begin{remark} \label{Rem1}
The asymptotic results of Theorem~\ref{Theo1} along with the relations given by Eq.~\eqref{erhomusigmaxy} imply that, when $(G(t),L(t)) \equiv (\overline{G},\overline{L})$,  
\beq
\label{Th1iresc}
\rho(t) \longrightarrow \rho_{\infty}(\overline{G},\overline{L}), \quad \mu(t) \longrightarrow \mu_{\infty}(\overline{G},\overline{L}), \quad \sigma^2(t) \longrightarrow \sigma^2_{\infty}(\overline{G},\overline{L})  \quad \text{as } t \to \infty,
\eeq
where
\beq
\label{Th1iiy}
\begin{split}
&\rho_{\infty}(\overline{G},\overline{L}) = \max\Bigg(0,\, \frac{1}{d} \left[a(\overline{G},\overline{L})-\sqrt{\Phi(\overline{G},\overline{L})\,b(\overline{G},\overline{L})}\, - \frac{\big(\Psi^+(\overline{G},\overline{L})\big)^2}{4\Phi(\overline{G},\overline{L})}\right]\Bigg), \\[5pt]
&\mu_{\infty}(\overline{G},\overline{L}) = y_L + (y_H - y_L) \left[X(\overline{G},\overline{L}) + \frac{\Psi^+(\overline{G},\overline{L})}{2\sqrt{\Phi(\overline{G},\overline{L})\,b(\overline{G},\overline{L})}}\right], \\[5pt]
& \sigma^2_{\infty}(\overline{G},\overline{L}) = (y_H - y_L)^2 \sqrt{\frac{\Phi(\overline{G},\overline{L})}{b(\overline{G},\overline{L})}}.
\end{split}
\eeq
\end{remark}

\clearpage

\section{Supplementary figures}\label{sec:sipfigtab}

\begin{figure}[htb!]
\centering
\includegraphics[width=0.95\linewidth]{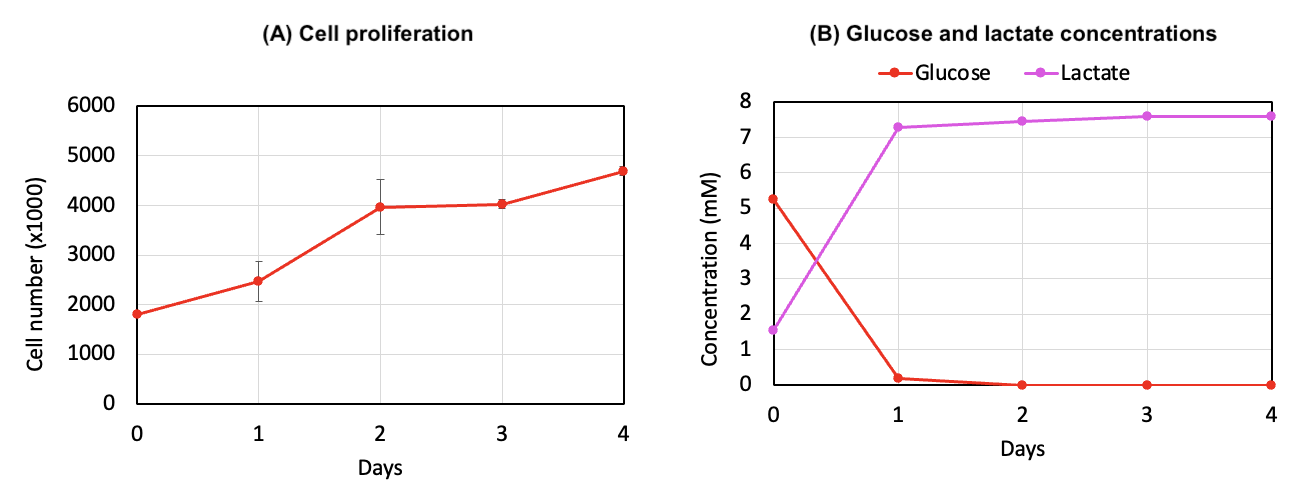}
\caption{\label{SFig:data} \textbf{Dynamics of cell proliferation and glucose and lactate concentrations in `glucose-deprivation' experiments conducted on MCF7 cells.} Dynamics of  cell proliferation (panel {\bf (A)}), glucose concentration (panel {\bf (B)}, red line, left y-axis) and lactate concentration (panel {\bf (B)}, pink line, right y-axis) in `glucose-deprivation' experiments conducted on MCF7 cells for four days. Cell proliferation was assessed by counting the number of viable cells upon seeding (i.e. day 0) and at the end of each day of culture (i.e. days 1-4). Glucose and lactate concentrations were measured in the culture medium at days 0-4. \rev{The figure in panel {\bf (A)} displays the average (dots) and standard deviation (error bars) of two replicate experiments.}}
\end{figure}
%Full details of experimental materials and methods can be found in Sec.S2.1.

\begin{figure}[htb!]
\centering
\includegraphics[width=0.5\linewidth]{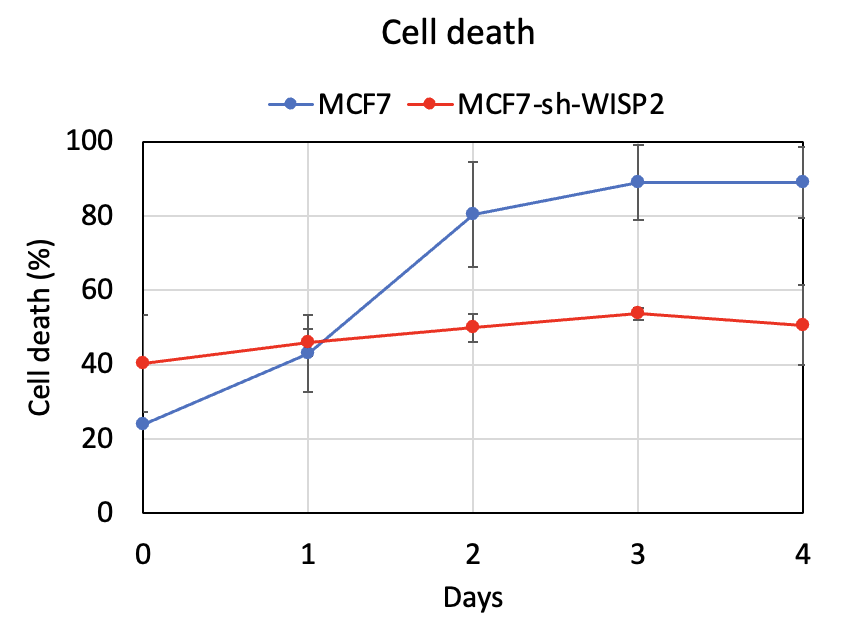}
\caption{\label{SFig:death}\textbf{Dynamics of cell death in `glucose-deprivation' experiments conducted on MCF7 and MCF7-sh-WISP2 cells.} Dynamics of cell death in `glucose-deprivation' experiments conducted on MCF7 cells (blue line) and MCF7-sh-WISP2 cells (red line) for four days. Cell death was assessed by measuring the percentage of apoptotic cells upon seeding (i.e. day 0) and at the end of each day of culture (i.e. days 1-4). \rev{This figure displays the average (dots) and standard deviation (error bars) of two replicate experiments.} }
\end{figure}

\begin{figure}[htb!]
\centering
\includegraphics[width=0.6\linewidth]{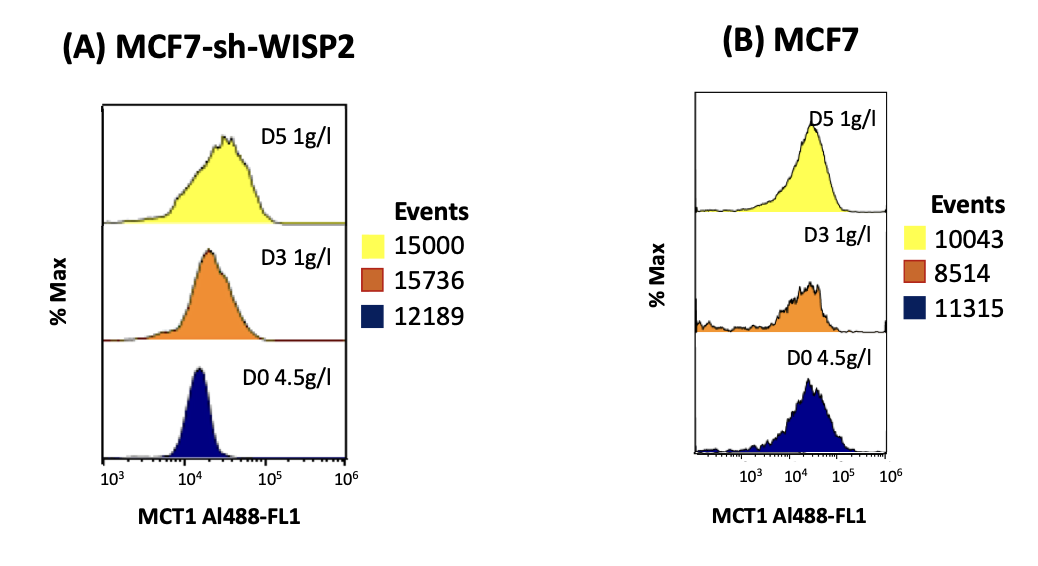}
\caption{\label{Fig:cytometryMCF7} \textbf{Dynamics of MCT1 expression in `glucose-deprivation' experiments conducted on MCF7-sh-WISP2 and MCF7 cells.} Comparison between MCT1 protein expression of MCF7-sh-WISP2 cells (panel {\bf (A)}) and MCF7 cells (panel {\bf (B)}), assessed through flow cytometry analysis, upon seeding (i.e. on day 0) and on days 3 and 5 of `glucose-deprivation' experiments conducted for five days (sub-panel D0 and sub-panels D3 and D5). \rev{The `Events' legends indicate the number of events (i.e. the total number of cells analysed) for each distribution plotted on a logarithmic scale.}}
\end{figure}

%with 1g/l of glucose
\begin{figure}[htb!]
\centering
\includegraphics[width=1\linewidth]{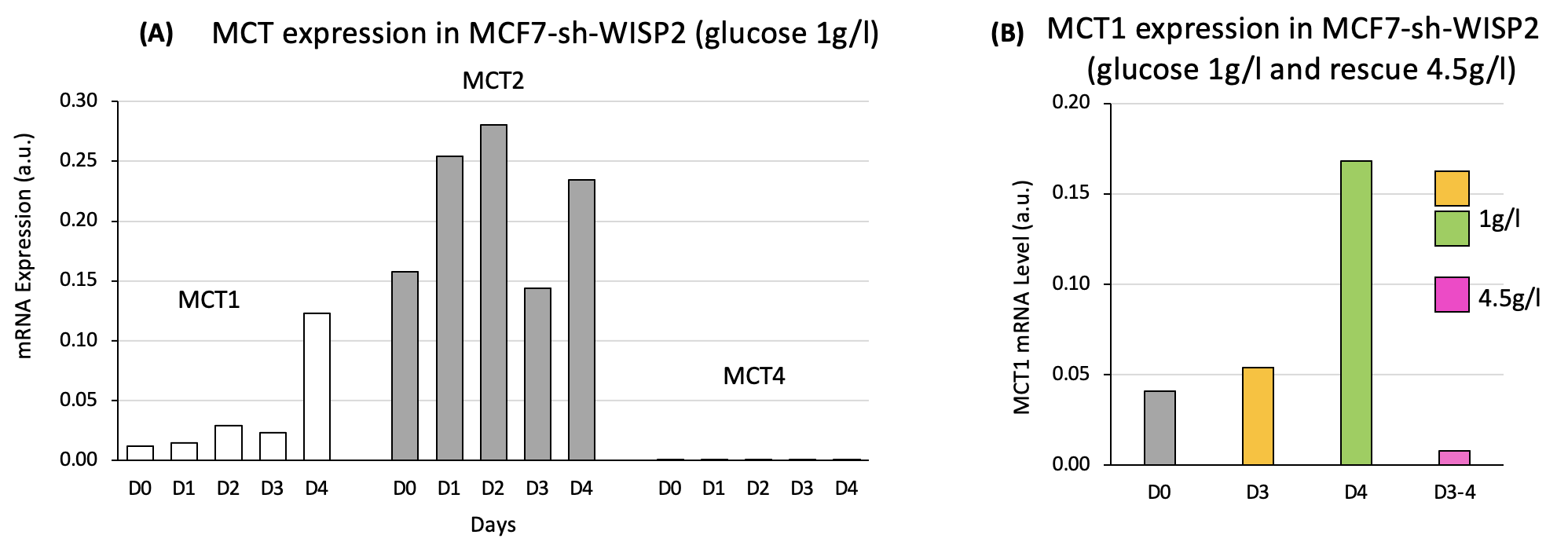}
\caption{\label{SFig:mRNA} \textbf{Dynamics of MCT expression in `glucose-deprivation' and `rescue' experiments conducted on MCF7-sh-WISP2 cells.} {\bf (A)} MCT1, MCT2 and MCT4 mRNA expression of MCF7-sh-WISP2 cells, assessed through RT-qPCR, upon seeding (i.e. on day 0) and on days 1-4 of `glucose-deprivation' experiments conducted for four days (column D0 and columns D1-D4). {\bf (B)} MCT1 mRNA expression of MCF7-sh-WISP2 cells, assessed through RT-qPCR, upon seeding (i.e. on day 0) and on days 3 and 4 of `glucose-deprivation' experiments conducted for four days (column D0 and columns D3 and D4). MCT1 mRNA expression during the phase of rescue from glucose deprivation in the corresponding `rescue' experiments (i.e. on days 3 and 4) is also displayed (column D3-4). The mRNA levels in the plots indicate the abundance of the target gene relative to that of endogenous control Actin used to normalise the initial quantity and purity of total RNA.}
%Full details of experimental materials and methods can be found in Sec.S2.1.
%MCT1 expression is here normalised to $\beta$actin. 
%\caption{\label{SFig:mRNA}\textbf{MCT expression in the MCF7-sh-WISP2 cell line.} (A) MCF7-sh-WISP2 cells were cultured for four days in a medium initially containing 1g/l of glucose and the expression of MCT1, MCT2 and MCT4 mRNA was measured through  (D0-D4 data). (B) Cells were seeded at 4.5g/l of glucose for 24h (D0 data, grey) and then the medium was changed to a medium initially containing 1g/l of glucose for three (D3 data, yellow) or four days (D4 data, green). In the rescue experiments, the medium was changed to a medium initially containing 4.5g/l of glucose (rescue D3 and D4, pink). At each time point, MCT1 mRNA was measured through RT-qPCR. }
\end{figure}

\begin{figure}[htb!]
\centering
\includegraphics[width=\linewidth]{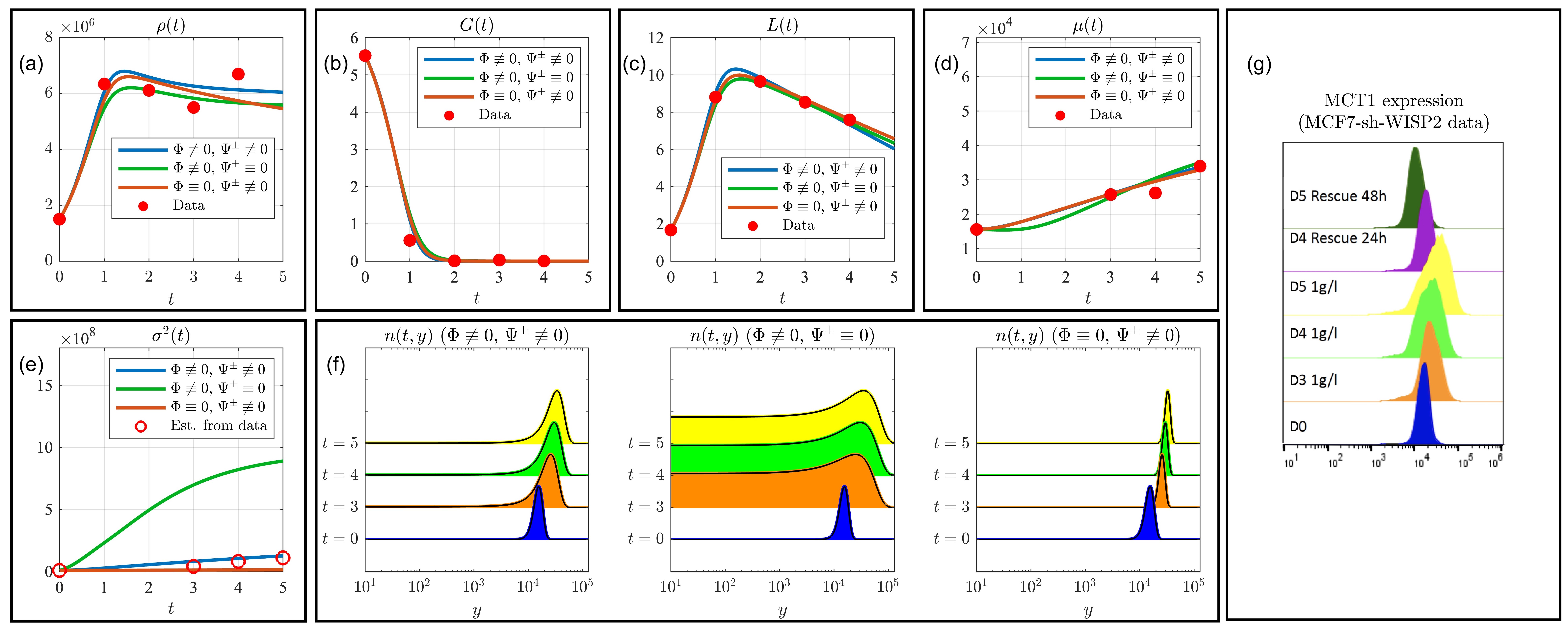}
\caption{\label{Fig:calib} \textbf{Numerical simulations of `glucose-deprivation' experiments conducted on MCF7-sh-WISP2 cells.} 
Simulated dynamics of the cell number $\rho(t)$ (panel {\bf (a)}), the glucose concentration $G(t)$ (panel {\bf (b)}), the lactate concentration $L(t)$ (panel {\bf (c)}), the mean level of MCT1 expression $\mu(t)$ (panel {\bf (d)}), the related variance $\sigma^2(t)$ (panel {\bf (e)}), and the MCT1 expression distribution $n(t,y)$ (panel {\bf (f)}, $t=0$ - $t=5$) in `glucose-deprivation' experiments conducted on MCF7-sh-WISP2 cells. Numerical simulations were carried out for the calibrated model in which both \rev{SPCs and FECs} in MCT1 expression are included (i.e. $\Phi\not\equiv0$, $\Psi^{\pm}\not\equiv0$) and for calibrated reduced models in which only \rev{FECs} in MCT1 expression are included (i.e. $\Phi\not\equiv0$, $\Psi^{\pm}\equiv0$) or only \rev{SPCs} in MCT1 expression are included (i.e. $\Phi\equiv0$, $\Psi^{\pm}\not\equiv0$),  under the OPS reported in Tab.S1. The MCT1 expression distribution is plotted on a logarithmic scale as for the outputs of flow cytometry analyses (panel {\bf (g)}) to facilitate visual comparison. 
%MCT1 expression distributions during the phase of rescue from glucose deprivation in the corresponding simulations of `rescue' experiments are also displayed (panel {\bf (f)}, $t=4 \, {\rm R}$ and $t=5 \, {\rm R}$) along with the mean levels of MCT1 expression (panel {\bf (e)}, dashed lines). 
The red markers highlight experimental data that are used to carry out model calibration. % with circles and triangles corresponding to `glucose-deprivation' and `rescue' experiments, respectively. 
The values of $t$ are in days, while the values of \rev{$G(t)$ and $L(t)$ are in mM}.}
\end{figure}

\begin{figure}[htb!]
\centering
\includegraphics[width=\linewidth]{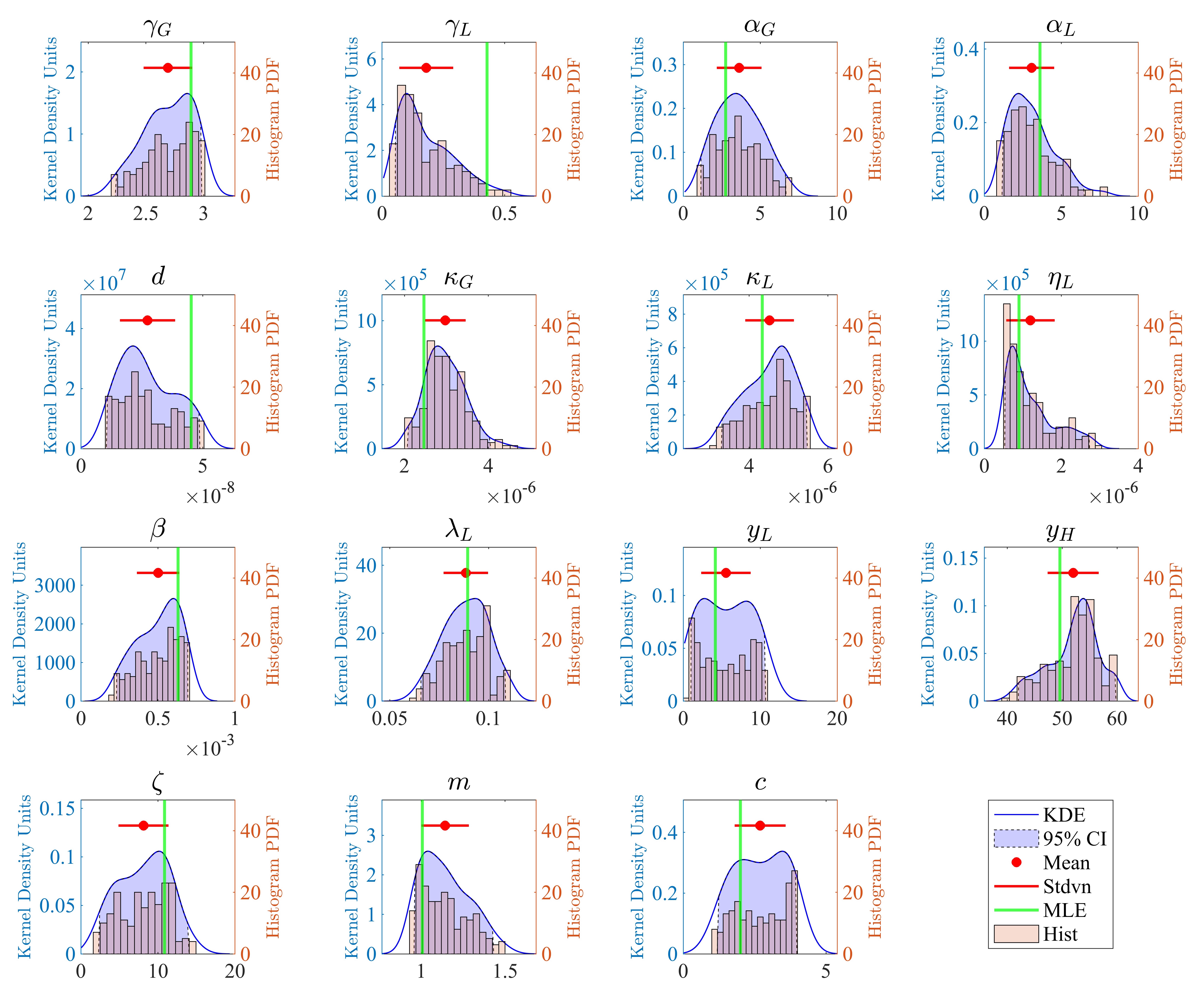}
\caption{\label{Fig:bootstrap} \rev{\textbf{Bootstrap sampling distributions of the model parameters obtained by fitting data from `glucose-deprivation' experiments.} Parameter distributions obtained through the bootstrapping algorithm described in Sec.\ref{sec:methods:bootstrapping} for the model in which both FECs and SPCs in MCT1 expression are included (i.e. $\Phi\not\equiv0$, $\Psi^{\pm}\not\equiv0$), setting $G^*>5.5$mM (and thus ignoring $\lambda_G$ as irrelevant for predicting dynamics under glucose deprivation), generating $J=200$ bootstrap samples.  For each parameter the following statistics are displayed: the probability density function (PDF) of the samples (orange histogram); the kernel density estimation (KDE), i.e. the smooth PDF obtained from the bootstrap samples by applying the {\sc Matlab} function  \texttt{ksdensity} (blue line); the empirical 95\% confidence interval (blue area); the bootstrap mean (red dot) and standard deviation (red line); the parameter value in the optimal parameter set (OPS) listed in the second column of Tab.S1 (green line), for comparison.   }}
\end{figure}

\begin{figure}[htb!]
\centering
\includegraphics[width=0.4\linewidth]{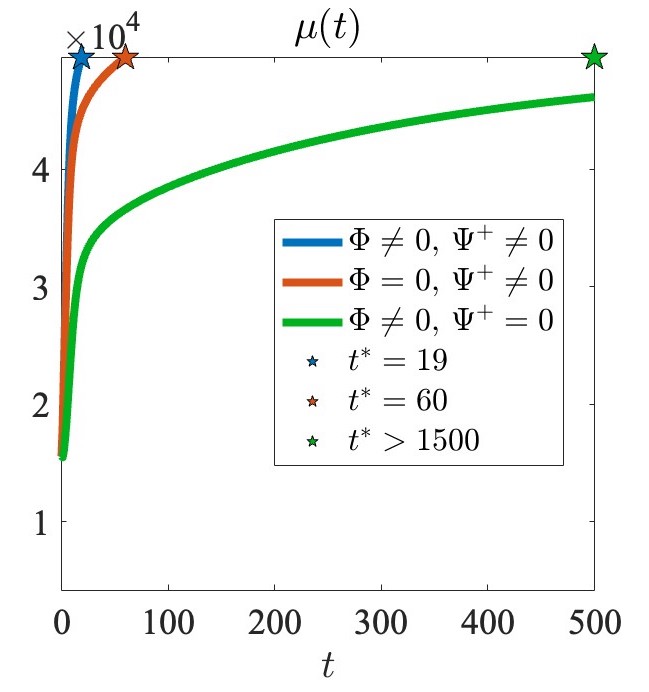}
\caption{\label{Fig:mu1}\textbf{Numerical simulations of long-term dynamics of the mean level of MCT1 expression of glucose-deprived MCF7-sh-WISP2 cells.} Long-term dynamics of the mean level of MCT1 expression of MCF7-sh-WISP2 cells $\mu(t)$ in `glucose-deprivation' experiments simulated through the calibrated model in which both  \rev{FECs and SPCs}  in MCT1 expression are included (i.e. $\Phi\not\equiv0$, $\Psi^{\pm}\not\equiv0$) and through calibrated reduced models in which only \rev{SPCs} in MCT1 expression are included (i.e. $\Phi\equiv0$, $\Psi^{\pm}\not\equiv0$) or only \rev{FECs} in MCT1 expression are included (i.e. $\Phi\not\equiv0$, $\Psi^{\pm}\equiv0$), under the OPS reported in Tab.S1. Dynamics are shown for $t\in[0,t^*]$ (in days), with $t^*$ being the first time instant when the mean level of MCT1 expression attains the value $y_H$,  which in our modelling framework is the level endowing MCF7-sh-WISP2 cells with the maximum capability of taking lactate from the extracellular environment and reusing it to produce the energy required for their proliferation under glucose deprivation. 
The value of $t^*$ is marked by a star, i.e. $t^*=72$ for the model with $\Phi\not\equiv0$ and $\Psi^{\pm}\not\equiv0$, $t^*=325$ for the model with $\Phi\equiv0$ and $\Psi^{\pm}\not\equiv0$, and $t^*>2000$ for the model with $\Phi\not\equiv0$ and $\Psi^{\pm}\equiv0$.
}
\end{figure}

\begin{figure}[htb!]
\centering
\includegraphics[width=1\linewidth]{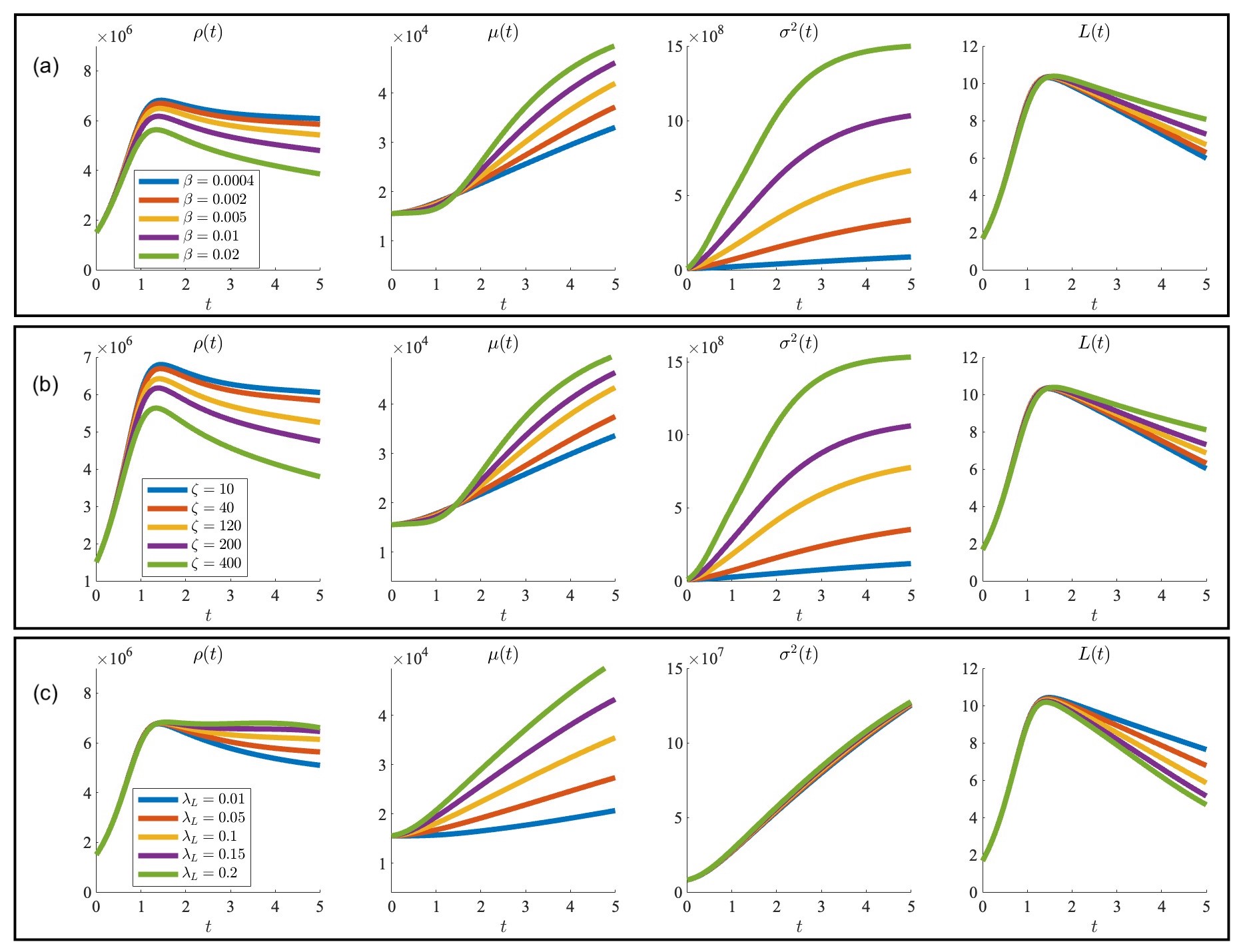}
\caption{\label{SFig:beta}\textbf{Additional numerical simulations of `glucose-deprivation' experiments conducted on MCF7-sh-WISP2 cells.} 
Simulated dynamics of the cell number $\rho(t)$ (first column), the mean level of MCT1 expression $\mu(t)$ (second column), the related variance $\sigma^2(t)$ (third column), and the lactate concentration $L(t)$ (fourth column) in `glucose-deprivation' experiments conducted on MCF7-sh-WISP2 cells. Numerical simulations were carried out for the calibrated model in which both  \rev{FECs and SPCs}  in MCT1 expression are included (i.e. $\Phi\not\equiv0$, $\Psi^{\pm}\not\equiv0$), under the OPS reported in Tab.S1 but for different values of the parameter $\beta$ (panel {\bf (a)}) or different values of the parameter $\zeta$ (panel {\bf (b)}), which correspond to different values of the rate of \rev{FECs} in MCT1 expression $\Phi$ (cf. the definition given by Eq.\eqref{ephi}), or different values of the parameter $\lambda_L$ (panel {\bf (c)}), which correspond to different values of the rate at which \rev{SPCs} lead to an increase in MCT1 expression $\Psi^+$ (cf. the definition given by Eq.\eqref{epsi2}). In particular: in panel {\bf (a)}, $\beta=0.0004$ (blue lines), $\beta=0.002$ (orange lines), $\beta=0.005$ (yellow lines), $\beta=0.01$ (purple lines), and $\beta=0.02$ (green lines); in panel {\bf (b)}, $\zeta=7.9143$ (blue lines), $\zeta=40$ (orange lines), $\zeta=120$ (yellow lines), $\zeta=200$ (purple lines), and $\zeta=400$ (green lines); in panel {\bf (c)}, $\lambda_L=0.0693$ (blue lines), $\lambda_L=0.1$ (orange lines), $\lambda_L=0.18$ (yellow lines), $\lambda_L=0.25$ (purple lines), and $\lambda_L=0.32$ (green lines). The values of $t$ are in days, while the values of $L(t)$ are in mM.}
\end{figure}

\begin{figure}[ht!]
\centering
\includegraphics[width=1\linewidth]{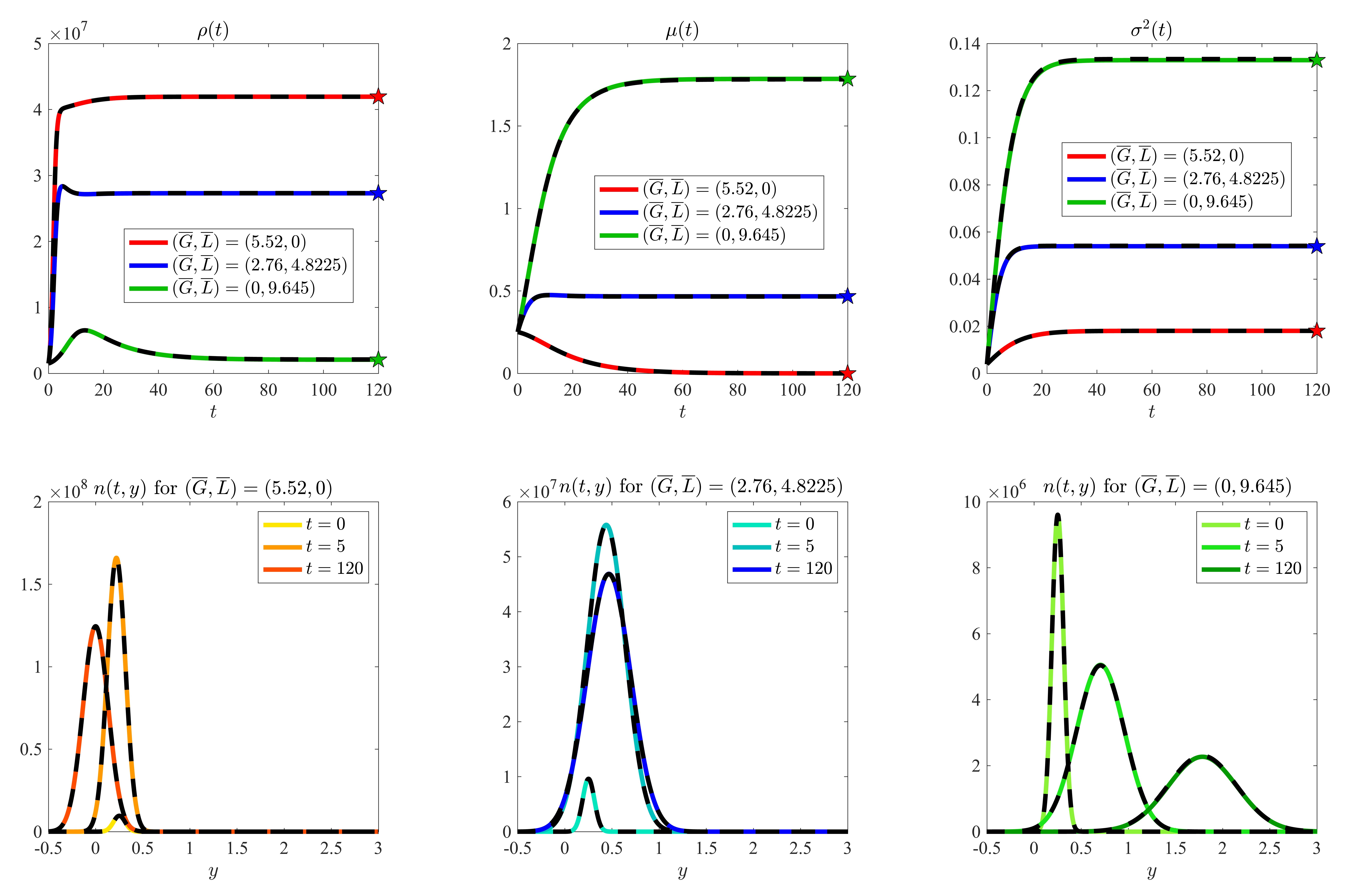}
\caption{\label{SFig:analysis}\textbf{Numerical simulations of long-term dynamics of glucose-deprived MCF7-sh-WISP2 cells under constant concentrations of glucose and lactate.} 
{\bf Top row.} Simulated dynamics of the cell number $\rho_r(t)$ (first column), the rescaled mean level of MCT1 expression $\mu_r(t)$ %=\dfrac{\mu(t)-y_L}{y_H - y_L}
 (second column), and the related variance $\sigma^2_r(t)$ (third column) under constant concentrations of glucose and lactate, i.e. $(G(t), L(t)) \equiv (\overline{G}, \overline{L})$ with $(\overline{G},\overline{L})=(5.52,0)$ (red lines), $(\overline{G},\overline{L})=(2.76,4.8225)$ (blue lines) and $(\overline{G},\overline{L})=(0,9.645)$ (green lines). Numerical simulations were carried out for the calibrated model in which both  \rev{FECs and SPCs}  in MCT1 expression are included (i.e. $\Phi\not\equiv0$, $\Psi^{\pm}\not\equiv0$), under the OPS for cell dynamics reported in Tab.S1. The black, dashed lines highlight the dynamics of the same quantities obtained by solving numerically the Cauchy problem~\eqref{eq:rhomuv} complemented with Eq.~\eqref{IC:pars} and with $(G(t), L(t)) \equiv (\overline{G}, \overline{L})$, while the coloured stars mark the analytical equilibrium values computed via Eq.\eqref{Th1ii}. \textbf{Bottom row.} Corresponding dynamics of the rescaled MCT1 expression distribution $n_r(t,x)$ for $(\overline{G},\overline{L})=(5.52,0)$ (left panel), $(\overline{G},\overline{L})=(2.76,4.8225)$ (central panel) and $(\overline{G},\overline{L})=(0,9.645)$ (right panel). Coloured, solid lines refer to different times $t$ and the black, dashed lines highlight the rescaled MCT1 expression distribution given by Eq.~\eqref{gaussian} whereby $\rho_r(t)$, $\mu_r(t)$ and $\sigma^2_r(t)$ are obtained by solving numerically the Cauchy problem~\eqref{eq:rhomuv} complemented with Eq.~\eqref{IC:pars} and with $(G(t), L(t)) \equiv (\overline{G}, \overline{L})$. {The values of $t$ are in days, while the values of $\overline{G}$ and $\overline{L}$ are in mM}}
\end{figure}

\clearpage

\begin{figure}[htb!]
\centering
\includegraphics[width=\linewidth]{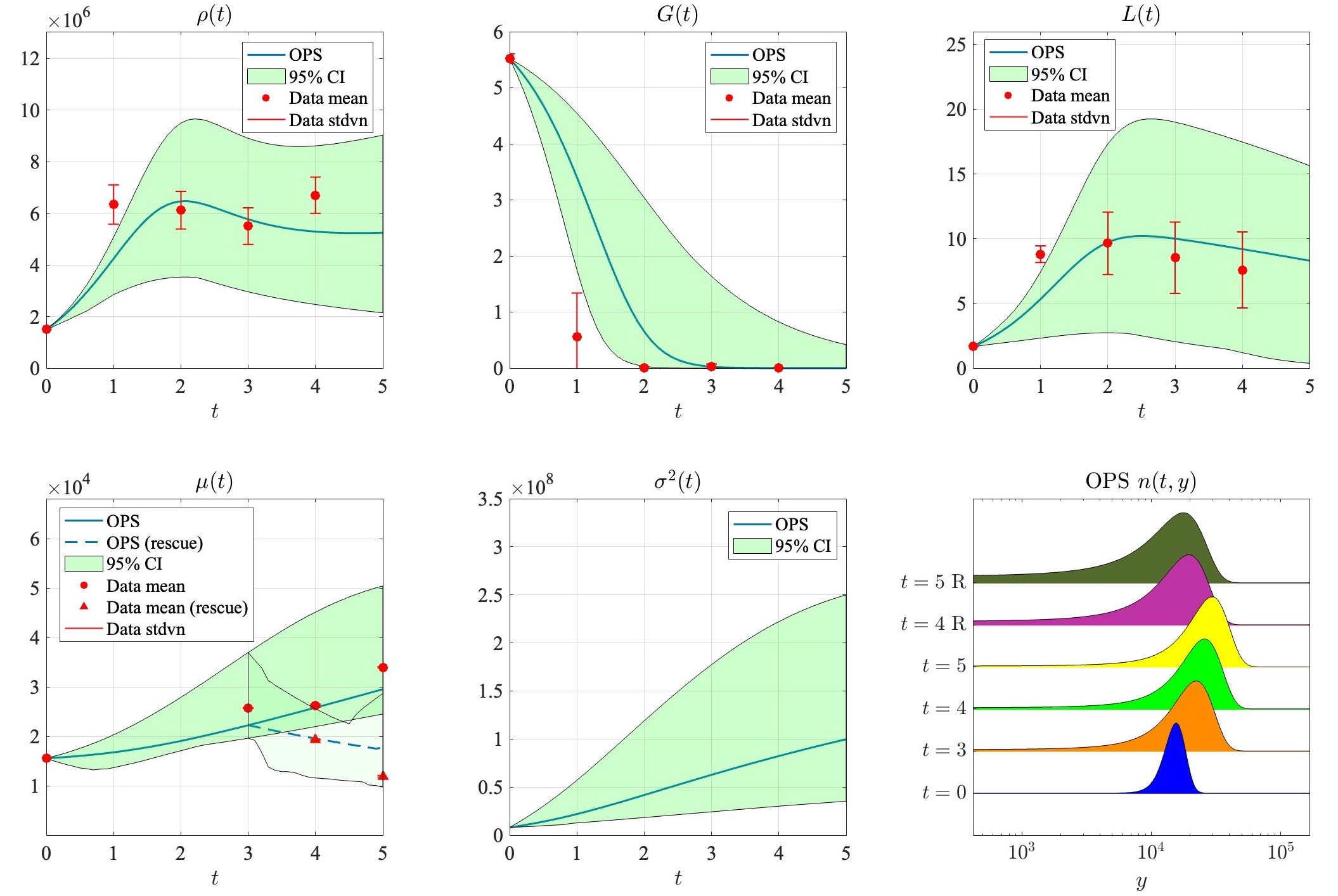}
\caption{\label{Fig:calib_ALT} \textbf{Numerical simulations of `glucose-deprivation' and `rescue' experiments conducted on MCF7-sh-WISP2 cells.}
Simulated dynamics of the cell number $\rho(t)$ (top-left panel), the glucose concentration $G(t)$ (top-central panel), the lactate concentration $L(t)$ (top-right panel), the mean level of MCT1 expression $\mu(t)$ (bottom-left panel, solid line), the related variance $\sigma^2(t)$ (bottom-central panel), and the MCT1 expression distribution $n(t,y)$ (bottom-right panel, $t=0$ - $t=5$) in `glucose-deprivation' experiments conducted on MCF7-sh-WISP2 cells. Numerical simulations were carried out for the calibrated model in which both  \rev{FECs and SPCs}  in MCT1 expression are included (i.e. $\Phi\not\equiv0$ and $\Psi^{\pm}\not\equiv0$), under the OPS reported in Tab.\ref{tab:bootstrap:ALT} \rev{(blue line), and under 200 parameter sets generated by random sampling from the empirical 95\% confidence interval (CI) of the bootstrap sampling distributions (green area) -- see Fig.~\ref{Fig:bootstrap_ALT}.} The MCT1 expression distribution \rev{obtained under the OPS} is plotted on a logarithmic scale as for the outputs of flow cytometry analyses to facilitate visual comparison. The MCT1 expression distribution during the phase of rescue from glucose deprivation in the corresponding simulations of `rescue' experiments is also displayed (bottom-right panel, $t=4 \, {\rm R}$ and $t=5 \, {\rm R}$) along with the mean level of MCT1 expression (bottom-left panel, dashed \rev{blue} line \rev{and light green area}). The red markers highlight \rev{average (scatter points) and standard deviation (error bars) of the} experimental data that are used to carry out model calibration, with circles and triangles corresponding to `glucose-deprivation' and `rescue' experiments, respectively. The values of $t$ are in days, \rev{while the values of $G(t)$ and $L(t)$ are in mM}.}
\end{figure}

\begin{sidewaysfigure}[htb!]
\centering
\includegraphics[width=\linewidth]{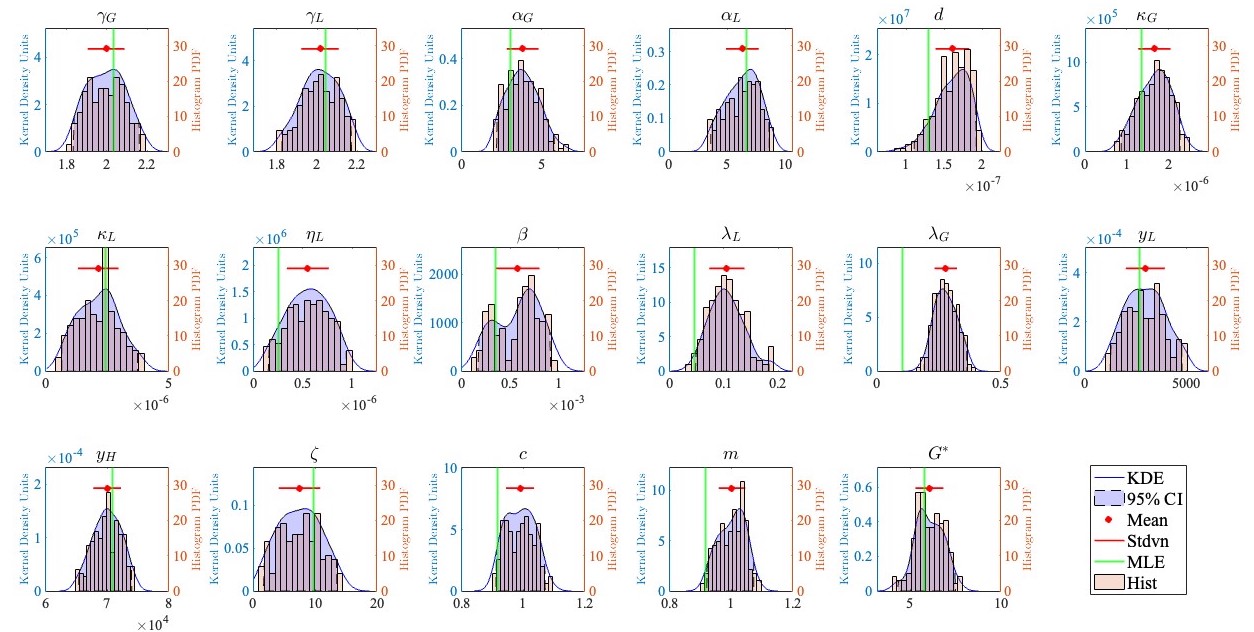}
\caption{\label{Fig:bootstrap_ALT} \rev{\textbf{Bootstrap sampling distributions of the model parameters obtained by fitting data from `glucose-deprivation' and `rescue' experiments.} Parameter distributions obtained through the bootstrapping algorithm described in Sec.\ref{sec:methods:bootstrapping} for the model in which both FECs and SPCs in MCT1 expression are included (i.e. $\Phi\not\equiv0$, $\Psi^{\pm}\not\equiv0$), generating $J=200$ bootstrap samples.  For each parameter the following statistics are displayed: the probability density function (PDF) of the samples (orange histogram); the kernel density estimation (KDE), i.e. the smooth PDF obtained from the bootstrap samples by applying the {\sc Matlab} function  \texttt{ksdensity} (blue line); the empirical 95\% confidence interval (blue area); the bootstrap mean (red dot) and standard deviation (red line); the parameter value in the optimal parameter set (OPS) listed in the second column of Tab.~\ref{tab:bootstrap:ALT} (green line), for comparison.}}
\end{sidewaysfigure}

\clearpage

%\makeatletter 
%\renewcommand\@biblabel{1}{S#1} 
%\makeatother 
\bibliographystyle{siam}
\bibliography{Lactate.bib}